\documentclass[11pt]{article} %for arxiv sub 
\pdfoutput=1
\makeatletter
\renewcommand{\@biblabel}[1]{[#1]\hfill}
\makeatother
\usepackage[margin=.95in]{geometry}
\usepackage[font={small,bf}]{caption}
\usepackage[letterpaper, colorlinks,linkcolor=ForestGreen,citecolor=ForestGreen,
backref, bookmarks, bookmarksopen, bookmarksnumbered]{hyperref}
\usepackage{amsmath,amssymb, amsthm}    %
\usepackage{thmtools}
\usepackage{thm-restate}
\usepackage{bm}
\usepackage{bbm}
\usepackage[algo2e]{algorithm2e} %need the option when loading both algorithm and algorithm2e
\usepackage{algorithm,algpseudocode,algorithmicx}
\usepackage{enumerate}

\newtheorem{theorem}{Theorem}
\newtheorem{definition}{Definition}
\newtheorem{lemma}{Lemma}[theorem]

\newtheorem{corollary}{Corollary}

\usepackage{verbatim}           %
\usepackage{graphicx,float}     %
\usepackage[vflt]{floatflt}     %
\usepackage[small,compact]{titlesec}
\usepackage{setspace}

\usepackage{color}
\definecolor{ForestGreen}{rgb}{0.1333,0.5451,0.1333}
\usepackage{url}
\usepackage{enumitem}
\usepackage{xfrac} 
\usepackage[toc, page]{appendix}
\usepackage{physics}
\usepackage{psfrag}
\usepackage{xifthen} 
\numberwithin{equation}{section}
\usepackage{booktabs, array, makecell}

\usepackage[noabbrev,capitalize]{cleveref} 

\newtheorem{claim}{Claim}

\newtheorem{fact}{Fact}[theorem]

\newcommand{\eat}[1]{}
\newcommand{\hide}[1]{{\color{red}Contents hidden!}}
\newcommand{\ones}{\mathbf 1}

\newcommand{\reals}{\mathbb{R}}
\newcommand{\symm}{\mathbb{S}}  % symmetric matrices
\newcommand{\diag}{\mathop{\mathbf{diag}}}
\newcommand{\Prob}{\mathop{\mathbf{Prob}}}
\newcommand{\nnz}[1]{\mathop{\mathbf{nnz}}(#1)}
\DeclareMathOperator{\E}{\mathbf{E}}
\DeclareMathOperator{\Var}{\mathbf{Var}}
\newcommand{\N}{\mathcal{N}} % for normal distribution
\DeclareMathOperator*{\argmin}{\arg\!\min}
\newcommand{\argmax}{\mathop{\mathrm{ argmax}}}
\newcommand{\dom}{\mathop{\mathbf{dom}}} % domain
\newcommand{\intr}{\mathop{\mathbf{int}}}

\newcommand{\compresslist}{ % Define a command to reduce spacing within itemize/enumerate environments, this is used right after \begin{itemize} or \begin{enumerate}
\setlength{\itemsep}{1pt}
\setlength{\parskip}{0pt}
\setlength{\parsep}{0pt}
}

\newcommand{\D}{\mathcal{D}}
\newcommand{\Ord}{\mathcal{O}}
\newcommand{\X}{\mathcal{X}}

\newcommand{\opnorm}[1]{\norm{#1}_{\textrm{op}}}
\newcommand{\nnorm}[1]{\norm{#1}_{\mathrm{nuc}}}
\newcommand{\dnorm}[1]{{\vert\kern-0.25ex\vert\kern-0.25ex\vert #1 
    \vert\kern-0.25ex\vert\kern-0.25ex\vert}}
\newcommand{\Ndnorm}[1]{{\vert\kern-0.25ex\vert\kern-0.25ex\vert #1 \vert\kern-0.25ex\vert\kern-0.25ex\vert}}

\newcommand{\zeros}{\mathbf{0}}
\newcommand{\hatD}{\widehat{D}}
\makeatletter
\def\ALG@special@indent{%
    \ifdim\ALG@thistlm=0pt\relax
        \hskip-\leftmargin
    \else
        \hskip\ALG@thistlm
    \fi
}
\newcommand{\NoIndentInAlg}[1]{\item[]\noindent\ALG@special@indent \ #1}

\newcommand{\inner}[2]{\left\langle {#1} , {#2} \right\rangle}

\providecommand{\norm}[1][]{$\left\vert\kern-0.25ex\left\vert {}\cdot{} \right\vert\kern-0.25ex\right\vert \ifx\\#1\\\else\left\vert\kern-0.25ex\left\vert#1\right\vert\kern-0.25ex\right\vert\fi$}

\providecommand{\abs}[1]{$\vert #1 \vert$}

\newcommand{\twopartdef}[4]
{
	\left\{
		\begin{array}{lll}
			#1 & \mbox{if } #2 \\
			#3 & \mbox{ } #4 
		\end{array}
	\right.
}

\newcommand\numberthis{\addtocounter{equation}{1}\tag{\theequation}}

\global\long\def\code#1{\textbf{#1}}

\global\long\def\jl{\mathrm{\code{RandProj}}}
\global\long\def\poly{\mathrm{\code{TaylorExp}}}

\global\long\def\cheby{\mathrm{\code{ChebyExp}}}
\global\long\def\polyg{\mathrm{\code{InvSqrt}}}

\newcommand{\Tin}{T_{\textrm{inner}}}
\newcommand{\Tout}{T_{\textrm{outer}}}

\global\long\def\UpdateEst{\mathbf{\code{UpdateEstimator}}}
\global\long\def\estthetaone{\mathbf{\code{Estimator1}}}

\newcommand{\lr}[1]{\left( #1 \right)}

\newcommand{\PSDSet}[0]{\mathbb{S}_{\geq 0}}
\newcommand{\tcheby}{\mathrm{T}_{\mathrm{Cheby}}}
\newcommand{\dcheby}{\delta_{\mathrm{Cheby}}}
\newcommand{\tjl}{\mathrm{T}_{\mathrm{jl}}}

\newcommand{\dexp}{\delta_{\mathrm{exp}}}
\newcommand{\btonei}{b_{1_i}}
\newcommand{\bttwoi}{b_{2_i}}
\newcommand{\testjl}{\mathrm{T}_{{\mathrm{est}}_{\mathrm{jl}}}}
\newcommand{\testinvsqrt}{\mathrm{T}_{{\mathrm{est}}_{\mathrm{isq}}}}
\newcommand\encircle[1]{%
  \tikz[baseline=(X.base)] 
    \node (X) [draw, shape=circle, inner sep=0] {\strut #1};}

\newcommand*\numcircledtikz[1]{\tikz[baseline=(char.base)]{\node[shape=circle,draw,inner sep=1.2pt] (char) {#1};}}

\global\long\def\defeq{\stackrel{\mathrm{{\scriptscriptstyle def}}}{=}}

\usepackage{enumerate}
\usepackage{enumitem}
\usepackage{xfrac, bm} 
\usepackage[toc, page]{appendix}
\usepackage[makeroom]{cancel}  
\usepackage{physics}
\usepackage{graphicx,psfrag}
\usepackage{thm-restate}
\usepackage{caption}
\usepackage{tcolorbox,tikz,color,graphics,subcaption}
\usepackage{algorithm,algpseudocode,algorithmicx}
\usepackage{float} 
\usepackage{xifthen} 

\usepackage{booktabs}
\usepackage{array}
\numberwithin{equation}{section}
\usepackage{makecell} 
\usepackage[noabbrev,capitalize]{cleveref}

\newcommand{\swati}[1]{\textcolor{black}{#1}}

\title{An $\widetilde{\Ord}(m/\varepsilon^{3.5})$-Cost Algorithm for Semidefinite Programs with Diagonal Constraints }
\author{
			Yin Tat Lee\thanks{\texttt{yintat@uw.edu}. University of Washington.}
			\and
			Swati Padmanabhan\thanks{\texttt{pswati@uw.edu}. University of Washington.}
}

\begin{document}
\maketitle
\begin{abstract}
We provide a first-order algorithm for semidefinite programs (SDPs) with diagonal constraints on the matrix variable. Our algorithm outputs an $\varepsilon$-optimal solution with a run time of $\widetilde{\Ord}(m/\varepsilon^{3.5})$, where $m$ is the number of non-zero entries in the cost matrix. This improves upon the previous best run time of $\widetilde{\Ord}(m/\varepsilon^{4.5})$ by \cite{AroraKale}. As a corollary of our result, given an instance of the Max-Cut problem with $n$ vertices and $m \gg n$ edges, our algorithm returns a $(1 - \varepsilon)\alpha_{GW}$ cut in the faster time of $\widetilde{\Ord}(m/\varepsilon^{3.5})$, where $\alpha_{GW} \approx 0.878567$ is the approximation ratio by \cite{gw}. Our key technical contribution is to combine an approximate variant of the Arora-Kale framework of mirror descent for SDPs with the idea of trading off exact computations in every iteration for variance-reduced estimations in most iterations, only periodically resetting the accumulated error with exact computations. This idea, along with the constructed estimator, are of possible independent interest for other problems that use the mirror descent framework. 
\end{abstract}

\section{Introduction} Consider the SDP maximizing $ C\bullet X \defeq \Tr (C X)$ over the set of $n \times n$ positive semidefinite matrices with every diagonal entry bounded by a constant:
\begin{equation}
\begin{aligned}
\label[optn]{prob}
\max C\bullet X \text{ subject to } X \succeq 0, X_{ii} \leq 1 \text{ for all } i \in [n]. 
\end{aligned}
\end{equation} We seek a matrix $\widetilde{X}^* \succeq 0$ with  $\widetilde{X}^*_{ii} \leq 1$ for all $i\in[n]$ satisfying $C\bullet \widetilde{X}^* \geq C\bullet X^* - \varepsilon \sum_{i, j} \abs{C_{ij}}$, where $X^*$ is an optimal solution of (\ref{prob}). \swati{This is not an $\varepsilon$-multiplicative guarantee ($C\bullet \widetilde{X}^* \geq C\bullet X^*(1-\varepsilon)$), but a slightly weaker one, since $\sum_{i, j} |C_{ij}| \geq C\bullet X^*$. A multiplicative guarantee is not always easy to provide; indeed, many classical optimization algorithms also provide a guarantee only additive in some quantity that bounds from above the difference of the function values between the initial and optimal points. For example, gradient descent on an $L$-smooth convex function $f$ over a set with diameter $D$ returns, after $k$ iterations, a point $x_{k}$ such that $f(x_{k})-f(x^{*})\leq O(LD^2 k^{-1})$, where $f(x_0) - f(x^*) \leq O(LD^2)$.}

To solve (\ref{prob}) as per the above accuracy criterion, it suffices to solve (\ref{final-prob}): 
\begin{equation}
\begin{aligned}
\label[optn]{final-prob}
\min f(X) \defeq - \widehat{C} \bullet X  + \sum_{i =1}^n (X_{ii}-\rho_i)^{+}, \text{ subject to } X\succeq 0.  
\end{aligned} 
\end{equation} This problem is derived from (\ref{prob}) by promoting the diagonal constraints to the objective and appropriately scaling $C$ to $\widehat{C} \defeq \diag(1/\sqrt{\rho}) C \diag(1/\sqrt{\rho})$, where $\rho \in \reals^n$ such that $\rho_i =  \sum_{j\in [n]} \abs{C_{ij}}$. By rescaling $C_{ij} = n C_{ij}/\sum_{i, j} |C_{ij}|$, we assume $\sum_{i\in[n]}\rho_i = n$. Lemma~\ref{lem-dual-var-range} gives a solution of (\ref{prob}) from a solution of (\ref{final-prob}). 

For (\ref{prob}), \cite{AroraKale} have the previous best run time linear in $m \defeq \nnz{C}$, the size of the input. Though there exist algorithms with better dependence on $\varepsilon$, their dependence on $n$ is superlinear, as we describe in Section~\ref{sec-related-work}. In this paper, we operate in the regime of moderate $\varepsilon$ and large $n$, focusing on first-order methods. 

\cite{AroraKale} use the matrix multiplicative weights (MMW) update, which can be interpreted as mirror descent in the nuclear norm\footnote{The nuclear norm of a matrix $X \in \reals^{m \times n}$ is   the sum of its singular values: $\norm{X}_{\textrm{nuc}} \defeq \sum_{i=1}^{\min(m, n)} \sigma_i(X)$.}, using the negative entropy function, $\Phi(X) = X\bullet \log X$, over the scaled simplex, $\D = \{X: X\succeq 0, \Tr X = n\}$, as the mirror map. Their iterates at iteration $t$ are given by  \[X^{(t)} = n \frac{\exp(Y^{(t)})}{\Tr \exp(Y^{(t)})}, \, \text{ where } Y^{(t)} = \sum_{s = 1}^{t-1} -\eta \nabla f (X^{(s)}),\numberthis\label{ak-grad} \] with step size $\eta = \Ord(\varepsilon)$ and gradient $\nabla f(M) = \diag(\ones_{M \geq \rho}) - \widehat{C}$. Computing this gradient entails only comparing the diagonal entries of the current iterate with a fixed vector. Therefore, the na\"{i}ve computational cost of this method is dominated by  $\Omega(n^\omega)$ for the matrix exponentiation \cite{PanC99}, prohibitively expensive for a large problem dimension. \cite{AroraKale} circumvent this by \textit{approximating} the diagonal entries of the matrix exponential. Therefore, their overall cost is composed of the following three parts: (1) mirror descent requiring $\Ord(1/\varepsilon^2)$ iterations to converge, (2) degree $\Ord(1/\varepsilon)$ Taylor approximation of the matrix exponential, each matrix-vector product costing $\Ord(m)$, and (3) $\Ord(1/\varepsilon^2)$ random projections \cite{jl} to estimate the diagonal entries of the matrix exponential; combined, these give a run time of $\widetilde{\Ord}(m/\varepsilon^{5})$, which, \cite{FTCL} observe, can be sped up to $\Ord(m/\varepsilon^{4.5})$ by using Chebyshev (instead of Taylor) approximation of matrix exponentials (see \cite{SachdevaVishnoi}).

\paragraph{Our contribution.} In this work, we solve (\ref{prob}) with a run time of $\widetilde{\Ord}(m/\varepsilon^{3.5})$, thus speeding up the current best run time for this problem. Our result (formally stated in Theorem~\ref{thm-main-result}) is effected by careful technical work that incorporates variance-reduced estimators and fast products of matrix exponentials with vectors into the Arora-Kale framework of mirror descent for SDPs. We use the generalized negative entropy, $\Phi(X) = X \bullet \log (X) - \Tr X$, as our mirror map, and our primary high-level idea is the following: \textit{instead of exactly computing the primal iterate in each iteration, we frequently approximate it at a low accuracy (to reduce the cost) and infrequently at a high accuracy (to ``reset'' the error resulting from approximation)}. This idea is inspired by recent variance-reduction methods \cite{Shalev-Shwartz-sdca, svrg, saga, var_red_proj_free, Schmidt-sag}. The periodic high-accuracy computations and small bias and variance of estimators in the low-accuracy computations ensure sufficient closeness, in the appropriate norm, of the estimated iterates to the true ones,  which, by the convergence guarantee of approximate mirror descent, leads to an $\varepsilon$-optimal solution. Making this variance-reduction work in the MMW setting requires several technical ideas, as follows.

We introduce the technical idea of expanding the domain of our mirror map by a polylogarithmic factor. Due to the expanded domain and our choice of the mirror map, the gradient step of mirror descent falls in the \textit{interior} of this domain. The upshot of this is that the primal iterate is related to the dual via simply a matrix exponential, with no trace normalization as in Equation~\ref{ak-grad}. Thus, the quantity for which we require an estimator is greatly simplified. Drawing on the observation of \cite{AroraKale} that the gradient uses only the diagonal entries of the primal iterate, we build an estimator, with a small bias and variance, for the change in diagonal entries of the (dual) matrix exponential. We also prove the strong convexity parameter of our mirror map on the expanded domain by confecting classical results from convex analysis in a novel way. Due to the ubiquity of the MMW framework in optimization, efficient algorithms for SDPs, balanced separators, Ramanujan sparsifiers, packing/covering, and machine learning, we anticipate that our technical contributions will be useful for problems that hinge on the MMW foundation. 

\paragraph{Applications.} When $C$ is a graph Laplacian, (\ref{prob}) is the SDP relaxation of the Max-Cut problem \cite{gw}. An NP-complete problem \cite{Karp72}, Max-Cut has seen widespread utility in circuit design \cite{chen-kajitani-chan}, statistical physics \cite{Barahona-statphy}, semi-supervised learning \cite{WangJebaraChang}, and phase recovery \cite{Waldspurger2015}. Another instance of (\ref{prob}) is max-norm regularization \cite{JaggiThesis}, a convex surrogate for rank minimization \cite{maxnorm-srebro-props} enforcing simplicity in modeling observations \cite{FHBrankminapp}. SDPs of the form of (\ref{prob}) have also found applications in community detection \cite{abbe2015exact, guedon2016community, montanari2016semidefinite} and as relaxations to the maximum-likelihood estimator in the group synchronization problem \cite{singer2011three, bandeira2014multireference}. 

\subsection{Related work}\label{sec-related-work}
We describe in this section previous work on (\ref{prob}) using first-order methods, other than that of \cite{AroraKale}. Of note is that most papers below solve problems more general than (\ref{prob}), and the run times we mention occur when specialized to (\ref{prob}). 

\paragraph{Saddle-point formulation.} Since any SDP can be instantiated as an online convex optimization problem, we  apply to our setting some notable results from this area. To do so, we first reduce (\ref{prob}) to a feasibility problem  following the approach of \cite{AroraHK05}. Recall our assumption that $\sum_{i, j} |C_{ij}| = n$. The facts $X^*\succeq 0$ and $X^*_{ii} \leq 1$ for $i\in [n]$ imply $|X^*_{ij}|^2 \leq X^*_{ii} X^*_{jj} \leq 1$, which in turn bounds the optimum from above as $OPT = \sum_{i, j} C_{ij} X^*_{ij} \leq \sum_{i, j} |C_{ij}| |X^*_{ij}| \leq n$. We can also bound the optimum from below by choosing $X$ to be the zero matrix, thus bounding $OPT$ with $\lambda \in [0, n]$. Let $A_0 = \frac{1}{\lambda}C$, $b_0 = 1$, $A_i = -e_i e_i^\top$, and $b_i = -1$ for $i \in [n]$. Therefore, solving (\ref{prob}) requires, for each guess of $\lambda$ (obtained via a binary search over its range), solving the feasibility problem:  
\begin{equation}
\begin{aligned}
\label[optn]{feas-prob}
\text{Find $Z \in \PSDSet^n$ subject to } A_i \bullet Z - b_i \geq 0, \text{ for all $i \in \{0, n\}$},  \Tr Z \leq n. 
\end{aligned}
\end{equation} 
To do so, we leverage the saddle-point problem studied by \cite{GarberH16}, \[\max_{X \in \PSDSet^n, \Tr X = 1} \min_{p\in \reals^m_{\geq 0}, \norm{p}_1 = 1} \sum_{i=1}^m p_i (A_i \bullet X - b_i).\numberthis\label[optn]{GHprob}\] If the optimum of (\ref{GHprob}) is non-negative, solving it up to an additive accuracy of $\varepsilon$ is equivalent to finding a solution in the spectrahedron that satisfies all $A_i \bullet X -b_i \geq 0$ upto an additive error of $\varepsilon$. For (\ref{feas-prob}), this means the solution for (\ref{GHprob}) satisfies $X_{ii} \approx 1/n \pm \varepsilon$. However, due to the  requirement of $X_{ii} \approx 1 \pm \varepsilon$ in (\ref{prob}), the accuracy parameter of (\ref{GHprob}) must be $\varepsilon/n$. This causes the run time of \cite{GarberH16} for (\ref{prob}) to be $\widetilde{\Ord}(m (n/\varepsilon)^{2.5})$. By the same reasoning, when solving (\ref{prob}) to $\varepsilon$ multiplicative accuracy, the work of \cite{BaesBN13}, which uses a randomized Mirror-Prox algorithm, incurs a cost of $\widetilde{\Ord}(n^5/\varepsilon^3)$, and the recent algorithms of Follow the Compressed Leader by \cite{FTCL} and rank-1 sketch by \cite{CarmonDST19} incur a cost of $\widetilde{\Ord}(m (n/\varepsilon)^{2.5})$. \swati{It must be noted that \cite{GarberH16}, \cite{FTCL}, and \cite{CarmonDST19} provide algorithms satisfying $\varepsilon$-additive accuracy. When we translate our accuracy results to their language, the costs are not quite comparable. For instance, \cite{CarmonDST19}, for $\varepsilon$-additive accuracy for (\ref{prob}), incurs a cost of $m (n \norm{C}_{\infty}/\varepsilon)^{2.5}$. Our algorithm, using this accuracy criterion, incurs a cost of $m (\sum_{i, j} |C_{ij}|/\varepsilon)^{3.5}$. Unless we assume  additional structure on the matrix $C$, the comparison between these two costs is inconclusive.}

\paragraph{Low-rank updates.} When $C$ is the graph Laplacian in (\ref{prob}), there exists an $\varepsilon$-accurate solution of rank $\Ord(1/\varepsilon)$ \cite{Raghavendra2009, Montanari2016, mei2017solving}. Many papers capitalize on this fact and perform low-rank updates, which reduces cost per iteration. For example, \cite{KleinLu} base their algorithm on the framework of \cite{pst-framework} in conjunction with the power method to achieve a run time of $\widetilde{\Ord}(mn/\varepsilon^3)$. As another example, \cite{HazanSDP}  incorporates into the Frank-Wolfe algorithm \cite{frank2006algorithm} fast computation of an approximate minimum eigenvector and provides an $\widetilde{\Ord}(m n^3/\varepsilon^3)$-algorithm. A recent noteworthy result \cite{yurtsever2019scalable} returns a rank-$R$ approximation to an $\varepsilon$-optimal solution at a cost $\widetilde{\Ord}(R/\varepsilon^2 + n/\varepsilon^3)$. Even though, as alluded to earlier, there exists a rank-$\Ord(1/\varepsilon)$ solution to the MaxCut SDP, perturbing such a solution by an appropriately small amount gives an $\varepsilon$-optimal solution that is in fact full rank. Indeed, per Theorem $6.2$ of \cite{yurtsever2019scalable}, for any $r<R$, the iterate $\widehat{X}_t$ returned by their algorithm in iteration $t$ satisfies $\limsup_{t\rightarrow \infty} \E_{\Omega} \mathrm{dist}_*(\widehat{X}_t, \Psi_*) \leq (1 + r/(R - r - 1)) \cdot \max_{X \in \Psi_*} \| X - [X]_r\|_*$, where $\Omega$ is the randomness in their algorithm, $\Psi_*$ is the solution set, $R$ is the rank of the iterate returned, and $[X]_r$ is an $r$-truncated singular value decomposition of matrix $X$. The existence of full-rank matrices in the solution set $\Psi^*$ implies a possibly large bound above, so one cannot conclude that \cite{yurtsever2019scalable} improves upon our run time. 

\paragraph{Polynomial mirror map.} One of the contributions of \cite{FTCL} is a ``polynomial-style'' mirror map such as $\Phi(X) = \frac{1}{1 + 1/2p} \Tr X^{1 + 1/2p}$. The projection step with this map is $X = (Y^+)^{2p}$, where $Y^+$ is the matrix obtained by zeroing out the negative eigenvalues of $Y$, which is as expensive as matrix exponentiation.  

\paragraph{Variance-reduction methods.} Standard variance reduction algorithms such as SVRG \cite{svrg} minimize an objective that is a sum of functions, employing an unbiased estimator of the gradient. Unfortunately, neither is (\ref{final-prob}) a sum of functions, nor is its gradient ($\diag(\ones_{X>=\rho})$) cheap to estimate. 

\subsection{Preliminaries}\label{sec-prelims}  
\paragraph{Notation.} We use $\reals^n$ to denote the subspace of $n$-dimensional real vectors, $\ones$ for the vector of all ones, and $\ones_{\{\mathcal{E}\}}$ for the all-zero vector with one at coordinates where $\mathcal{E}$ is true. We use $x^+$ to denote the non-smooth function $x$ when $x\geq 0$ and zero otherwise. Denote by $\symm^n$ the subspace of $n \times n$ symmetric matrices and by $I_n$ the $n\times n$ identity matrix. For  $u\in \reals^n$,  $\diag(u)$ is the $n\times n$ diagonal matrix with $\diag(u)_{ii} = u_i$. For $A, B\in \symm^n$, the trace inner product is $A\bullet B\defeq \Tr (AB) = \sum_{i, j} A_{ij} B_{ij}$. We define $\dnorm{A} = \sum_i \abs{A_{ii}}$. Given a scalar function $f$ and a vector $u$, we use $f(u)$ to mean that entrywise, and similarly, for a symmetric matrix $A = U \Lambda U^\top$, $f(A) = U f(\Lambda) U^\top$. Given $A\in \reals^{n \times n}$ and $p \in \reals^n$,  $A \geq p$ means $A_{ii} \geq p_i$ for all $i \in [n]$. For $u \in \reals^n$,  $N\in \mathbb{N}$, and  vectors $\zeta_k \stackrel{{\mathrm{\small{i.i.d.}}}}{\sim} \N (0, I_n)$ for $k \in [N]$, the scalar $v = \jl(u, N)  \defeq \tfrac{1}{N}\sum_{k = 1}^N (u^T \zeta_k)^2$. This implies $\E v = \|u\|_2^2$. We extend the definition to $A \in \symm^n$ with each row of $A$ as the vector $u$. Then the diagonal matrix $B= \jl(A, N)$ satisfies $\E B = \diag A^2$. We use $\widetilde{\Ord}$ to denote polylogarithmic factors. The superscript $^*$ denotes optimality for variables and Fenchel conjugate for functions.
\begin{fact}[\cite{Allen-ZhuLO16}]\label{fact-extendedLiebThirring}
Given $A\succeq 0$, $B\in \symm^n$, and $\alpha \in [0, 1]$, the inequality $\Tr (B A^{\alpha} B A^{1-\alpha}) \leq A\bullet B^2$ holds.
\end{fact}
\begin{fact}[\cite{Wilcox67}]\label{fact-derivativeMatExp}
For a symmetric matrix-valued function $X(t)$ with argument scalar $t$, we have $\frac{d}{dt} \exp(X(t)) = \int_{\alpha = 0}^1 \exp(\alpha X(t)) \frac{d }{dt} X(t) \exp((1-\alpha) X(t)) d\alpha$. 
\end{fact}

\paragraph{Setup.} Our underlying algorithm to solve (\ref{final-prob}) is a slight variant of lazy mirror descent (also called Nesterov's Dual Averaging \cite{NesterovDualAvg}), which we term \emph{approximate lazy mirror descent}. To solve $\min_{ x \in \X} f(x)$ using this algorithm, select a mirror map $\Phi: \D \rightarrow \reals$ and a norm; the associated Bregman Divergence is $\D_{\Phi}(x, y) \defeq \Phi(x) - \Phi(y) - \inner{\nabla \Phi(y)}{x -y}$; set  $x^{(1)} \in \argmin_{\X \cap \D} \Phi(x)$ and $z^{(1)} \in \nabla^{-1} \Phi(0)$. We repeat, in succession, the gradient update, $\nabla \Phi(z^{(t+1)}) = \nabla \Phi(z^{(t)}) - \eta \nabla f(x^{(t)})$, and the \textit{approximate} projection, finding $\widetilde{x}^{(t+1)}$ satisfying $\E \|\widetilde{x}^{(t+1)} - x^{(t+1)}\|\leq \delta$, where $x^{(t+1)} \in \argmin_{x \in \X \cap \D} \D_{\Phi}(x, z^{(t+1)})$. 
%This is displayed in Algorithm~\ref{alg-almd} with its convergence guarantee in Theorem~\ref{thm-almd}.  

\begin{restatable}[Convergence of Lazy Mirror Descent]{theorem}{thmalmd}\label{thm-almd} Fix a norm $\|{}\cdot{}\|$. Given an $\alpha$-strongly convex mirror map $\Phi: \D \rightarrow \reals$ and a convex, $G$-Lipschitz  objective $f: \X \rightarrow \reals$, run Algorithm~\ref{alg-almd} with step size $\eta$ and $\E\|x^{(t)} - \widetilde{x}^{(t)}\| \leq \delta$. Let  $D \defeq \sup_{x \in \mathcal{X} \cap \mathcal{D}} \Phi\lr{x} - \inf_{x \in \X \cap \D}\Phi\lr{ x }$. Then, Algorithm~\ref{alg-almd} after $T$ iterations returns $\widetilde{x}^{t^*}$, satisfying  \[\E f(\widetilde{x}^{\lr{t^*}} )-f\lr{x^{*}}\leq\frac{D}{T\eta}+\frac{2\eta G^{2}}{\alpha}+\delta G. \numberthis\label[ineq]{errbnd-almd}\]
\end{restatable}
\begin{restatable}{lemmma}{dualvarrange}\label{lem-dual-var-range}Given  $C\in \reals^{n\times n}$ and $ 0 \preceq X$, let $\rho\in \reals^n$ with $\rho_i =  \sum_{j = 1}^n \abs{C_{ij}}$; diagonal matrix $S$ with $S_{ii} = \min(\sfrac{1}{\sqrt{\rho_i}}, \sfrac{1}{\sqrt{X_{ii}}})$ for $i \in [n]$;  $\widehat{X} = S  X  S$;  $\widehat{C} = \diag\left({\sfrac{1}{\sqrt{\rho}}}\right) C  \diag(\sfrac{1}{\sqrt{\rho}})$. Then, $\widehat{X}\succeq 0$, ${\widehat{X}}_{ii} \leq 1$ for all $i\in [n]$, and $\widehat{C} \bullet X - \sum_{i = 1}^n\left( X_{ii} - \rho_i\right)^{+} \leq C\bullet \widehat{X}$.
\end{restatable} 
\section{Our approach}\label{proposed-approach}  
We present our algorithm below, parameters in Table~\ref{TableNewParams}, and main result in Theorem~\ref{thm-main-result}.

\begin{algorithm}[!ht]
\DontPrintSemicolon
\caption{\textbf{Our Algorithm}}\label{alg-new}
%\begin{algorithmic}[1]
\KwIn{Cost matrix $C\in \reals^{n \times n}$, accuracy $\varepsilon$}
%\NoIndentInAlg{{\bf{Input}}: Cost matrix $C\in \reals^{n \times n}$, accuracy $\varepsilon$}
{\bf{Parameters}}: Displayed in Table~\ref{TableNewParams} \;

 Initialize $t \gets 0$, $Y^{(1)} \gets \zeros$. Set $\widehat{C}$ and $\rho$ from Lemma~\ref{lem-dual-var-range} and $\nabla f(X) = \diag(\ones_{X \geq \rho}) - \widehat{C}$ \;
 
\For{$\Tout$ \textrm{iterations}}{
	 $t \gets t+1$\;
	 
     $\widetilde{\exp}(\tfrac{1}{2}Y^{(t)}) \gets  \cheby( \tfrac{1}{2} Y^{(t)}, \tcheby, \dcheby )$ \Comment{Defined in  Corollary~\ref{chebyexp-def}}\;
     
	 $\widetilde{X}^{(t)} \gets \jl(\widetilde{\exp}(\tfrac{1}{2}Y^{(t)}), \tjl)$ \Comment{High-accuracy projection}\;
	 
	 $Y^{\lr{t+1}} \gets Y^{\lr{t}} - \eta \nabla f(\widetilde{X}^{\lr{t}})$ \Comment{Gradient update}\;
	 
	\For {$t_i = 1 \to \Tin$}{
		 $t \gets t + 1$\;
		 
		 $\widehat{\theta}^{\lr{t_i}} \gets \UpdateEst(\widetilde{X}^{\lr{t-1}}, Y^{\lr{t-1}}, \varepsilon, \eta)$ \Comment{See Algorithm~\ref{alg-est-theta}}\;
		 
		 $\widetilde{X}^{(t)}_{jj} \gets (\sqrt{\widetilde{X}^{(t-1)}_{jj} + 1} + \widehat{\theta}_j^{(t_i)})^2 - 1 $ for $j\in [n]$  \Comment{Constant-accuracy projection}\;
		 
		 $Y^{(t+1)} \gets Y^{(t)} - \eta \nabla f(\widetilde{X}^{(t)})$\Comment{Gradient update}\;
	}
}	
For $t^* \stackrel{\text{unif.}}{\sim} \left\{1, 2, \dotsc, t\right\}$, return $Y^{\lr{t^*}}$ and $S$,  where  $S$ is from Lemma~\ref{lem-dual-var-range}. \;
%\end{algorithmic}
\end{algorithm}	 

\begin{table}[htbp]
	\centering
	    \resizebox{\linewidth}{!}{% Resize table to fit within \linewidth horizontally
	\begin{tabular}{c c c}
	\toprule
	\multicolumn{1}{c}{\textbf{Parameter}}
	& \multicolumn{1}{c}{\textbf{Value}}
	& \multicolumn{1}{c}{\textbf{Proof}}\\
	\midrule
	Diameter $D$ & $K\log K$ & Lemma~\ref{lem_new_set_diam}\\
	Strong convexity $\alpha$ & $1/(4K)$ & Lemma~\ref{lem-new-alg-strong-conv}\\
	Step size $\eta$ & $\tfrac{1}{8 \times 10^4 (\log (n/\varepsilon))^{11}} \varepsilon^2$ & Lemma~\ref{pf-algnew-distbndinalg}\\ 
	Inner iteration count $\Tin$ & $\varepsilon^{-2}$ & Section~\ref{new-num-inner-iters}\\
    Outer iteration count $\Tout$ & $\tfrac{1}{\varepsilon}\cdot 24 \times 10^5 (\log (n/\varepsilon))^{11} \log n$ & Lemma~\ref{lem-newalg-distbnd-opt}\\
	JL projection count $\tjl$ &  $(2 \times 10^5)\cdot (\log n)^{21} \cdot \varepsilon^{-2} $ & Lemma~\ref{pf-algnew-distbndinalg} \\
	 Chebyshev approximation degree $\tcheby$ & $150 \log(n/\varepsilon) \cdot \varepsilon^{-1/2}$ & Lemma~\ref{algcurr-cheby-cost-alg} \\
     Chebyshev approximation accuracy $\dcheby$ & $(\varepsilon/n)^{401}$ & Lemma~\ref{algcurr-cheby-cost-alg}\\
	\bottomrule
	\end{tabular}}
	\caption{All Algorithm~\ref{alg-new} parameters and where their values are set. $K =  40 n (\log n)^{10} $.}
	\label[table]{TableNewParams}
	\vspace{-4mm}
	\end{table} 
	
\begin{theorem}[Main Result]\label{thm-main-result} Given $C \in \reals^{n \times n}$ with $m \geq n$ non-zero entries and  $0<\varepsilon\leq\tfrac{1}{2}$, we can find, in time $\tilde{\Ord}(m/\varepsilon^{3.5})$ and with high probability, a  matrix $Y \in \symm^{n}$ with $\Ord\lr{m}$ non-zero entries and a diagonal matrix $S \in \reals^{n \times n}$ so that\footnote{Since $\widetilde{X}^*$ can be dense, we represent it implicitly by only returning the matrices $Y$ and $S$.}  $\widetilde{X}^* \defeq S \cdot \exp Y   \cdot S$ satisfies  $\widetilde{X}^* \succeq 0$, $\widetilde{X}^*_{ii}\leq 1$ for $i \in [n]$, and $C\bullet {\widetilde{X}}^* \geq C\bullet X^* -\varepsilon  \sum_{i, j} |C_{ij}|.$
\end{theorem} 
As a corollary, for the Max-Cut problem on a graph with $n$ nodes and $m$ edges, our algorithm gives a cut that is $ (1-\varepsilon) \alpha_{GW}$ optimal\footnote{Assuming the Unique Games Conjecture, this is the best we can  hope for Max-Cut \cite{Khot2007}.}, in time $\widetilde{\Ord}(m/\varepsilon^{3.5})$, where $\alpha_{GW} \approx 0.878567$. Before proceeding to the proof sketch of Theorem~\ref{thm-main-result}, we call attention to a technical concept crucial to our analysis: \textit{we add to (\ref{final-prob}) the constraint $\Tr X \leq K$, where $K = 40 n (\log n)^{10}$}. The optimal $X^*$ remains valid under this constraint because $\Tr X^* = n$. Throughout our algorithm, this inequality remains inactive (Lemma~\ref{lem-newalg-distbnd-opt}). Coupled with the Legendre dual of our mirror map $\Phi(X) = X\bullet \log X - \Tr X$, this results in the primal and the dual being related by $X = \exp(Y)$ (Lemma~\ref{proj-new}). Since the gradient requires only diagonal entries of the primal iterate, we need estimators only for the diagonal entries of $\exp(Y)$. \smallskip
 
\begin{proof}[Proof Sketch of Theorem~\ref{thm-main-result}]\label{pf-sketch-main-result-main-body}
In this proof sketch, we compute the run time of Algorithm~\ref{alg-new}, proving the claims in Theorem~\ref{thm-main-result}. In doing so, we  provide intuition for the choice of parameters in Table~\ref{TableNewParams}. This sketch assumes that we are in iteration $t$ and drops all superscripts, and aside from that, follows the notation of Algorithm~\ref{alg-new}.
\begin{enumerate}[leftmargin=*]
\item To compute $\exp(Y)_{ii}$, we first approximate $\widetilde{\exp}(Y/2)$ to $\varepsilon$-accuracy using Chebyshev polynomials. We show in Lemma~\ref{spectrum-Y} that the spectrum of $Y$ lies in the range $[-\Ord(1/\varepsilon), \widetilde{\Ord}(1)]$, which allows for Chebyshev approximation with  $\widetilde{\Ord}(1/\sqrt{\varepsilon})$ terms, thus giving the cost of each projection to be $\widetilde{\Ord}(m/\sqrt{\varepsilon})$. The upper bound of $\widetilde{\Ord}(1)$ on the spectrum is critical to getting this cost, for in case of a symmetric range of $[-\Ord(1/\varepsilon), \Ord(1/\varepsilon)]$, the cost would be $\widetilde{\Ord}(1/\varepsilon)$. The $\widetilde{\Ord}(1/\sqrt{\varepsilon})$ terms is in contrast with the  $\Ord(1/\varepsilon)$ required for Taylor approximation. We then estimate each $\exp(Y)_{ii}$ with $\widetilde{\Ord}(1/\varepsilon^2)$  projections via the JL sketch in the high-accuracy steps, and   $\widetilde{\Ord}(1)$ randomized projections in the $\Tin$ low-accuracy steps. Therefore the total cost of the algorithm over $\Tout$ iterations is roughly $\Tout \cdot (m/\sqrt{\varepsilon})\cdot (1/\varepsilon^2 + \Tin)$. From this expression, the optimal choice of $\Tin$ (up to polylogarithmic factors) is $\Tin = 1/\varepsilon^2$. 
\item Due to the small bias and variance of our estimator, after $\Tin$ inner iterations, the estimated iterate is roughly within $\varepsilon K$ distance of the true iterate. Thus, the condition in Theorem~\ref{thm-almd} is satisfied, and its the error bound applies at the end of our algorithm: $\E f(\widetilde{X}^{*} )-f(X^{*})\leq D/(T\eta)+ 2\eta G^{2}/\alpha+\delta G$. Using $D$, $G$, and $\alpha$ from Table~\ref{TableNewParams} and $T_{\mathrm{inner}}$ from Step 1 and bounding by $\varepsilon K$, this inequality simplifies to $\varepsilon^2/(\eta T_{\mathrm{outer}}) + \eta  \leq \varepsilon$.
\item The step size $\eta$ is chosen by studying the error generated in each estimation step versus the error our framework can tolerate. Estimating $(\exp (Y + \Delta))_{ii}$ from $(\exp Y )_{ii}$ via a first-order approximation accrues an error of $\Tr ( \Delta \exp Y )$. Applying H\"older's inequality, the value of $G$, and the trace bound enforced by Lemma~\ref{lem-newalg-distbnd-opt} yields  $ \Tr (\Delta \exp Y) \leq \eta K$. Therefore, after $\Tin$ iterations, the variance of the error is $\Tin \eta^2 K^2$. Equivalently, the overall error after $\Tin$ iterations is $\sqrt{\Tin}\eta K$. For this to be bounded by $\varepsilon K$, we must have $\eta \leq \varepsilon/\sqrt{\Tin}$. Plugging in $\Tin$ from Step 1 gives the step size: $\eta \approx \varepsilon^2$. 
\item The value of $\eta$ from Step 3 and the inequality from Step 2 give $\Tout \approx 1/\varepsilon$. Plugging this value of $\Tout$ above gives the overall algorithm cost $\widetilde{\Ord}(m/\varepsilon^{3.5})$.
% We show the cost breakdown comparing our algorithm to Arora-Kale in Table~\ref{TableCosts}.
\end{enumerate}
We boost our result to the high probability statement of Theorem~\ref{thm-main-result} over multiple runs of the algorithm. We sidestep the issue of storage cost of $\widetilde{X}^*$ and cost of matrix-matrix products by dimension reduction techniques. This finishes the proof of our error guarantee. Lemma~\ref{lem-dual-var-range} implies that $\widetilde{X}^*\succeq 0$ and satisfies the diagonal constraints.
\end{proof}

\begin{table}[htbp]
\centering
    \resizebox{\linewidth}{!}{% Resize table to fit within \linewidth horizontally
\begin{tabular}{c c c c c}
\toprule
& \multicolumn{1}{c}{\textbf{\cite{AroraKale}}} &
\multicolumn{3}{c}{\textbf{Algorithm~\ref{alg-new} (This Paper)}} \\
\cmidrule(l){3-5}
& \thead{(Previous Best)} & \thead{\begin{tabular}{@{}c@{}}Low accuracy \\ steps \end{tabular}} & $+$ &  \thead{\begin{tabular}{@{}c@{}}High accuracy \\ steps \end{tabular}}\\
\midrule
Number of iterations & $\widetilde{\Ord}(\varepsilon^{-2})$ & $\widetilde{\Ord}(\varepsilon^{-3})$ & $+$ & $\widetilde{\Ord}(\varepsilon^{-1})$ \\
Number of projections per iteration & $\widetilde{\Ord}(\varepsilon^{-2})$ & $\widetilde{\Ord}(1)$ & $+$ & $\widetilde{\Ord}(\varepsilon^{-2})$\\
Cost per projection & $\Ord(m\varepsilon^{-1})$ & $\widetilde{\Ord}(m\varepsilon^{-1/2})$& $+$ & $\widetilde{\Ord}(m\varepsilon^{-1/2})$ \\
\textbf{Total Cost} & $\widetilde{\Ord}(m\varepsilon^{-5})$ & $\widetilde{\Ord}(m\varepsilon^{-3.5})$ & $+$ & $\widetilde{\Ord}(m\varepsilon^{-3.5})$\\
\bottomrule
\end{tabular}}
\caption{Comparing \cite{AroraKale} to our algorithm.}
\label[table]{TableCosts}
\vspace{-4mm}
\end{table}
\subsection{The estimator}\label{sec-outline-est} In this section, we consider the $t_i$'th iteration in the inner loop of Algorithm~\ref{alg-new}; suppose this is the $t$'th overall iteration. For now, we drop all superscripts and fix the notation below. 
\begin{definition}\label{defs-estimator} Let $\Delta = - \eta \nabla f(X)$,  $Y_s = Y +s\Delta$ for $ s\in [0, 1]$, $\bar{\tau} = 1-\tau$, $\dexp = \frac{4800 \varepsilon^{401}}{n^{390}}$, $\theta_{1_i} = (\exp(Y_s)_{ii} + 1)^{-1/2}$, $\theta_{2_i} = \frac{1}{2}(\exp(\bar{\tau} Y_s) \Delta \exp((\tau - 1/2) Y_s) \exp((1/2) Y_s))_{ii}$, $\btonei = \theta_{1_i} (2\dexp + \sqrt{2} (1+ 2\dexp) (\varepsilon/n)^{400})$, and $\bttwoi = 15\dexp \eta K$. 
\end{definition}
  To construct an estimator for the update from $\exp(Y)$ to $\exp(Y + \Delta)$, we estimate the update in $\sqrt{(\exp Y)_{ii} + 1}$. The motivation for this choice of function is two-fold: (1) because of the square root, the variance is controlled by the trace of the matrix exponential, bounded by Lemma~\ref{lem-newalg-distbnd-opt}; (2) since the derivative of square root is the inverse square root, we need $\sqrt{\exp(Y)_{ii} + 1}$ instead of $\sqrt{\exp(Y)_{ii}}$ to prevent the update term from becoming unbounded. By chain rule, Fact~\ref{fact-derivativeMatExp}, and the fundamental theorem of Calculus, 
\begin{align*} \sqrt{(\exp (Y + \Delta) )_{jj} + 1} &= \sqrt{(\exp(Y))_{jj} + 1} \\ &+  \underbrace{ \int_{s = 0}^1   \underbrace{ ((\exp Y_s)_{jj} + 1)^{-1/2}}_{\text{$\defeq \theta_{1_j}$; estimated using $\widehat{\theta}_{1_j}$}}  \underbrace{\tfrac{1}{2}( \int_{\tau = 0}^1 \exp (\tau Y_s) \Delta  \exp (\bar{\tau} Y_s) d\tau )_{jj}}_{\text{$\defeq \theta_{2_j}$; estimated using $ \widehat{\theta}_{2_j}$}} ds}_{\text{$\defeq \theta_{j}$; estimated using $\widehat{\theta}_j$}}. \numberthis\label{estFTC}
\end{align*} As indicated in Equation~\ref{estFTC}, we split the quantity to be estimated into two parts, separately estimating each. Estimating the first part, $\widehat{\theta}_{1_j}$, requires first estimating $\exp(Y_s)_{jj} + 1$ using a JL sketch and then passing through the following Taylor approximation for the function $g(u) = u^{-1/2}$, where $g^{(k)}(x)$ is the $k$'th derivative of $g$ at $x$, \[\polyg (\widetilde{X}, N) \defeq \sum_{k = 0}^{N-1}\frac{1}{k!} g^{(k)}(x_0) \prod_{j = 1}^k ( x_{k, j} - x_0 ), \, \mbox{where } x_0, x_{k, j} \stackrel{\text{i.i.d.}}{\sim} \widetilde{X}. \numberthis\label{def-polyg}\] Since $\widehat{\theta}_{1_j}$ must be \textit{unbiased}, it is essential to do the Taylor approximation instead of simply evaluating $g(u) = u^{-1/2}$ at the estimator of $\exp(Y_s)_{jj} + 1$. Indeed, for a general $f$ and a random variable $\widetilde{x}$ that is an unbiased estimator of $x$, $\E f(\widetilde{x}) = f(\E \widetilde{x})$ does not hold, as evidenced by Jensen's inequality; on the other hand, the intuition for the quantity from Equation~\ref{def-polyg} to be unbiased is that each term in the sum is a product of independent, unbiased random variables. Estimating $\theta_{2_j}$ is done by splitting it into carefully chosen parts and applying the JL sketch. Algorithm~\ref{alg-est-theta} is the complete subroutine for the estimator. 
\begin{algorithm}[!ht]
\caption{$\UpdateEst(\text{Primal }X, \text{dual }Y, \text{accuracy  }\varepsilon, \text{step size }\eta)$}\label{alg-est-theta}
\begin{algorithmic}[1]
\State Parameters $\testjl = 2^{22} 10^4 (\log (n/\varepsilon))^2$ and $\testinvsqrt = 1600 \log (n/\varepsilon)$ (set in Lemma~\ref{lem-est-term1})
\State Sample $s$ and $\tau$ uniformly from $[0, 1]$. Compute $\Delta$ and $Y_s$ as per Definition~\ref{defs-estimator}. Let $\widetilde{X}_s = \jl(\widetilde{\exp}(Y_s/2), \testjl)$. Sample $\zeta \sim \N(0, I_n)$. 
\State Compute $\widehat{\theta}_{1_j} = \polyg({\widetilde{X}}_{s_{jj}} + 1, \testinvsqrt)$ for $j \in [n]$. 
\State Compute $\widehat{\theta}_{2_j} = \tfrac{1}{2} (\widetilde{\exp}((\tau - \tfrac{1}{2}) Y_s) \Delta \widetilde{\exp}(\bar{\tau} Y_s) \zeta)_j \lr{\widetilde{\exp}(Y_s/2)  \zeta}_{j}$ for $j\in [n]$. 
\State Return the overall estimator, $\widehat{\theta}_j =  \widehat{\theta}_{1_j}  \widehat{\theta}_{2_j}$, for $j \in [n]$. \Comment{Coordinate-wise product}
\end{algorithmic}
\end{algorithm} 

\textbf{Properties of the estimator.} The bounds on bias and variance of the estimator, as required by Theorem~\ref{thm-main-result}, are stated in Lemma~\ref{lem-var-bound-on-xupdate}. Since $\widehat{\theta}$ is constructed from $\widehat{\theta}_1$ and $\widehat{\theta}_2$, we first state their properties and use them to sketch a proof of Lemma~\ref{lem-var-bound-on-xupdate}. 
\begin{restatable}{lemmma}{biasandvarbndest}\label{lem-var-bound-on-xupdate} 
The estimator $\widehat{\theta}^{(t)}$ has the following bounds on its first and second moments. 
\begin{description}
\compresslist{
\item [{(1)}] $|\E\widehat{\theta}_i - \int_{s=0}^1 \int_{\tau = 0}^1 \theta_{1_i} \theta_{2_i}  ds d\tau| \leq \btonei \theta_{2_i} + \bttwoi \theta_{1_i} + \btonei \bttwoi$ for $i\in[n]$. 
\item  [{(2)}]$\E \|\widehat{\theta}\|_2^2 \leq  19600 \log(n/\varepsilon) K \eta^2 + 147000 K^2 \eta^2 \dexp$. 
}
\end{description}
\end{restatable} 

\begin{restatable}{lemmma}{termone}\label{lem-est-term1} Given $\testinvsqrt = 1600 \log(n/\varepsilon)$, $\testjl = 2^{14} \testinvsqrt^2$, $Z \in \symm^n$, and $\varepsilon\in \lr{0, \sfrac{1}{2}}$, let $\widetilde{Z^2} = \jl(Z, \testjl)$ and $\widehat{\theta}_{1_i} \sim \polyg((\widetilde{Z^2})_{ii} + 1, \testinvsqrt)$ for $i \in [n]$. Then, 
\begin{description}
\compresslist{
\item [{(1)}] The first moment satisfies $\abs{ \E \widehat{\theta}_{1_i} - \frac{1}{\sqrt{(Z^2)_{ii} + 1}} } \leq \frac{\sqrt{2}(\varepsilon/n)^{400}}{\sqrt{(Z^2)_{ii} + 1}}$. 
\item [{(2)}] The second moment satisfies $\E |\widehat{\theta}_{1_i}|^2\leq \tfrac{1}{(Z^2)_{ii}} 1630 \log (n/\varepsilon)$.
}
\end{description}
\end{restatable}
\begin{restatable}{lemmma}{termtwo}\label{lem-est-term2}
Consider $Z_1, Z_2, Z$, and $\Delta$ all in $\symm^n$. Sample $\zeta\sim \mathcal{N}(\zeros,I_n)$, and define  $\widehat{\theta}_2 \in \reals^n$ as $\widehat{\theta}_{2_i} = (Z_1 \Delta Z_2 \zeta)_i \lr{Z \zeta}_i$. Define $\theta_{2_i} \defeq  (Z_1 \Delta Z_2 Z)_{ii}$.  Then for $i \in [n]$:
\begin{description}
\compresslist{
\item [{(1)}] The first moment satisfies $\E \widehat{\theta}_{2_i} =  \theta_{2_i}$
\item [{(2)}] The second moment satisfies $\E | \widehat{\theta}_{2_i}|^{2} \leq 3\lr{Z_1 \Delta Z_2^2 \Delta Z_1}_{ii}\lr{Z^2}_{ii} $. 
}\end{description}\end{restatable}

 \begin{proof}[Proof sketch for Lemma~\ref{lem-var-bound-on-xupdate}]
   By construction, \[\E_{s, \tau, \zeta_1, \zeta_2} \|\widehat{\theta}\|_2^2 = \int_{s=0}^1  \int_{\tau = 0}^1 \sum_{i=1}^n \E_{\zeta_1} |\widehat{\theta}_{1_i}|^2 \E_{\zeta_2} |\widehat{\theta}_{2_i}|^2 ds d\tau.\] Plugging in the second moment bounds from Lemma~\ref{lem-est-term1} and Lemma~\ref{lem-est-term2} gives
\begin{align*}
\E_{s, \tau, \zeta_1, \zeta_2}\|\widehat{\theta}\|_2^2 &= 4890  \log(n/\varepsilon) \int_{s=0}^1 \int_{\tau=0}^1  \Tr (\widetilde{\exp}(2\bar{\tau} Y_s) \Delta \widetilde{\exp}((2\tau - 1)Y_s) \Delta) ds d\tau. 
	\end{align*} This step is made possible by the careful choice of split in $\widehat{\theta}_2$ that enable cancellations of $\frac{1}{(\widetilde{\exp} Y_s)_{ii}}$ and $(\widetilde{\exp} Y_s)_{ii}$. Applying Fact~\ref{fact-extendedLiebThirring} and the fact that $\widetilde{\exp} Y_s$ is close to the true $\exp Y_s$, the above trace term is bounded by $\Tr (\exp(Y + s\Delta) \Delta^2)$ (plus a small error term).  Applying H\"older's Inequality, Lemma~\ref{lem-newalg-distbnd-opt}, and values of $\eta$ and $G$ completes the proof. 
\end{proof}

To provide proof sketches of Lemma~\ref{lem-est-term1} and Lemma~\ref{lem-est-term2}, we need two technical lemmas about $\jl$ and $\polyg$, the main workhorses for our estimators. These lemmas follow from properties of Gaussian and the scaled chi-squared distribution.  
\begin{restatable}{lemmma}{prefirstpartfirstterm}\label{lem-est-pre-term1}
Consider a positive random variable $x$ sampled from a distribution $X$ with mean $\mu$ and variance $\sigma^2$. For some integer $k>0$, construct the distribution $\mathcal{G}(X) = \polyg\lr{X, k}$ defined in Equation~\ref{def-polyg}. Then the random variable $g \sim \mathcal{G}(X)$ satisfies  
\begin{description}
\compresslist{
\item [{(1)}] $|\E g -{\mu}^{-{1/2}}| \leq \E \lr{ \frac{|x-\mu|^{k}}{\min\lr{\mu, x}^{k + {1/2}}}}$
\item [{(2)}] $\E|g|^{2}\leq k \sum_{j=0}^{k-1}\E\lr{\frac{\left(\sigma^{2}+\lr{\mu-x}^{2}\right)^j}{x^{2j+1}}}$. 
}
\end{description}
\end{restatable}
\begin{restatable}{lemmma}{presecondpartfirstterm}\label{lem-est-pre2-term1}
Given $u \in \reals^n$ such that $\mu \defeq \norm{u}_2^2 \neq 0$, and positive integers $k > 1$ and $N \geq 4k + 6$, the following are true for $x$ sampled from   $X = \jl\lr{u, N}$. 
\begin{description}
\compresslist{
\item [{(1)}] $\E x = \mu$
\item [{(2)}] $\sigma^2 \defeq \E \lr{x - \mu}^2 = \frac{2 \mu^2}{N}$ 
\item [{(3)}] $\E \left( \frac{\lr{\sigma^2 + \lr{ x - \mu }^2 }^k}{\min\lr{x, \mu}^{2k + 1}} \right) \leq \frac{1}{\mu} \lr{ \frac{e^{{N/2}}}{2^{N - 17k}}  + \frac{ 2^{13k} k^{2k}}{N^k} }$ 
}
\end{description}
\end{restatable}
%%%%%%%%%%%%%%%%%%%%%%%%%%%%%%%%%%%%%%%%%%%%%%%%%%%%%%%%%%%%%%%%%%%%%%%%%%%%%%%%%%%%%%%%%%
\begin{proof}[Proof sketches of Lemmas~\ref{lem-est-term1} and \ref{lem-est-term2}]
 Consider $x \sim \widetilde{Z^2}_{ii}$.  By Lemma~\ref{lem-est-pre2-term1}, $\E x =  Z^2_{ii}$. This satisfies the bias requirement of Lemma~\ref{lem-est-pre-term1}, and therefore 
	\begin{align*}
		\abs{ \E \widehat{\theta}_{1_i} - \frac{1}{\sqrt{1 + (Z^2)_{ii}}}  } &\leq \E \lr{ \frac{\abs{x - (Z^2)_{ii}}^{\testinvsqrt}}{\min(x + 1, (Z^2)_{ii}+1)^{\testinvsqrt + \tfrac{1}{2}}}}\\
																&\leq \sqrt{\E \frac{\lr{x - (Z^2)_{ii}}^{2\testinvsqrt}}{\min\lr{x + 1, (Z^2)_{ii}+1}^{2\testinvsqrt + 1}} }\\
																&\leq \sqrt{ \frac{1}{(Z^2)_{ii} + 1} \left(\frac{e^{\sfrac{\testinvsqrt}{2}}}{2^{\testjl - 17\testinvsqrt}}  + \frac{2^{13\testinvsqrt} {\testinvsqrt}^{2\testinvsqrt}}{{\testjl}^{\testinvsqrt}} \right) }. 
    \end{align*} where the first step is by Lemma~\ref{lem-est-pre-term1}, the second is by Jensen's inequality, and the third step is by a slight modification of $(3)$ in Lemma~\ref{lem-est-pre2-term1}. The values of $\testinvsqrt$ and $\testjl$ from Algorithm~\ref{alg-est-theta} give the final bias bound. The second moment bound follows similarly, and the properties of $\widehat{\theta}_2$ follow from simple properties of the Gaussian distribution.
    \end{proof}
   
\subsection{Technical Concepts: Domain Expansion and Strong Convexity}\label{sec-technicalideasmainbody}
In this section we state and sketch the proofs of two key technical concepts: (1) the addition of the trace constraint as described before the proof of Theorem~\ref{thm-main-result}, and (2) the value of the strong convexity parameter of our mirror map over this new domain. 
\begin{restatable}{lemmma}{newdistbndfromopt}\label{lem-newalg-distbnd-opt} 
With the choice of parameters in Algorithm~\ref{alg-new}, the iterate $\widetilde{X}^{(t)}$ at any iteration $t$ satisfies $\Tr \widetilde{X}^{(t)} < K$ for $K = 40 n (\log n)^{10}$. 
\end{restatable}
\begin{proof}[Proof sketch] We assume that for any iteration $t$, the primal iterate is close to the optimal point and satisfies $\dnorm{\widetilde{X}^{\lr{t}} - X^*} \leq 38 n \lr{\log n}^{10}$. In Algorithm~\ref{alg-new}, $Y^{(1)} = 0$ implies $\widetilde{X}^{\lr{1}} = I$. We also know that the optimal point satisfies $\Tr X^* = n$. Therefore, in the base case, $\dnorm{\widetilde{X}^{\lr{1}} - X^*} \leq 2n \leq 38 n \lr{\log n}^{10}$.  Suppose that the hypothesis is true for some $ t = t^{\prime}$. We complete the proof  by first proving a weak bound for $\dnorm{\widetilde{X}^{(t)} - X^*}$ using the triangle inequality of norms and then boosting our bound (thereby obtaining the stronger guarantee of the induction hypothesis) by invoking the strong convexity of the Bregman divergence. The full proof is presented in Section~\ref{pf-proj-in-exp-int}. 
\end{proof}
We now sketch the proof of the strong convexity parameter of our mirror map, the \emph{generalized} negative entropy function. This mirror map is different from the negative entropy function and has recently appeared in \cite{AllenOrecchia2015-parallel}. 
\begin{restatable}{lemmma}{newalgstrongconvexity}\label{lem-new-alg-strong-conv}
The function $\Phi(X) =  X\bullet \log X - \Tr X$ is $\frac{1}{4K}$-strongly convex with respect to the nuclear norm over the domain $\D = \{X: X \succeq 0, \Tr X \leq K\}$. 
\end{restatable}
\begin{proof}[Proof sketch] We invoke the duality between strong convexity and smoothness by \cite{KakadeDuality}, the characterization of matrix smooth functions by \cite{Judit-Nemirov}, and the generalization of convexity of a permutation-invariant function on vectors to a spectral function on matrices by \cite{lewis95}. Our proof requires the following definition. 
\begin{definition}\label{defs-sc} Define the vector functions $\psi_1(y) = \sum_{i=1}^n \exp y_i$, $\psi_2(y) = 2K \log \psi_1(y) - 2K \log(2K) + 2K$, $\psi(y) = \psi_1(y)$ if $\psi_1(y) \leq 2K$ and $\psi_2(y)$ otherwise; $\Psi(Y) = \Psi_1(Y)$ if $\Psi_1(Y) \leq 2K$ and $\Psi_2(Y)$ otherwise; and $\phi(x) = \sum_{i=1}^n x_i \log x_i - \sum_{i=1}^n x_i$.  Define the corresponding matrix functions $\Psi_1(Y) = \Tr \exp Y$, $\Psi_2(Y) = 2K \log \Psi_1(Y) - 2K \log (2K) + 2K$, and $\Phi(X) = X\bullet \log X - \Tr X$. 
\end{definition}
Our first step is to show that $\Psi$, the matrix version of $\psi$, satisfies the property $\Psi^*(Y) = \Phi(Y)$ over $\{Y: Y\succeq 0, \Tr Y \leq K\}$. To prove this, we first prove that $\psi$ and its matrix version, $\Psi$, are both continuously differentiable at the boundary of definition of their respective two parts. We then show that $\psi_1$ and $\psi_2$ are convex; combining this with the claim about continuous differentiability implies convexity of $\psi$, which immediately extends to $\Psi$ by a result of \cite{lewis95}. We then show that $\psi$ and $\phi$ satisfy $\psi_1^*(x) = \phi(x)$ for $x\in \reals^n_{+}$,
and given an input $x \in \{x:x_i\geq 0, \sum_{i=1}^n x_i \leq K\}$, the point $y$ attaining the optimum in computing $\psi_1^*(x)$ lies in the \emph{interior} of the set $\{y: \psi_1(y)\leq 2K \}$. Therefore, given an input $x \in \{x:x_i \geq 0, \sum_{i=1}^n x_i \leq K\}$, we invoke the preceeding facts to conclude that the point at which the value of $\psi^*(x)$ is attained must be the same as that for $\psi_1^*(x)$. This implies $\psi^*(x) = \psi_1^*(x)$ for $x\in \{x: x_i\geq 0, \sum_{i=1}^n x_i\leq K\}$. By a result of \cite{lewis95}, this extends to $\Psi^* = \Phi$ on $\{X: X\succeq 0, \Tr X \leq K\}$. We then use \cite{Judit-Nemirov} and continuous differentiability at the boundary to show that $\Psi$ is $4K$-smooth in the operator norm which in turn implies, by \cite{KakadeDuality}, that $\Psi^*$ is $1/(4K)$-strongly convex in the nuclear norm, finishing the proof. Our full proof is in Section~\ref{sec-new-params}. 
\end{proof}

\section*{Acknowledgment} We are very grateful to the anonymous reviewers of SODA $2019$, SODA $2020$, and COLT $2020$ for their careful reading and constructive suggestions on improving our presentation, Kevin Tian (Stanford) for helpful explanations of his paper \cite{CarmonDST19}, and Sidhanth Mohanty (UC Berkeley) for his clear explanation of bounds on the rank of the MaxCut SDP from \cite{Raghavendra2009} and \cite{Montanari2016}. 

\newpage
\bibliography{MaxCutPaper.bib} 
\newpage
\begin{appendices}
We organize the appendix into four parts: Section~\ref{app-curr-alg}, analysis common to \cite{AroraKale} and us; Section~\ref{pf-currparams} and Section~\ref{pf-proposed-approach}, analysis of \cite{AroraKale} and our algorithm, respectively; Section~\ref{lin-al-results}, general technical results. 
\section{Analysis Common to Both Algorithms}\label{app-curr-alg}
In this section we provide proofs for two results: the first is that a solution to the reformulated problem (\ref{final-prob}) is indeed $\varepsilon$ close to that of the original; the second is the convergence guarantee of approximate lazy mirror descent, the framework for both the Arora-Kale algorithm as well as ours. 
 \begin{algorithm}[!ht]
 \DontPrintSemicolon
	\caption{Approximate lazy mirror descent}\label{alg-almd}
	\KwIn{ Objective function $f:\X \rightarrow \reals$, accuracy parameter $\varepsilon$.}
	{\bf{Parameters}}: Mirror map $\Phi:\D \rightarrow \reals$, norm $\norm{{}\cdot{}}$, step size $\eta$, iteration $T$, error bound $\delta$.\;
	
	Initialize: $x^{\lr{1}} \in \argmin_{x\in \X \cap \D} \Phi(x)$, $\widetilde{x}^{\lr{1}} = x^{\lr{1}}$, and $z^{\lr{1}}$ satisfying $\nabla \Phi(z^{\lr{1}}) = 0$.\;
	
	\For{$t = 1 \to T$}{	
	    $\nabla \Phi(z^{(t+1)}) \gets \nabla \Phi(z^{(t)}) - \eta \nabla f(\widetilde{x}^{(t)})$ \Comment{Lazy gradient update}\;
	    
	    Find $\widetilde{x}^{\lr{t+1}}$ such that $\E \|\widetilde{x}^{\lr{t+1}} - x^{\lr{t+1}}\| \leq \delta$, where $x^{\lr{t+1}} \in \argmin_{x \in \X \cap \D}  \D_{\Phi}(x, z^{\lr{t+1}})$ \Comment{Approximate projection}
	}
	For $t^*\stackrel{\text{unif.}}{\sim} \left\{1, 2, \dotsc, T\right\}$, return $\widetilde{x}^{\lr{t^*}}$. \;
\end{algorithm}
\subsection{From the Reformulated to the Original SDP}\label{pf-dual-var-range}
Our claim of reformulating (\ref{prob}) as (\ref{final-prob}) works because once we have a solution $X$ for the latter, we can apply the following result to obtain a matrix $\widehat{X}$ which satisfies all the required constraints of (\ref{prob}), and at which the objective value in (\ref{prob}) is better than that at $X$ in (\ref{final-prob}). 
\dualvarrange*
\begin{proof}
We first prove the positive semidefiniteness. Observe that since $\widehat{X}$ and $X$ are similar matrices, $X \succeq 0$ implies $\widehat{X}\succeq 0$ as well. Next, we define a matrix $Y$ as $Y_{ij} = \frac{X_{ij}}{\sqrt{\rho_i} \sqrt{\rho_j}}$. Without loss of generality, assume $Y_{11} \geq Y_{22} \geq \dotsc \geq Y_{nn}$. We also define a diagonal matrix, $\hatD$ as ${\hatD}_{ii} = \min(1, 1/\sqrt{Y_{ii}})$. If $Y_{ii} \geq 1$, then $\widehat{X}_{ii} = \frac{\rho_iY_{ii}}{\sqrt{\rho_i Y_{ii}}\sqrt{\rho_i Y_{ii}}} = 1$; otherwise, $\widehat{X}_{ii} = {Y_{ii}}$. This proves that $\widehat{X}_{ii} \leq 1$ for all $1\leq i \leq n$, which is precisely the claim bounding every diagonal entry. We now prove the claim about the objective value. By definition of $\widehat{D}$, $\widehat{X}$ and $Y$, we have $\widehat{X} = \widehat{D}\cdot Y\cdot \widehat{D}$. Therefore we get
\begin{align*}
	{C}\bullet (\widehat{X} - Y) - \sum_{i = 1}^n C_{ii} Y_{ii} ({\hatD}_{ii}^2 - 1) &=  \sum_{ i = 1}^n \sum_{j \neq i} C_{ij} Y_{ij} ( {\hatD}_{ii}{\hatD}_{jj} - 1)\\
				  &=  2\sum_{i = 1}^n\sum_{i < j} C_{ij} Y_{ij} ( {\hatD}_{ii}{\hatD}_{jj} - 1). 
\end{align*} The definition of $\hatD$ and the ordering assumption on $\{Y_{ii}\}$ imply $0 < {\hatD}_{11} \leq {\hatD}_{22} \leq \dotsc \leq {\hatD}_{nn} \leq 1$, which in turn means ${\hatD}_{ii} {\hatD}_{jj} \geq {\hatD}_{ii}^2$. Further, since $X\succeq 0$ and $Y = \diag (\sfrac{1}{\sqrt{\rho}})\cdot X\cdot \diag (\sfrac{1}{\sqrt{\rho}})$, we have $Y\succeq 0$. Therefore $Y_{ii} Y_{jj} \geq Y_{ij} Y_{ji}$. By symmetry of $Y$ and the assumed ordering of $\left\{Y_{ii}\right\}_1^n$, this can be simplified to $Y_{ii} \geq \abs{Y_{ij}}$ for $i < j$. These two facts simplify the above to 
\begin{align*}
{C}\bullet (\widehat{X} - Y) - \sum_{i = 1}^n C_{ii} Y_{ii} ({\hatD}_{ii}^2 - 1) &\geq 2\sum_{i = 1}^n \sum_{i < j} \abs{ C_{ij}} \abs{Y_{ij}}({\hatD}_{ii}^2 - 1)\\
	   &\geq 2\sum_{i = 1}^n\sum_{i < j } \abs{C_{ij}}Y_{ii}({\hatD}_{ii}^2 - 1)
\end{align*} Finally, since $\widehat{D}_{ii} \leq 1$, we have $\widehat{D}_{ii}^2 - 1 \leq 0$. Rearranging the terms in the last inequality, we get
\begin{align*}
{C}\bullet (\widehat{X} - Y) &\geq  \sum_{i = 1}^n C_{ii} Y_{ii}( {\hatD}_{ii}^2 - 1) + \sum_{i = 1}^n Y_{ii} ({\hatD}_{ii}^2 - 1) (\sum_{j > i} \abs{C_{ij}} + \sum_{j<i} \abs{C_{ij}})\\
					  &= \sum_{ i= 1}^n Y_{ii} ({\hatD}_{ii}^2 - 1)\underbrace{ \lr{C_{ii} + \sum_{i > j} \abs{ C_{ij} } + \sum_{i < j} \abs{C_{ij}}}}_{\text{$ \leq \rho_i$}} \\
					  &\geq \sum_{i = 1}^n Y_{ii} \rho_i ({\hatD}_{ii}^2 - 1) \\
					  &= -\sum_{i = 1}^n \rho_i \left( Y_{ii} - 1\right)^{+}
\end{align*} where we used $\hatD_{ii} = \min(1, 1/\sqrt{Y_{ii}})$ in the last step. Rearranging the terms in the last inequality gives \[{C}\bullet {\widehat{X}} \geq {C}\bullet {Y} - \sum_{i = 1}^n \rho_i \lr{Y_{ii} - 1}^{+} =  \widehat{C}\bullet {X} - \sum_{i = 1}^n(X_{ii} - \rho_i)^{+}, \] where the last step is  by definition of matrix $Y$. 

\end{proof}

%%%%%%%%%%%%%%%%%%%%%%%%%%%%%%%%
%%%%%%%%%%%%%%%%%%%%%%%%%%%%%%%%
\subsection{Analysis of Approximate Lazy Mirror Descent}\label{pf-almd}
 We now derive the convergence bound of approximate lazy mirror descent. The proof closely follows that of Theorem $4.3$ in Bubeck's monograph \cite{Bubeck}. 
\thmalmd*
\begin{proof} By convexity of $f$,%, for any point $x \in \dom\lr{f}$,
	\begin{align*}
		\sum\limits_{t=1}^T (  f(\widetilde{x}^{\lr{t}}) - f(x)) &\leq \sum\limits_{t=1}^T \inner{\nabla f(\widetilde{x}^{\lr{t}})}{\widetilde{x}^{\lr{t}} - x}=  \underbrace{ \sum\limits_{t=1}^T \inner{\nabla f(\widetilde{x}^{\lr{t}})}{\widetilde{x}^{\lr{t}}- x^{\lr{t}}}}_{\text{\encircle{A}}} + \underbrace{\sum_{t=1}^T \inner{\nabla f(\widetilde{x}^{\lr{t}})}{x^{\lr{t}} - x}}_{\text{\encircle{B}}}.\numberthis\label[ineq]{almd-ineq1}
	\end{align*} The term $\encircle{A}$ can be bounded by Cauchy-Schwarz inequality and the invariant $\E \norm{x^{\lr{t}} - \widetilde{x}^{\lr{t}}} \leq \delta$:
\[\encircle{A} \leq \sum\limits_{t=1}^T  \norm{\Delta^{\lr{t}}} \norm{\nabla f\left(\widetilde{x}^{\lr{t}}\right)}_{*} \leq \delta G T. \numberthis\label[ineq]{almd-circ2-bnd} \]Next, recall that Algorithm~\ref{alg-almd} initializes $x^{\lr{1}} \in \argmin_{\X \cap \D} \Phi(x)$ and $z^{\lr{1}}$ satisfying $\nabla \Phi(z^{\lr{1}}) = 0$, and repeats the following two steps: 
	\begin{align*}
		\nabla \Phi(z^{\lr{t}}) &= \nabla \Phi(z^{\lr{t-1}}) - \eta \nabla f(x^{\lr{t}})\\
		x^{\lr{t}} &= \argmin_{\X \cap \D} D_{\Phi}(x, z^{\lr{t}}).
	\end{align*} Now consider the potential function $\widetilde{\Psi}_t(x) \defeq \Phi(x) + \eta \inner{x}{\sum_{s  = 1}^t \nabla f(\widetilde{x}^{\lr{s}})}$. Applying the recursive definition of the gradient step, we can express $x^{\lr{t+1}} = \argmin\limits_{x \in \X \cap \D} \widetilde{\Psi}_{t}\lr{x}$. Since $\Phi$ is $\alpha$-strongly convex, so is the potential function $\Psi_t$. We can express these two statements as follows: 
	\begin{align*}
		\widetilde{\Psi}_t(x^{\lr{t+1}}) - \widetilde{\Psi}_t(x^{\lr{t}}) &\leq \underbrace{\inner{\nabla \widetilde{\Psi}_t(x^{\lr{t+1}})}{ x^{\lr{t+1}} - x^{\lr{t}} }}_{\text{$\leq 0$, by optimality of $x^{\lr{t+1}}$}} - \tfrac{\alpha}{2}\norm{x^{\lr{t+1}} - x^{\lr{t}}}^2 \\
														&\leq -\tfrac{\alpha}{2}\norm{x^{\lr{t+1}} - x^{\lr{t}}}^2. \numberthis\label[ineq]{potential-sc}
	\end{align*} We can also write a lower bound for the left hand side of Inequality~\ref{potential-sc} by evaluating the potential function $\widetilde{\Psi}_t$ at points $x^{\lr{t+1}}$ and $x^{\lr{t}}$:
	\begin{align*}
		\widetilde{\Psi}_t(x^{\lr{t+1}}) - \widetilde{\Psi}_t(x^{\lr{t}}) &= \Phi\lr{x^{\lr{t+1}}} + \eta \sum_{s = 1}^t \inner{\nabla f(\widetilde{x}^{\lr{s}}) }{ x^{\lr{t+1}} } - \Phi(x^{\lr{t}}) - \eta \sum_{s = 1}^t \inner{\nabla f(\widetilde{x}^{\lr{s}})}{ x^{\lr{t}} }\\
															&=  \underbrace{\widetilde{\Psi}_{t-1}(x^{\lr{t+1}}) - \widetilde{\Psi}_{t-1}(x^{\lr{t}})}_{\text{$ \geq 0 $, since $x^{\lr{t}}$ minimizes $\widetilde{\Psi}_{t-1}\lr{x}$}} + \eta \inner{\nabla f(\widetilde{x}^{\lr{t}})}{x^{\lr{t+1}} - x^{\lr{t}}}  \\
															&\geq \eta \inner{\nabla f(\widetilde{x}^{\lr{t}})}{x^{\lr{t+1}} - x^{\lr{t}}}. \numberthis\label[ineq]{potential-fo-opt}
	\end{align*} Reverse and chain Inequalities~\ref{potential-sc} and \ref{potential-fo-opt}, and apply Cauchy-Schwarz inequality to get \[ \frac{\alpha}{2}\norm{x^{\lr{t+1}} - x^{\lr{t}}}^2 \leq \eta \inner{\nabla f(\widetilde{x}^{\lr{t}})}{x^{\lr{t}} - x^{\lr{t+1}}} \leq \eta G \norm{x^{\lr{t}} - x^{\lr{t+1}}}.\numberthis\label[ineq]{almd-chain} \] This shows that \[ \norm{x^{\lr{t}} - x^{\lr{t+1}}} \leq\frac{2\eta G}{\alpha}, \numberthis\label{almd-chain-conc} \] and applying this to the second part of Inequality~\ref{almd-chain} gives \[ \inner{\nabla f(\widetilde{x}^{\lr{t}})}{x^{\lr{t}} - x^{\lr{t+1}}} \leq \frac{2\eta G^2}{\alpha}. \numberthis\label[ineq]{almd-inner-prod-bound} \] We now claim \[\sum_{t= 1}^{T} \inner{\nabla f(\widetilde{x}^{\lr{t}})}{x^{\lr{t}} - x} \leq \sum_{t = 1}^{T} \inner{ \nabla f(\widetilde{x}^{\lr{t}})}{x^{\lr{t}} - x^{\lr{t+1}}} + \tfrac{1}{\eta} ( \Phi(x) - \Phi(x^{\lr{1}}) ). \numberthis\label[ineq]{almd-claim}\] Note that this claim immediately gives the desired error bound; this can be seen as follows: the left-hand side is exactly the term $\mbox{\numcircledtikz{2}}$ in Inequality~\ref{almd-ineq1}; the first term of the right-hand side is bounded in Inequality~\ref{almd-inner-prod-bound}, and the second one is bounded by the definition of set size $D$. Therefore Inequality~\ref{almd-claim} simplifies to 
	\begin{align*}
		\encircle{B} &\leq \frac{2\eta G^2 T}{\alpha} + \frac{D}{\eta}. \numberthis\label[ineq]{almd-circ1-bnd}
	\end{align*}  Combine Inequalities~\ref{almd-circ1-bnd} and \ref{almd-circ2-bnd} with \ref{almd-ineq1}, apply Jensen's inequality, and the fact that $t^*$ is picked uniformly at random from $\{1, 2, \dotsc, T\}$, to get the desired error bound. We now prove Inequality~\ref{almd-claim}. First, we rewrite it as \[\sum_{t = 1}^T \inner{\nabla f(\widetilde{x}^{\lr{t}})}{x^{\lr{t+1}}} + \frac{\Phi(x^{\lr{1}})}{\eta} \leq \sum_{t = 1}^T \inner{\nabla f(\widetilde{x}^{\lr{t}})}{x} + \frac{\Phi(x)}{\eta} .\] The claim is true for $T = 0$ for all $x \in \X$, by the choice of $x^{\lr{1}}$. Assume it holds for all $x \in \X$ at time $ T = t^{\prime} - 1$. Therefore in particular, it holds at the point $x = x^{\lr{t^{\prime}+1}}$. This implies 
    \begin{align*}
\sum_{t= 1}^{t^{\prime}} \inner{ \nabla f(\widetilde{x}^{\lr{t}})}{x^{\lr{t+1}}} + \frac{\Phi(x^{\lr{1}})}{\eta} &= \inner{\nabla f(\widetilde{x}^{\lr{t^{\prime}}})}{x^{\lr{t^{\prime} + 1}}} + \underbrace{ \sum_{t = 1}^{t^{\prime}-1} \inner{\nabla f(\widetilde{x}^{\lr{t}})}{x^{\lr{t+1}}} + \frac{\Phi(x^{\lr{1}})}{\eta}}_{\text{Apply induction hypothesis at $x^{\lr{{t^{\prime}+1}}}$}} \\
				&\leq \inner{\nabla f(\widetilde{x}^{\lr{t^{\prime}}})}{x^{\lr{t^{\prime} + 1}}} + \sum_{t = 1}^{t^{\prime}  - 1} \inner{\nabla f(\widetilde{x}^{\lr{t}})}{x^{\lr{t^{\prime}+1}}} + \frac{\Phi(x^{\lr{t^{\prime}+1}})}{\eta}\\
				&= \sum_{t = 1}^{t^{\prime}} \inner{\nabla f(\widetilde{x}^{\lr{t}})}{x^{\lr{t^{\prime} + 1}}} + \frac{\Phi\lr{x^{\lr{t^{\prime}+1}}}}{\eta}\\
				&= \frac{1}{\eta}\widetilde{\Psi}_{t^{\prime}}\lr{x^{\lr{{t^{\prime}+1}}}} \\
				&\leq \frac{1}{\eta}\widetilde{\Psi}_{t^{\prime}}(x) \\
    		&= \sum_{t = 1}^{t^{\prime}} \inner{\nabla f\lr{\widetilde{x}^{\lr{t}}}}{x} + \frac{\Phi(x)}{\eta},
	\end{align*} where the last inequality is by optimality of $x^{\lr{t^{\prime}+1}}$ in minimizing $\widetilde{\Psi}_{t^{\prime}}$. This completes the induction, and therefore proves Inequality~\ref{almd-claim}, thus completing the proof of the error bound. 
\end{proof}

\section{Analysis of the Arora-Kale Algorithm}\label{pf-currparams} In this section, we display Algorithm~\ref{alg-curr} in the approximate mirror descent framework and provide its analysis. In Section~\ref{sec-ak-params}, we derive the values of all parameters; in Section~\ref{pf-ana-algcurr}, we derive the computational costs of the main steps. We then conclude with the correctness and cost of their algorithm. The main export of this section is the following theorem. 
\begin{theorem}[Run Time \cite{AroraKale}]\label{curr-alg-cost}
Given $C \in \reals^{n \times n}$ with $m \geq n$ non-zero entries and  $0<\varepsilon\leq\tfrac{1}{2}$, we can find, in time $\tilde{\Ord}\left(m/\varepsilon^{5}\right)$, a matrix $Y \in \symm^n$ with $\Ord(m)$ non-zero entries and a diagonal matrix $S\in \reals^{n\times n}$ such that $\widetilde{X}^*=   S  \cdot \tfrac{K\exp \lr{Y}}{\Tr \exp\lr{Y}} \cdot S$ satisfies $\widetilde{X}^*\succeq 0$, $\widetilde{X}^*_{ii}\leq 1$ for all $i \in [n]$, and $\E (C\bullet \widetilde{X}^*)\geq  C\bullet X^* -\varepsilon \cdot \sum_{i,j}\abs{C}_{ij}$.
\end{theorem}

\begin{algorithm}[!ht]
\caption{Reinterpreting \cite{AroraKale}}\label{alg-curr}
\KwIn{Cost matrix $C \in \reals^{n \times n}$, accuracy parameter $\varepsilon$.}
{\bf{Parameters}}: $T = 256 \log n/\varepsilon^2$,  $T^{\prime} = 10240 \log n/\varepsilon^2$, $T^{\prime\prime} = (1/\varepsilon)\cdot 64 \log n$,  $\eta = \varepsilon/64$. Set $\widehat{C}$ and $\rho$ as defined in Lemma~\ref{lem-dual-var-range}. \;

Initialize $Y^{(1)} \gets \zeros$.\;  

Define $\nabla f(M) \defeq \diag \ones_{M\geq \rho} - \widehat{C}$.\;

\For{$t = 1 \KwTo T$}{
    $\widetilde{\exp}\lr{\tfrac{1}{2}Y^{\lr{t}}} \gets  \poly\left( \tfrac{1}{2} Y^{\lr{t}}, T^{\prime\prime} \right)$. \Comment{Approximate matrix exponential}\;
    
    $\widehat{\exp} Y^{\lr{t}} \gets \jl\left(\widetilde{\exp}\left(\tfrac{1}{2}Y^{\lr{t}}\right), T^{\prime} \right)$. \Comment{Approximate projection}\;
    
    $\widetilde{X}^{\lr{t}}\gets n \tfrac{\widehat{\exp}\left(Y^{\lr{t}}\right)}{\Tr \widehat{\exp} Y^{\lr{t}} }$ \Comment{Scaling due to the trace constraint}\;
    
	$Y^{\lr{t+1}}\gets Y^{\lr{t}} - \eta \nabla f(\widetilde{X}^{\lr{t}})$.	 \Comment{Gradient update.}\;
}
For $t^* \stackrel{\text{unif.}}{\sim} \left\{1, 2, \dotsc, T\right\}$, return $Y^{\lr{t^*}}$ and $S$, where  $S$ is from Lemma~\ref{lem-dual-var-range}.\;
\end{algorithm} 

\subsection{Parameters}\label{sec-ak-params}
As can be seen in Algorithm~\ref{alg-almd}, approximate lazy mirror descent requires five parameters: the set diameter, Lipschitz constant of the objective, strong convexity of the mirror map, step size, and number of iterations. The first three depend on our choice of mirror map $\Phi$ and objective $f$. The last two can be chosen based on these parameters and Inequality~\ref{errbnd-almd}. 
\begin{restatable}[Set Diameter]{lemma}{currsetdiam}\label{lem_curr_set_diam}
Given $\Phi\lr{X} =  X\bullet \log X$ and the domain $\{X: X\succeq 0, \Tr X = n\}$, the set diameter measured by $\Phi$ is given by $D =  n \log n$. 
\end{restatable}

\begin{restatable}[Lipschitz constant]{lemma}{currfunclips}\label{lem_curr_alg_lips}
The problem objective $\widehat{f}(X) = -\widehat{C}\bullet X + \sum_{i=1}^n (X_{ii} - \rho_i)^{+}$ (recall that $\rho_i = \sum_{j = 1}^n \abs{C_{ij}}$) is $2$-Lipschitz in the nuclear norm. Recall that nuclear norm of a matrix is the sum of its singular values. 
\end{restatable} 
\begin{proof}\label{pf_curr_alg_lips} 
The gradient of the objective at point $X$ is $\nabla \widehat{f}(X) = \diag(\ones_{\left\{X \geq \rho \right\} })  - \widehat{C}$. By the Gershgorin Disk Theorem, we have \[ \opnorm{\diag\lr{\tfrac{1}{\rho}} C} \leq \max\limits_{i \in [n]} \left( \frac{1}{\rho_i}\cdot \abs{C_{ii}} + \frac{1}{\rho_i} \cdot \sum_{j\neq i} \abs{C_{ij}}  \right) = \max\limits_{i \in [n]} \left( \frac{1}{\rho_i} \cdot \sum\limits_{j = 1}^n \abs{C_{ij}}\right) = 1,\numberthis\label{lips-opnorm-chat}\] where in the last equality we use the choice of $\rho_i = \sum_{j = 1}^n \abs{C_{ij}}$. Since the matrices $\diag(\sfrac{1}{\rho})\cdot C$ and $\widehat{C} = \diag(\sfrac{1}{\sqrt{\rho}})\cdot C\cdot \diag(\sfrac{1}{\sqrt{\rho}})$ are similar, they have the same set of eigenvalues (and therefore, the same operator norm). Therefore \[ \opnorm{\diag(\ones_{\left\{X \geq \rho\right\}}) - \widehat{C}} \leq  \opnorm{\diag(\ones_{\left\{X \geq \rho\right\}}) } + \opnorm{\widehat{C}} = 1 + 1 = 2.\] When we have $\norm{\nabla \widehat{f}} \leq G$ for some $G$, it implies $f$ is $G$-Lipschitz in $\norm{{}\cdot{}}_*$ (the dual norm). Therefore, in our case, we have that $\widehat{f}$ is $2$-Lipschitz in the nuclear norm (dual of the operator norm).
\end{proof}

\begin{restatable}[Strong Convexity]{lemma}{currfunsc}(\cite{KakadeDuality})\label{neg-entropy-sc}
The mirror map $\Phi\lr{X} = X\bullet \log X$ is $1/(2n)$-strongly convex with respect to the nuclear norm on the domain $\left\{X \in \symm^n : X \succeq 0, \Tr\lr{X} = n\right\}$. 
\end{restatable} 

\begin{lemma}\label{lem-eta-algcurr}
	Choosing $\eta = \varepsilon/64$ and $T = 256 \log n/\varepsilon^2$ in Algorithm~\ref{alg-curr} gives an accuracy of $\varepsilon n$. 
\end{lemma}
\begin{proof} We show in Lemma~\ref{lem-curr-dist-after-jl-poly} that Algorithm~\ref{alg-curr} maintains the invariant $\E\dnorm{X^{\lr{t}} - \widetilde{X}^{\lr{t}}}\leq \delta = \varepsilon n/4$. Therefore we are in the framework of approximate lazy mirror descent and can use its error bound from Inequality~\ref{errbnd-almd} and bound it by $\varepsilon K$. We plug in the parameters from Lemmas~\ref{lem_curr_set_diam}, \ref{lem_curr_alg_lips}, and \ref{neg-entropy-sc} in the bound and get
	\begin{align*}
		\E f(\widetilde{x}^{\lr{t^*}} )-f(x^{*})\leq \frac{n \log n}{T \eta} + \frac{2 \eta \cdot 2^2}{1/2n} + \lr{ \frac{\varepsilon n}{4}}\cdot 2. 
	\end{align*} We optimize for $\eta$ by setting the first two terms equal, and get \[ \eta = \tfrac{1}{4} \sqrt{\frac{\log n}{T}}.\numberthis\label{eta}\] With this expression for $\eta$, setting the bound for the right-hand side above to be $\varepsilon n$ gives $T \geq 256 \log n/\varepsilon^2$; plug this back in Equation~\ref{eta} to get $\eta = \varepsilon/64$. 
\end{proof}

\subsection{Computational Cost}\label{pf-ana-algcurr}
From Algorithm~\ref{alg-curr}, we see that there are three main parts to be computed to get the overall cost of the Arora-Kale algorithm: the number of iterations, the number of JL projections per iteration, and the cost of approximating a matrix exponential and multiplying it with a vector. We derive these values in this section.  
\subsubsection{Taylor Approximation for Matrix Exponential}\label{algcurr-cost-polyexp}
In Algorithm~\ref{alg-curr}, before we do the randomized projection to get the diagonal entries, we approximate the matrix exponential $\widetilde{\exp}\left(Y^{\lr{t}}/2\right) = \poly\lr{Y^{\lr{t}}/2, T^{\prime\prime}}$. Here we show a bound on $\abs{  \frac{  \exp\lr{Y^{\lr{t}}}_{ii}}{\Tr \exp\lr{Y^{\lr{t}}}}  -  \frac{ \widetilde{\exp}\lr{Y^{\lr{t}}}_{ii}}{\Tr \widetilde{\exp}\lr{Y^{\lr{t}}}} }$ for any $1 \leq i \leq n$. We do so by first proving a bound on $\abs{\frac{A_{ii}}{\Tr A} -\frac{B_{ii}}{\Tr B}}$ for a matrix $B$ approximating the general matrix $A$; then we prove a general result on the number of terms required to approximate a matrix exponential using Taylor series; finally, we combine these results to get an appropriate choice of $T_{\textrm{poly}}$ for approximating $\exp \lr{Y^{\lr{t}}/2}$.
\begin{restatable}{lemma}{currcostpolyexp}\label{algcurr-taylor-cost-general}
Given positive definite matrices $A$ and $B$ such that $\opnorm{A - B} \leq \varepsilon$, where $\varepsilon\leq \frac{1}{2n} \Tr A$, we have $\abs{ \frac{A_{ii}}{\Tr A} - \frac{B_{ii}}{\Tr B} } \leq 2\frac{\varepsilon \lr{ \Tr A +  n A_{ii}}}{\lr{\Tr A}^2}$.
\end{restatable}

\begin{proof}\label{pf-detailed-currcost-explem} We have the following chain of inequalities.
\[\abs{ \frac{B_{ii}}{\Tr B} - \frac{A_{ii}}{\Tr A} } \stackrel{\mbox{\numcircledtikz{1}}}{\leq} \abs{ \frac{A_{ii} + {\varepsilon}}{\Tr A - n {\varepsilon}} - \frac{A_{ii}}{\Tr A}  } = \frac{\varepsilon\lr{ \Tr A +  n A_{ii}}}{\lr{\Tr A} \lr{\Tr A - \varepsilon n}} \stackrel{\mbox{\numcircledtikz{2}}}{\leq} 2 \frac{\varepsilon\lr{ \Tr A +  n A_{ii}}}{\lr{\Tr A}^2},\] where $\mbox{\numcircledtikz{1}}$ is by the worst case values for $B_{ii}$ from the operator norm bound, and $\mbox{\numcircledtikz{2}}$ is by the bound on $\varepsilon$. 
\end{proof}

\begin{lemma}\label{num-taylor-exp-terms}
For $T \geq e^2 \opnorm{Y}$, we have $\opnorm{\exp\lr{Y} - \sum\limits_{j = 0}^T \frac{Y^j}{j!}} \leq \exp\lr{-T}$. 
\end{lemma}
\begin{proof}
We have the following chain: \[\opnorm{\exp \lr{Y} - \sum_{j = 0}^{T} \tfrac{1}{j!} Y^j} \stackrel{\mbox{\numcircledtikz{1}}}{=} \opnorm{\sum_{j = T+1}^{\infty} \tfrac{1}{j!} Y^j} \stackrel{\mbox{\numcircledtikz{2}}}{\leq} \sum_{j = T+1}^{\infty} \opnorm{\tfrac{1}{j!}Y^j}= \sum_{j = T+1}^{\infty} \tfrac{1}{j!}\opnorm{Y}^j \stackrel{\mbox{\numcircledtikz{3}}}{\leq} \sum_{j = T+1}^{\infty} \tfrac{e^j}{j^j}\opnorm{Y}^j,\numberthis\label[ineq]{generic-exp-taylor-terms-1}\] where $\mbox{\numcircledtikz{1}}$ is by the Taylor series expansion of the matrix exponential, $\mbox{\numcircledtikz{2}}$ is by triangle inequality of norms, and $\mbox{\numcircledtikz{3}}$ is by Stirling's approximation, $j! \geq \lr{ j/e }^j$. Since the right hand side of the above inequality is indexed over $j \geq T\geq e^2 \opnorm{Y}$,  we can bound it further to get \[	\opnorm{\exp Y - \sum_{j = 0}^{T} \tfrac{1}{j!}Y^j}  \leq \sum_{j = T+1}^{\infty} e^{-j} = \frac{(e^{-1})^{T+1}}{1 - e^{-1}} \leq e^{-T}.\] 
\end{proof}

\begin{restatable}{lemma}{currcostTaylorexpmain}\label{algcurr-taylor-cost-alg}
In Algorithm~\ref{alg-curr}, for $n \geq 2$ and $\varepsilon \leq \tfrac{1}{2}$, set $T_{\textrm{poly}} = \frac{64 \log n}{\varepsilon}$, and let $\widetilde{\exp}\left(Y^{\lr{t}}/2\right) := \poly\left( Y^{\lr{t}}/2, T_{\textrm{poly}} \right)$. Then for each coordinate $i$, we have $\abs{\frac{\exp\left(Y^{\lr{t}}\right)_{ii}}{\Tr \exp Y^{\lr{t}}} - \frac{\left(\widetilde{\exp} Y^{\lr{t}}\right)_{ii}}{\Tr  \widetilde{\exp} Y^{\lr{t}}}} \leq \frac{\varepsilon}{8n}$. 
\end{restatable}
\begin{proof}Let  $\widetilde{\exp}\lr{Y^{\lr{t}}/2} = \exp\lr{Y^{\lr{t}}/2} + \Delta$, and $\opnorm{\Delta} = \varepsilon_1$.  Then 
	\begin{align*}
		\opnorm{\exp Y^{\lr{t}} - \widetilde{\exp} Y^{\lr{t}}} &= \opnorm{(\exp(Y^{\lr{t}}/2))^2 - (\widetilde{\exp}(Y^{\lr{t}}/2))^2}\\
								             &= \opnorm{\Delta^2 + \Delta \exp(Y^{\lr{t}}/2)+ \exp(Y^{\lr{t}}/2) \Delta}\\
												&\leq \varepsilon_1^2 + 2\varepsilon_1 \opnorm{\exp(Y^{\lr{t}}/2)}.\numberthis\label[ineq]{AK-diff-opnorm-exp}
	\end{align*} 
	
Observe that in each iteration of Algorithm~\ref{alg-curr}, we add $-\eta \nabla f(\widetilde{X}^{\lr{t}})$ to the current $Y^{\lr{t}}$ in the gradient step; therefore at the end of all the $T$ iterations, $\opnorm{Y^{\lr{t}}} \leq \abs{\eta T} \| \nabla f(\widetilde{X}^{\lr{t}})\|_{\mathrm{op}}$. From the values of $\eta$ and $T$ as set in Algorithm~\ref{alg-curr} (and explained in Section~\ref{pf-currparams}), the worst-case value is \[\opnorm{Y^{\lr{t}}/2} \leq \frac{1}{2}\cdot \frac{\varepsilon}{64}\cdot \frac{256 \log n}{\varepsilon^2}\cdot 2 = \frac{ 4 \log n}{\varepsilon}.\numberthis\label[ineq]{opnormY}\] Next, from Lemma~\ref{num-taylor-exp-terms}, we require the first $\max\left\{e^2 \opnorm{Y^{\lr{t}}/2}, \log \left( 1/\varepsilon_1 \right) \right\}$ terms of  the Taylor series of $\exp\lr{Y^{\lr{t}}/2}$ to get an $\varepsilon_1$ accuracy in approximation. Since $T_{\mathrm{poly}} = 64 \log n/\varepsilon  \geq e^2 \opnorm{Y^{\lr{t}}/2}$ (from Inequality~\ref{opnormY}), this choice of number of Taylor series terms corresponds to an accuracy of $\varepsilon_1 = n^{-64/\varepsilon}$. From Inequality~\ref{opnormY}, we get that \[\|\exp(Y^{\lr{t}}/2)\|_{\mathrm{op}} \leq e^{4 \log n/\varepsilon} = n^{4/\varepsilon}. \numberthis\label[ineq]{AK-opnormexpY}\] Then we have
\begin{align*}
	\varepsilon_1^2 + 2\varepsilon_1 \|\exp (Y^{\lr{t}}/2)\|_{\mathrm{op}} &\leq n^{-128/\varepsilon} + 2n^{-64/\varepsilon}n^{4/\varepsilon} \\
					&\leq 4n^{-60/\varepsilon}\\					
					&\leq \frac{n^{-4/\varepsilon}}{2} \leq \frac{1}{2n} \Tr \exp(Y^{\lr{t}}/2),\numberthis\label[ineq]{AK-diff-opnorm-exp-2}
\end{align*} where the last inequality is by Inequality~\ref{AK-opnormexpY}. Chaining Inequalities~\ref{AK-diff-opnorm-exp} and \ref{AK-diff-opnorm-exp-2}, the condition in Lemma~\ref{algcurr-taylor-cost-general} is satisfied. Applying the result of Lemma~\ref{algcurr-taylor-cost-general},
\begin{align*}
	\abs{\frac{\left(\exp \lr{Y^{\lr{t}}}\right)_{ii}}{\Tr \exp \lr{Y^{\lr{t}}}} - \frac{\left( \widetilde{\exp} \lr{Y^{\lr{t}}} \right)_{ii}}{\Tr \widetilde{\exp} \lr{Y^{\lr{t}}}}} &\leq 2 \lr{\varepsilon_1^2 + 2\varepsilon_1 \opnorm{\exp\lr{Y}}} \frac{\Tr \exp\lr{Y^{\lr{t}}} + n \exp \lr{Y^{\lr{t}}}_{ii}}{\lr{\Tr \exp\lr{Y^{\lr{t}}}}^2}. \\
																					&\leq 2\frac{\lr{\varepsilon_1^2 + 2\varepsilon_1 n^{4/\varepsilon}} \lr{2n^{1 + 8/\varepsilon}}}{\lr{n^{-8/\varepsilon}}^2}\\								&\leq 4\lr{\frac{\varepsilon^2 }{10000 n^{41/\varepsilon}} + \frac{\varepsilon}{50 n^{4/\varepsilon}} }\\							&\leq \frac{\varepsilon}{8n}
\end{align*} 
\end{proof}

\subsubsection{Randomized Projections}\label{algcurr-cost-jl}
Suppose we approximate each entry of a vector using randomized projections. Then we can state the following result about the accuracy of the function $g(x) = x_i/\norm{x}_1$.
\begin{lemma}\label{algcurr-jl-cost}
For $0 \neq X \in \symm^n$, let $\widetilde{X} = \jl(X, 10240 \log n/\varepsilon^2 )$. Then $\abs{ \frac{ \widetilde{X}_{ii}}{\Tr \widetilde{X}} - \frac{ X^2_{ii}}{\Tr X^2}  } \leq \frac{\varepsilon }{8}$. 
\end{lemma}
To prove this result, we need the Johnson-Lindenstrauss lemma. 
\begin{lemma}[\cite{jl}]\label{thm-jl-statement}
For any $0 < \varepsilon < 1$, and any integer $n$, let $k$ be a positive integer such that $k \geq 20 \log n/\varepsilon^2$. Then for any set $V$ of $n$ points in $\reals^d$ and random matrix $A \in \reals^{k \times d}$, with high probability, for all $x \in V$, \[ ( 1- \varepsilon) \norm{x}_2^2 \leq \norm{(1/\sqrt{k})Ax}_2^2 \leq ( 1 + \varepsilon) \norm{x}_2^2.\] 
\end{lemma}
\begin{proof}[Proof of Lemma~\ref{algcurr-jl-cost}]
Applying Lemma~\ref{thm-jl-statement} to $\widetilde{X} = \jl\lr{X, \frac{10240\log n}{\varepsilon^2}}$, we have that with high probability, $\abs{X^2_{ii} - \widetilde{X}_{ii}} \leq \frac{\varepsilon}{32} \abs{X^2_{ii}}$. Therefore, $ \Tr X^2 \lr{ 1 - \tfrac{\varepsilon}{32}} \leq \Tr \widetilde{X}^2 \leq  \Tr X^2 \lr{ 1 + \tfrac{\varepsilon}{32}}$. Therefore $\frac{X^2_{ii}\lr{1 - \varepsilon/32}}{\Tr X^2\lr{1 +\varepsilon/32}} \leq \frac{\widetilde{X}_{ii}}{\Tr \widetilde{X}} \leq \frac{X^2_{ii}\lr{1 + \varepsilon/32}}{\Tr X^2\lr{1 - \varepsilon/32}}$ which can  be simplified to $\frac{X^2_{ii}}{\Tr X^2}\lr{1 - \varepsilon/16} \leq \frac{\widetilde{X}_{ii}}{\Tr \widetilde{X}} \leq \frac{X^2_{ii}}{\Tr X^2}\lr{1 + \varepsilon/8}$, where the last simplification is by the inequalities $\frac{1-x}{1+x} \geq 1 - 2x$ and $\frac{1+x}{1-x} \leq 1 + 4x$ for $x \in \lr{0, \tfrac{1}{2}}$. Therefore we have that $\abs{\frac{\widetilde{X}_{ii}}{\Tr \widetilde{X}} - \frac{X^2_{ii}}{\Tr X^2} } \leq \lr{\varepsilon/8} \frac{X^2_{ii}}{\Tr X^2}\leq \varepsilon/8$.  
\end{proof}

\subsubsection{Number of Iterations}\label{algcurr-app-true-dist} From  Lemmas~\ref{algcurr-taylor-cost-alg} and \ref{algcurr-jl-cost} proved above, we can infer that the choice of $T^{\prime\prime}$ and $T^{\prime}$ in Algorithm~\ref{alg-curr} gives us the following overall error in approximating the true primal iterate. 
\begin{restatable}{lemma}{lemcurrdistafterjlpoly}\label{lem-curr-dist-after-jl-poly}
In Algorithm~\ref{alg-curr}, we have that $\dnorm{\widetilde{X}^{\lr{t}} - X^{\lr{t}}} \leq \frac{\varepsilon n}{4}$. 
\end{restatable}
\begin{proof}
The quantity we want to bound is $\dnorm{\frac{n\exp(Y^{\lr{t}})}{\Tr \exp(Y^{\lr{t}})} - \frac{ \widetilde{X}^{\lr{t}}}{\Tr \widetilde{X}^{\lr{t}}}}$. Each term is bounded as: \[\abs{\frac{n \exp(Y^{\lr{t}})_{ii}}{\Tr \exp\lr{Y^{\lr{t}}}} - \frac{ \widetilde{X}^{\lr{t}}_{ii}}{\Tr \widetilde{X}^{\lr{t}}}} \leq n \underbrace{\abs{\frac{\exp\lr{Y^{\lr{t}}}_{ii}}{\Tr \exp\lr{Y^{\lr{t}}}} - \frac{\widetilde{\exp}\lr{Y^{\lr{t}}}_{ii}}{\Tr \widetilde{\exp}\lr{Y^{t}}}}}_{\text{$\poly$ error}} + \underbrace{\abs{ \frac{ n \widetilde{\exp}\lr{Y^{\lr{t}}}_{ii}}{\Tr \widetilde{\exp}\lr{Y^{\lr{t}}}} - \frac{n \widehat{\exp}\lr{Y^{\lr{t}}}_{ii}}{\Tr \widehat{\exp}\lr{Y^{\lr{t}}}}  }}_{\text{$\jl$ error}}. \] Apply the results of Lemmas~\ref{algcurr-taylor-cost-alg} and \ref{algcurr-jl-cost} to the right hand side terms. 
\end{proof}

\begin{corollary}
The number of iterations for convergence of the Arora-Kale algorithm is $\Ord(1/\varepsilon^2)$.
\end{corollary}
\begin{proof} Since the Arora-Kale algorithm only depends on the diagonal entries of $X$, we can assume that $\widetilde{X}$ and $X$ match on the off-diagonal entries. Then, $\dnorm{\widetilde{X}^{\lr{t}} - X^{\lr{t}}} \leq \frac{\varepsilon n}{4}$ is exactly equivalent to $\|\widetilde{X}^{\lr{t}} - X^{\lr{t}}\|_{\mathrm{nuc}} \leq \frac{\varepsilon n}{4}$. Therefore the algorithm meets the conditions of Algorithm~\ref{alg-almd} with $\delta = \tfrac{\varepsilon n}{4}$. Therefore by Theorem~\ref{thm-almd}, the number of outer iterations required for convergence is $\Ord(1/\varepsilon^2)$.
\end{proof}

\subsubsection{Combining All the Costs}
Recall from Algorithm~\ref{alg-curr} that $T^{\prime} = \Ord(1/\varepsilon^2)$, $T^{\prime\prime} = \Ord(1/\varepsilon)$, and the number of iterations is $\Ord(1/\varepsilon^2)$. The cost of a matrix-vector product is $\Ord(m)$. Therefore, multiplying these costs gives $\Ord(m/\varepsilon^5)$, the claimed cost of Arora-Kale algorithm. This completes the analysis.

\section{Analysis of our Proposed Algorithm}\label{pf-proposed-approach}  We now analyze Algorithm~\ref{alg-new}, organizing this section as follows. In Section~\ref{sec-new-params} we derive the values of parameters that appear in the error bounds. Next, in Section~\ref{mat-exp-cheby-approx}, we show how we construct a polynomial to approximate the matrix exponential. In Section~\ref{new-estimator-pfs}, we prove properties of the constructed estimators. We derive the number of inner iterations we have in Section~\ref{new-num-inner-iters}. In Section~\ref{sec-algnew-distbndinalg}, we establish the crucial distance invariance between true and estimated iterates, which ensures that our error is always under control.  We next show in Section~\ref{pf-proj-in-exp-int} why we do not need to normalize our projection step, which enables us to have a simple projection. Finally, we derive the error bound in Section~\ref{sec-app-errorbound}. 

\subsection{Parameters of Mirror Map} \label{sec-new-params}
As before, there are two parameters of the mirror map that we need to use in Theorem~\ref{thm-almd}: the diameter of the constraint set as measured by it, and its strong convexity parameter.  
\begin{restatable}[Set Diameter]{lemma}{newsetdiam}\label{lem_new_set_diam}
Given $\Phi(X) = X \bullet \log X - \Tr X$ and the domain $\D = \{X: X\succeq 0, \Tr X \leq K\}$, where $K \geq n$, the set diameter measured by $\Phi$ is given by $D =  K \log K$. 
\end{restatable}

\newalgstrongconvexity*

To prove the claimed strong convexity, we need the following tools. 
\begin{definition} A function $f: \reals^n \rightarrow \reals$ is $L$-smooth in norm $\norm{{}\cdot{}}$ if it is continuously differentiable and satisfies $\norm{\nabla f(x) - \nabla f(y)}_* \leq L\norm{x - y}$ for all $x$ and $y$ in $\dom f$. 
\end{definition}
For functions on symmetric matrices, we use the following equivalent definition of smoothness. 
\begin{definition}\label{matrix-smoothness} A function $f: \symm^n \rightarrow \reals$ is $L$-smooth in $\norm{{}\cdot{}}$ if and only if for $h: \reals \rightarrow \reals$ defined as $h\lr{t} = f(X + tH)$ for $H\in\symm^n$ such that $X+tH \in \dom(f)$, we have $h^{\prime\prime}\lr{0}\leq L \norm{H}^2$.
\end{definition}

\begin{theorem}[\cite{KakadeDuality}]\label{smooth-strong-duality}
Assume that $f$ is a closed and convex function. Then $f$ is $\beta$-strongly convex with respect to a norm $\norm{{}\cdot{}}$ if and only if its Fenchel dual, $f^*$, is $\frac{1}{\beta}$-smooth with respect to the dual norm $\norm{{}\cdot{}}_{*}$. 
\end{theorem} 

\begin{theorem}[\cite{Judit-Nemirov}]\label{judit-nemirov}
Let $\Delta$ be an open interval on the real axis, and $f$ be a twice differentiable function on $\Delta$ satisfying, for a certain $\theta \in \reals$, for all $a< b$, where $a, b \in \Delta$, $\frac{f^{\prime}(b) - f^{\prime}(a)}{b - a} \leq \theta \frac{f^{\prime\prime}(a) + f^{\prime\prime}(b)}{2}$. Let $\X_n(\Delta)$ be the set of all $n \times n$ symmetric matrices with eigenvalues belonging to $\Delta$. Then for $X \in \X_n(\Delta)$,  the function $F(X) = \Tr f(X)$ is twice differentiable, and for every $H \in \symm^n$, we have $D^2 F(X)[H, H] \leq \theta \Tr (H f^{\prime\prime}(X) H)$.
\end{theorem}

\begin{theorem}[\cite{lewis95}]\label{lewis95}
Suppose that the function $f: \reals^n \rightarrow \reals$ is symmetric (that is, $f(\sigma x) = f(x)$ for all $x\in \dom f$ and all permutations $\sigma$). Then if  $f$ is convex and lower semicontinuous, the corresponding unitarily invariant function $f\circ \lambda$ is convex and lower semicontinuous on $\reals^{n\times n}$ 
\end{theorem}

For our proof, we use definitions from Definition~\ref{defs-sc} in the following way. We first show that $\Psi$ satisfies \[\Psi^*(Y) = \Phi(Y), \,\text{on } \{Y: Y\succeq 0, \Tr Y \leq K\},\numberthis\label[eq]{wts-conj}\]where $\Phi(Y) = Y \bullet \log Y - \Tr Y$ is the mirror map, as defined in the statement of the lemma. We then prove that $\Psi$ is $\beta$-smooth with respect to the operator norm for a certain $\beta > 0$. Theorem~\ref{smooth-strong-duality} then immediately implies $1/\beta$-strong convexity of $\Psi^*$ with respect to the nuclear norm. Then Equation~\ref{wts-conj} implies that $\Phi$  is $1/\beta$-strongly convex with respect to the nuclear norm on the domain $\{Y: Y\succeq 0, \Tr Y \leq K\}$, which is to be proved. We accomplish our first goal (Equation~\ref{wts-conj}) in the following sequence of steps. 

Claim~\ref{cl1} proves that the function $\psi$ and its matrix version, $\Psi$, are both continuously differentiable at the boundary of definition of the two pieces. Claim~\ref{cl2} then proves that $\psi_1$ and $\psi_2$ are convex; in conjuncation with Claim~\ref{cl1}, this implies  $\psi$ is convex. Applying Theorem~\ref{lewis95} extends the property of convexity to $\Psi$. Claim~\ref{cl3} proves that the vector functions $\psi$ and $\phi$ satisfy $\psi_1^*(x) = \phi(x)$ for $x\in \reals^n_{+}$. Claim~\ref{cl4}  proves that given an input point $x \in \{x:x_i\geq 0, \sum_{i=1}^n x_i \leq K\}$, the point $y$ which attains the optimum in computing $\psi_1^*(x)$ lies in the \emph{interior} of the set $\{y: \psi_1(y)\leq 2K \}$. Claim~\ref{cl5} shows that $\psi^*(x) = \psi_1^*(x)$ for $x \in \{x: x_i \geq 0, \sum_{i=1}^n x_i \leq K\}$. This is obtained by combining the results of Claim~\ref{cl2} and \ref{cl4}. 

We then use these results as follows: since on the set $\{x: x_i\geq 0, \sum_{i=1}^n x_i \leq K\}$, we have $\psi^* = \phi$, this implies $\Psi^* = \Phi$ on the corresponding set, $\{X: X\succeq 0, \Tr X \leq K\}$. Next, to show smoothness of $\Psi$, we use Theorem~\ref{judit-nemirov} to compute the smoothness constants of each part of $\Psi$ (in Claims~\ref{cl6} and \ref{cl7}), and then combine with continuous differentiability at the boundary (from Claim~\ref{cl1}) to get the overall smoothness constant of $\Psi$. By the argument at the start of this proof, this immediately proves the desired strong convexity parameter.  We now proceed to prove all the claims aluded to above. 

\begin{claim}\label{cl1}
The functions $\Psi$ and $\psi$ are both continuously differentiable at the boundary.
\end{claim}
\begin{proof}[Proof of Claim] One can check that $\psi_1(y) = \psi_2(y)$ at the boundary. This implies continuity of the function $\psi$. The derivatives of the two functions are also the same at the boundary. The $i$-th component of the gradient is given by $\nabla_i \psi_2(y) = \frac{2K\nabla_i \psi_1(y)}{\psi_1(y)}$. At the boundary of the two parts of the function, we have $\psi_1(y) = 2K$. Substituting this into the above gradient gives $\nabla \psi_2(y) = \nabla \psi_1(y)$. This shows that $\psi$ is continuously differentiable at the boundary. We only used chain rule of derivatives here, which applies to matrices as well, so the exactly same reasoning also gives that $\Psi$ is continuously differentiable at the boundary. 
\end{proof}

\begin{claim}\label{cl2}
The functions $\psi$ and $\Psi$ are convex on their domains. 
\end{claim}
\begin{proof} The function $\psi$ is a piecewise function, each piece composed of a standard convex function (see \cite{BVbook}). Combine with continuous differentiability from Claim~\ref{cl1} gives convexity of $\psi$. Applying Theorem~\ref{lewis95} implies convexity of $\Psi$. 
 \end{proof}
  
\begin{claim}\label{cl3}For any input $x \in \reals^n_{+}$, we have $\psi_1^*(x) = \phi(x)$. 
\end{claim}
\begin{proof}[Proof of Claim]\label{pf-claim-3}
By definition, we have $\psi_1^*(x) = \sup_y (x^\top y - \sum_{i=1}^n \exp(y_i))$. Observe that the domain of $\psi_1^*$ is $\reals^n_{+}$ (because if there exists an input with a negative coordinate, then the corresponding coordinate of the maximizer $y^*$ can be made to go to $-\infty$). Therefore, given an input $x \in \reals^n_{+}$, the  supremum is attained at $y^*$ satisfying $x_i = \exp(y^*_i)$. This means the maximizer is $y_i^* = \log x_i$.  Therefore the conjugate is $\psi_1^*(x) = \sum_{i=1}^n x_i \log x_i - \sum_{i=1}^n x_i = \phi(x)$. 
\end{proof}

\begin{claim}\label{cl4}
For any $x$ in the set $\left\{x: x_i \geq 0, \sum_{i = 1}^n x_i \leq K\right\}$, the point $y^*=\argmax 
\lr{x^T y - \psi_1(y)}$ lies in  $\mbox{\bf{int}}\left\{ y: \psi_1(y) \leq 2 K \right\}$, where $\intr$ denotes the interior of the set. 
\end{claim} 
\begin{proof}[Proof of Claim]
From the proof of Claim~\ref{cl3}, for any $x\in \reals^n_{+}$, we have that $y^* = \argmax \lr{x^T y - \psi_1(y)}$ satisfies $y_i^* = \log x_i$ for $1\leq i\leq n$. In addition to this, the statement of the lemma also requires the input $x$ to satisfy $x_i \geq 0, \sum_{i = 1}^n x_i \leq K$. Plug in the values of $x$ in terms of $y$ in the above inequality to get $\sum_{i =1}^n \exp y_i^*  \leq K$,  which is the same as saying $\psi_1(y^*) \leq K< 2K$. This shows that the optimum, $y^*$, lies in $\mbox{\bf{int}}\left\{ y: \psi_1(y) \leq 2 K\right\}$.
\end{proof}

\begin{claim}\label{cl5}
	We have $\psi^*(x) = \psi_1^*(x)$ on all $x \in \left\{x:x_i \geq 0, \sum_{i = 1}^n x_i \leq K\right\}$. 
\end{claim}
\begin{proof}[Proof of Claim] By definition of conjugate and $\psi$, 
 \begin{align*}
\psi^*(x) &= \sup_y x^T y  - \psi(y) \numberthis\label{def-conj-psi}\\
&=\sup_y x^T y - \twopartdef{\psi_1(y)}{\psi_1(y) \leq 2K}{\psi_2(y)}{otherwise}
\end{align*} From Claim~\ref{cl2}, $\psi$ is convex. Therefore the function to be maximized in Equation~\ref{def-conj-psi} is concave. From Claim~\ref{cl4}, for input $x$ in the set $\left\{x: x_i \geq 0, \sum_{i = 1}^n x_i \leq K\right\}$, we have that the maximizer $\argmax_{y}\lr{ x^T y - \psi_1(y)}$ lies in the interior of $\{y: \psi_1(y) \leq 2K\}$. Therefore for input $x \in \left\{x:x_i \geq 0, \sum_{i = 1}^n x_i \leq K\right\}$, the maximizer of Equation~\ref{def-conj-psi} is also the same as that of $\psi_1^*(x)$. This gives $\psi^*(x) = \psi_1^*(x)$. 
\end{proof}

\begin{claim}\label{cl6} The function $\Psi_1(Y)$ defined over $\left\{ Y: \Tr \exp Y \leq 2 K\right\}$ is $4K$-smooth.
\end{claim}
\begin{proof}[Proof of Claim] Let  $g\lr{u} \defeq \exp(u)$. The function $g$ is convex and differentiable (any number of times). In particular, $g^{\prime\prime}$ is convex. For any $a$, $b$, applying the Mean Value theorem to some point $\zeta \in \lr{a, b}$, convexity of $g^{\prime\prime}$, and $g^{\prime\prime}\geq 0$ (due to convexity of $g$) gives \[	\frac{g^{\prime}\lr{b} - g^{\prime}\lr{a}}{b-a} = g^{\prime\prime}\lr{\zeta} \leq \max\lr{g^{\prime\prime}\lr{a}, g^{\prime\prime}\lr{b}} \leq 2\frac{g^{\prime\prime}\lr{a} + g^{\prime\prime}\lr{b}}{2}.\] This satisfies the right-hand side condition for Theorem~\ref{judit-nemirov} with $\theta = 2$; so Theorem~\ref{judit-nemirov} implies that on the domain $\left\{ Y: \Tr \exp Y  \leq K\right\}$, for $h\lr{t} \defeq \sum\limits_{i= 1}^n g\lr{\lambda_i\lr{Y + tH}} = \Tr \exp(Y + tH)$, we have,
\begin{align*} 
h^{\prime\prime}\lr{0}= D^2\Psi_1(Y)[H, H] &\leq 2 \Tr\lr{H  g^{\prime\prime}(Y)H}\\
						&= 2 \Tr\lr{\exp(Y) H^2}\\
						&\leq 2 \Tr \exp(Y) \cdot \opnorm{H}^2\\
						&\leq 2 \cdot 2K \cdot \opnorm{H}^2\\
						&= 4K \opnorm{H}^2 \numberthis\label[ineq]{smoo-psiI}, 
\end{align*} where we used the domain constraint for $\Psi_1$ in the last inequality, and the fact that matrix exponential is positive semidefinite in the first (H\"older's) inequality. By Definition~\ref{matrix-smoothness} then, we have the lemma. 
\end{proof}

\begin{claim}\label{cl7}
The smoothness constant of $\Psi_2(Y)$ over  the set $\{Y : \Tr \exp Y \geq 2 K\}$ is  $4K$. 
\end{claim}
\begin{proof} 
For ease of exposition, let $a \defeq 2K$.  Consider the same scalar function from Claim~\ref{cl6}, $h\lr{t} = \Tr\exp (Y + tH)$ and $\ell\lr{t} \defeq  a\log \lr{h\lr{t}} + 2K - 2K \log (2K)$. Then $\ell^{\prime}\lr{t} = a \frac{h^{\prime}(t)}{h(t)}$ and $\ell^{\prime\prime}\lr{t} = a \lr{\frac{h^{\prime\prime}(t)}{h(t)} - \lr{\frac{h^{\prime}(t)}{h(t)}}^2 } \leq a \frac{h^{\prime\prime}(t)}{h(t)}$. In particular, \[\ell^{\prime\prime}\lr{0} \leq a \frac{h^{\prime\prime}(0)}{h(0)}. \numberthis\label{ineq-h-ell}\] We already showed in Inequality~\ref{smoo-psiI} that $h^{\prime\prime}\lr{0} = D^2 \Psi_1(Y)[H, H] \leq 4K \opnorm{H}^2$. We also have that $h(0) = \Tr \exp(Y) \geq 2K$ (by assumption of the lemma). Plugging these along with the value of $a$ into Inequality~\ref{ineq-h-ell} gives us $\ell^{\prime\prime}\lr{0} \leq 2K \frac{4K}{2K} \cdot \opnorm{H}^2 = 4K \opnorm{H}^2$. This implies the claimed smoothness constant. 
\end{proof}

\begin{proof}[Proof of Lemma~\ref{lem-new-alg-strong-conv}]
For the functions defined in Definition~\ref{defs-sc}, we can combine Claims~\ref{cl3}  and \ref{cl5} to get that $\psi^*(x) = \phi(x)$ for  $x \in \{x:x_i \geq 0, \sum_{i = 1}^n x_i \leq K \}$. This implies the matrix version of this statement, $\Psi^*(X) = \Phi(X)$ for  $X \in \{X:X \succeq 0, \Tr X \leq K\}$. Next, applying Claims~\ref{cl1}, \ref{cl6}, and \ref{cl7}, we get that the function $\Psi$ is continuously differentiable with smoothness constant $4K$. Invoking Theorem~\ref{smooth-strong-duality}, we immediately obtain that $\Psi^*$ is strongly convex with parameter $\frac{1}{4K}$. This implies that $\Phi$ is strongly convex with the same parameter over the set $\left\{X: X \succeq 0, \Tr X \leq K \right\}$.  
\end{proof}

\subsection{Chebyshev Approximation of the Matrix Exponential}\label{mat-exp-cheby-approx} In this section, we show how to construct a polynomial approximation of our matrix exponential. The standard technique to do so involves truncating the Taylor series of the matrix exponential; however, a quadratically improved bound on the number of terms required for the computation is provided by Sachdeva and Vishnoi \cite{SachdevaVishnoi} using Chebyshev polynomials. We follow their notation and summarize their main results below for the sake of completeness. 

\subsubsection{A Brief Summary of Chebyshev Approximation}
For a non-negative integer $d$, we denote by $T_d(x)$ the Chebyshev polynomials of degree $d$, defined recursively as follows: 
\begin{align*}
	T_0(x) &= 1,\\
	T_1(x) &= x,\\
	T_d(x) &= 2xT_{d-1}(x)  -T_{d-2}(x).
\end{align*} Let $Y_i$ be i.i.d. variables taking values $1$ and $-1$ each with probability $1/2$. Let $D_s = \sum_{i=1}^s Y_i$, $D_0 \defeq 0$, and \[p_{s,d}(x) \defeq \E_{Y_1, Y_2, \dotsc, Y_s} \lr{T_{D_s}(x) \ones_{\abs{D_s}\leq d}}.\numberthis\label{def-p-sd}\] We can use these to construct a polynomial with degree roughly $\sqrt{s}$ that can well approximate $x^s$. The formal statement follows. \begin{theorem}[Theorem $3.3$ in \cite{SachdevaVishnoi}] For any positive integers $s$ and $d$, the degree $d$ polynomial $p_{s, d}$ defined by Equation~\ref{def-p-sd} satisfies \[ \sup_{x \in [-1, 1] } \abs{p_{s, d}(x) - x^s} \leq 2\exp(-d^2/(2s)).\] 
\end{theorem}Using this result, define the polynomial: \[ q_{\lambda, t, d}(x) \defeq \exp(-\lambda) \sum_{i = 0}^t \frac{(-\lambda)^i}{i!} p_{i, d}(x). \numberthis\label{def-q-sd}\]  Then we can use $q$ to approximate an exponential with the following error guarantee.
\begin{lemma}[Lemma $4.2$ of \cite{SachdevaVishnoi}]\label{lem-sv-exp-approx-degree} For every $\lambda > 0$ and $\delta \in (0, 1/2]$, we can choose $t = \max(\lambda, \log(1/\delta))$ and $d = \sqrt{ t\log(1/\delta)}$ such that \[\sup_{x \in [-1, 1]} \abs{\exp(-\lambda-\lambda x) - q_{\lambda, t, d}(x) } \leq \delta.\] 
\end{lemma} This is a quadratic improvement over the standard cost (degree) of approximating an exponential using truncated Taylor series. Finally, this lemma can be used to generalize the approximation from the $[-1, 1]$ interval to the interval $[0, b]$, as stated below.
\begin{theorem}[Theorem $4.1$ of \cite{SachdevaVishnoi}]\label{thm-chebyexp} For every $b>0$, and $0 < \delta \leq 1$, there exists a polynomial $r_{b, \delta}$ that satisfies \[\sup_{x\in [0, b]} \abs{\exp(-x) - r_{b, \delta}(x)} \leq \delta\] and has degree $\Ord(\sqrt{\max(b, \log(1/\delta)) \cdot \log(1/\delta)})$.
\end{theorem} The proof of this theorem uses $\lambda \defeq b/2$, and $t$ and $d$ from Lemma~\ref{lem-sv-exp-approx-degree} and the polynomial \[r_{b, \delta}(x) \defeq q_{\lambda, t, d} \lr{\frac{1}{\lambda}(x-\lambda)}.\numberthis\label{def-r-bdel}\]
\begin{corollary}[Our Chebyshev Approximation]\label{chebyexp-def} For every $b>0$, $a<b$, $0<\delta\leq 1$, and $d = \sqrt{\max\lr{\tfrac{1}{2}(b-a), \log \lr{\tfrac{1}{\delta}} }\log\lr{\tfrac{1}{\delta}}}$, there exists a degree-$d$ polynomial $\cheby(u, d, \delta)$ such that \[\sup_{u\in [a,b]} \abs{\exp(u) - \cheby(u, d, \delta)} \leq \delta \exp(b). \numberthis\label{general-cheb-approx}\] 
\end{corollary}
\begin{proof}
Using a simple linear transformation, Theorem~\ref{thm-chebyexp} generalizes to:
\begin{align*}
	\sup_{z \in [a, b]} \abs{\exp(-\frac{1}{2}(b-a)) \sum_{i=0}^t \frac{(-\tfrac{1}{2}(b-a))^i}{i!} p_{i, d}(\frac{z - (b+a)/2}{(b-a)/2}) - \exp(-(z-a))} &\leq \delta. 
\end{align*} By choosing $\lambda = \tfrac{1}{2}(b-a)$, and transforming $-z + a = u - b$, we get 
\begin{align*}
	\sup_{u \in [a, b]} \abs{q_{\tfrac{1}{2}(b-a), t, d}\lr{\frac{-u + (b+a)/2}{(b-a)/2}} - \exp(u - b)} &\leq  \delta.
\end{align*} Using Equation~\ref{def-r-bdel} above gives 
\begin{align*}
	\sup_{u\in [a,b]} \abs{\exp(b) r_{b-a, \delta}(b-u) - \exp(u)} &\leq \delta \exp(b).
\end{align*}  Therefore, let $d =  \sqrt{\max\lr{\tfrac{1}{2}(b-a), \log \lr{\tfrac{1}{\delta}} }\log\lr{\tfrac{1}{\delta}}}$ and $\cheby(u, d, \delta) =  \exp(b) r_{b-a, \delta}(b-u)$. Substitute these into the last inequality to get the statement of the lemma. 
\end{proof}

\subsubsection{Chebyshev Approximation in Our Algorithm} We can use the above results to approximate a matrix exponential as follows. Observe that 
\begin{align*}
	\opnorm{\exp(Y) - \cheby(Y, d, \delta)} &= \max_{i\in [n]} \abs{ \exp(\lambda_i) - \cheby(\lambda_i, d, \delta)},
\end{align*} where $\lambda_i$ are the eigenvalues of $Y$ and $\cheby$ is the subroutine described in Corollary~\ref{chebyexp-def}. We only need the spectrum of $Y$ in order to complete the approximation, and that is what we proceed to derive below. Once we have the spectrum, we simply combine it with the above results to get Lemma~\ref{algcurr-cheby-cost-alg}. 
 
\begin{restatable}{lemma}{spectrumOfY}\label{spectrum-Y}The spectrum of $Y$ lies in the range $\left[ -\frac{1}{\varepsilon} 60 \log n, \log K\right]$. 
\end{restatable}
\begin{proof}
Since we start Algorithm~\ref{alg-new} with $Y^{(1)} = 0$, at the $t$-th iteration, we have $Y^{(t)} = - \sum_{i=1}^{t-1} \eta \nabla f\lr{X^{(i)}}$. Plugging in the parameters displayed in Table~\ref{TableNewParams}, we get that the total number of iterations of the algorithm is $\Tin \times \Tout = \frac{1}{\varepsilon^3} 24 \times 10^5 \lr{\log (n/\varepsilon)}^{11} \log n$, the Lipschitz constant of the objective function is $\opnorm{\nabla f}\leq 2$, and the step size is $\eta = \frac{\varepsilon^2}{8 \times 10^4 \lr{\log (n/\varepsilon)}^{11}}$. Multiplying these gives \[ \opnorm{Y^{(t)}} \leq 2 \cdot \frac{\varepsilon^2}{8 \times 10^4 \times \lr{\log (n/\varepsilon)}^{11}} \cdot \frac{24 \times 10^5 \times \lr{\log (n/\varepsilon)}^{11} \log n}{\varepsilon^3} = \frac{1}{\varepsilon} 60 \log n.\] Therefore, the spectrum of $Y^{(t)}$ lies in \[ \lambda(Y^{(t)}) \in \left[ -\frac{1}{\varepsilon} 60 \log n, \frac{1}{\varepsilon} 60 \log n \right]. \numberthis\label{specbnd1}\] We now show a better upper bound on the spectrum. Since our algorithm maintains $\Tr X^{(t)} \leq K$ (see Lemma~\ref{lem-newalg-distbnd-opt}), and $X^{(t)} = \exp(Y^{(t)})$, it implies $\Tr \exp(Y^{(t)}) \leq K$. Since the matrix exponential is positive definite, this implies $\opnorm{\exp(Y^{(t)})} \leq K$, and therefore, \[\lambda_{max}(Y^{(t)}) \leq \log K. \numberthis\label{specbnd2}\] Combining the inclusion \ref{specbnd1} and Inequality~\ref{specbnd2} gives the claimed bound on the spectrum. 
\end{proof}

\begin{restatable}{lemma}{currcostChebyexpmain}\label{algcurr-cheby-cost-alg}
In Algorithm~\ref{alg-curr}, for $n \geq 2$ and $\varepsilon \leq \tfrac{1}{2}$, set $\tcheby = \frac{150}{\sqrt{\varepsilon}} \log (n/\varepsilon) $, $\dcheby = (\varepsilon/n)^{401}$, and let $\widetilde{\exp}\left(Y^{\lr{t}}/2\right) := \cheby\left( Y^{\lr{t}}/2, \tcheby, \dcheby \right)$. Then for all $1\leq i \leq n$, \[\abs{\exp\left(Y^{\lr{t}}\right)_{ii} - \left(\widetilde{\exp} Y^{\lr{t}}\right)_{ii}} \leq \dexp \defeq \tfrac{4800 \varepsilon^{401}}{n^{390}}.\]
\end{restatable}
\begin{proof}
We plug into Inequality~\ref{general-cheb-approx} the following bounds obtained from Lemma~\ref{spectrum-Y}: 
\begin{align*}
	a &= -\tfrac{60 \log n}{\varepsilon}, b = \log K\\
	u &= \lambda = \tfrac{1}{2}(b-a) = \tfrac{\log K}{2}+ \tfrac{30 \log n }{\varepsilon}
\end{align*} Applying Inequality~\ref{general-cheb-approx}, we then get
\begin{align*}
  \sup_{\lambda \in \left[-\tfrac{30 \log n}{\varepsilon}, \tfrac{1}{2}\log K\right]} \abs{ K r_{\tfrac{1}{2}\log K + \tfrac{30 \log n}{\varepsilon}, \delta}\lr{\tfrac{1}{2}\log K - \tfrac{1}{2}\lambda}  - \exp\lr{\tfrac{1}{2}\lambda}} &\leq \delta K
\end{align*} We have $K = 40 n \lr{\log n}^{10}$; therefore, if we want the error bound to be roughly $\tfrac{\varepsilon}{n}$, then we need to pick $\delta = \mbox{polylog}(\varepsilon, n)$. Because of technical details in Lemma~\ref{pf-algnew-distbndinalg}, we choose \[ \dcheby = \lr{\frac{\varepsilon}{n}}^{401}. \numberthis\label{choice-of-delta}\] This gives us the following result. \[\opnorm{\exp(Y^{(t)}/2) - \cheby(Y^{(t)}/2, \tcheby, \dcheby)} \leq 40 \frac{\varepsilon^{401}}{n^{396}}.\] From Lemma~\ref{lem-sv-exp-approx-degree}, we get that the degree of polynomial required to achieve this guarantee is \[ \mbox{Required Degree } = \sqrt{\frac{ 2 \times 10^4}{\varepsilon} \log n \log (n/\varepsilon)} \leq \frac{150}{\sqrt{\varepsilon}} \log (n/\varepsilon). \] This is the value of $\tcheby$ that we choose. We now bound the quantity we actually care about. We can write  $\widetilde{\exp}\lr{\tfrac{1}{2}Y^{\lr{t}}} = \exp\lr{\tfrac{1}{2} Y^{\lr{t}}} + \Delta$, where $\opnorm{\Delta} =  40 \frac{\varepsilon^{401}}{n^{396}}$, the error guarantee obtained above.  Simplifying with the application of  $\opnorm{\exp(Y^{(t)})}\leq K$   obtained from Lemma~\ref{lem-newalg-distbnd-opt} gives 
	\begin{align*}
		\opnorm{\exp\lr{Y^{\lr{t}}} - \widetilde{\exp}\lr{Y^{\lr{t}}}} &= \opnorm{\lr{\exp\lr{\tfrac{1}{2}Y^{\lr{t}}}}^2 - \lr{\widetilde{\exp}\lr{\tfrac{1}{2}Y^{\lr{t}}}}^2}\\
								             &= \opnorm{\Delta^2 + \Delta \exp\lr{\tfrac{1}{2} Y^{\lr{t}}}+ \exp\lr{\tfrac{1}{2}Y^{\lr{t}}} \Delta}\\
												&\leq (40 \frac{\varepsilon^{401}}{n^{396}})^2 + 2(40 \frac{\varepsilon^{401}}{n^{396}}) \opnorm{\exp\lr{\tfrac{1}{2}Y^{\lr{t}}}}\\
 &\leq (40 \frac{\varepsilon^{401}}{n^{396}})^2 + 2(40 \frac{\varepsilon^{401}}{n^{396}})K\\
				&\leq 3 (40 \frac{\varepsilon^{401}}{n^{396}}) K \\
				&\leq 3 \cdot \frac{40 \varepsilon^{401}}{n^{396}} \cdot 40 n\lr{\log n}^{10} \\
				&\leq  \frac{4800 \varepsilon^{401}}{n^{390}}.
\end{align*} Substituting our assumption $n \geq 4$ above gives the claimed bound. 
\end{proof}
In conclusion, we showed that we can approximate our matrix exponential to $\varepsilon$-accuracy using $\Ord(1/\sqrt{\varepsilon})$ terms in the polynomial approximation. 

\subsection{Properties of Estimators}\label{new-estimator-pfs}
Since we have an inner loop in Algorithm~\ref{alg-new} with estimated quantities, it is crucial for the convergence that these estimators have a small bias and variance. In this section we show that this is indeed the case. We first prove two technical results about the functions $\polyg$ and $\jl$ which are ``building blocks'' of our estimators. We then apply these results in proving properties of $\widehat{\theta}_1$ and $\widehat{\theta}_2$, and subsequently those of the overall estimator $\widehat{\theta}$. 
\subsubsection{Two Technical Results about Estimators}
\prefirstpartfirstterm*
\begin{proof} Recall that given a distribution $\widetilde{X}$ with a positive support, and integer $N >0$, we define $\polyg$ as the  approximation for $g\lr{u} = u^{-\sfrac{1}{2}}$ at $x_0$ sampled from $\widetilde{X}$: \[\polyg(\widetilde{X}, N) = \sum\limits_{k = 0}^{N-1}\frac{1}{k!} g^{\lr{k}}(x_0) \prod_{j = 1}^k ( x_{k, j} - x_0 ), \, \mbox{where } x_0, x_{k, j} \stackrel{\text{i.i.d.}}{\sim} \widetilde{X},\] where $g^{\lr{k}}\lr{u} = \frac{\lr{-1}^k}{2^k} {u}^{-j - \sfrac{1}{2}} \prod\limits_{\ell = 1}^j \lr{2\ell-1}$ denotes the $k$-th derivative of $g$ evaluated at $u$. Then the expected value of $g$ with respect to the distribution $\mathcal{G}(X)$ is 
\begin{align*}
	\E g &= \E \sum\limits_{j = 0}^{k-1} \frac{1}{j!} g^{\lr{j}}(x) \prod\limits_{\ell = 1}^j\lr{x_{j, \ell} - x}\\
	&= \E \sum\limits_{j = 0}^{k-1} \frac{1}{j!} g^{\lr{j}}(x) \prod\limits_{\ell = 1}^j\lr{\E x_{j, \ell} - x}\\
	&=  \E \sum\limits_{j = 0}^{k-1} \frac{1}{j!} g^{\lr{j}}(x) \lr{\mu - x}^j. \numberthis\label{est-tech1-bias-step1}
\end{align*} To see how the term on the right hand side of Equation~\ref{est-tech1-bias-step1} differs from the true quantity to be estimated, we apply Taylor's remainder theorem: for some point $\zeta$ lying between $\mu$ and $x$, we have 
\begin{align*}
	\left|\sum\limits_{j=0}^{k-1}\frac{1}{j!}g^{\lr{j}}(x)\lr{\mu - x}^{j}-{\mu}^{-\sfrac{1}{2}}\right| &\leq\frac{g^{\lr{k}}\lr{\zeta}}{k!}\left|x-\mu \right|^{k}\\  	 							&\leq\frac{\left|x-\mu\right|^{k}}{\min\lr{x, \mu}^{k+\sfrac{1}{2}}}, 
\end{align*} where the second inequality follows from \[\abs{\frac{g^{\lr{k}}\lr{u} }{k!}} \leq {u}^{-k - \frac{1}{2}}, \numberthis\label[ineq]{generic-jth-der-bound}\] and the fact that $\zeta$ lies between $x$ and $\mu$. Combining this with Jensen's inequality gives us the final bound on the first moment, \[ \left|\E  g-{\mu}^{-\sfrac{1}{2}}\right|\leq\E \left|g-{\mu}^{-\sfrac{1}{2}}\right|\leq\E \frac{\left|x-{\mu}\right|^{k}}{\min\lr{x,\mu}^{k+\sfrac{1}{2}}}. \numberthis\label[ineq]{est-pre-term1}\] To prove the bound on the second moment, we again start with the definition of $\polyg$,
\begin{align*}
	\E \abs{g}^2 &= \E \left(\sum\limits_{j=0}^{k-1}\frac{1}{j!}g^{\lr{j}}(x)\prod\limits_{\ell=1}^{j}\left(x_{j,\ell}-x\right)\right)^{2}\\
	&\stackrel{\mbox{\numcircledtikz{1}}}{\leq} k\E \sum_{j=0}^{k-1}\left(\frac{g^{\lr{j}}(x)}{j!}\prod_{\ell = 1}^{j}\left(x_{j,\ell}-x\right)\right)^{2}\\
	&\stackrel{\mbox{\numcircledtikz{2}}}{=} k \sum\limits_{j = 0}^{k-1} \E  \left( \left( \frac{g^{\lr{j}} (x)}{j!} \right)^{2} \left( \sigma^2 + \lr{x - \mu}^2 \right)^j \right)\\
	&\stackrel{\mbox{\numcircledtikz{3}}}{\leq} k \sum_{j = 0}^{k-1} \E \lr{ \frac{\lr{\sigma^2 + \lr{x - \mu}^2}^j}{x^{2j+1}} }. \numberthis\label[ineq]{est-pre-var-term1}
\end{align*}
 Here $\mbox{\numcircledtikz{1}}$ is by Cauchy-Schwarz inequality; $\mbox{\numcircledtikz{2}}$ is by using the fact that each $x_{j, \ell}$ is sampled independently and adding and subtracting $\mu$ from the term inside the square and using the definition of $\sigma^2$; $\mbox{\numcircledtikz{3}}$ uses Inequality~\ref{generic-jth-der-bound}.  
\end{proof}

\presecondpartfirstterm*
Before diving into this proof, we state below a tool we need about logconcave distributions. 
\begin{theorem}[Theorem 5.22 in \cite{lovasz-vempala-survey}]\label{lovasz-vempala-thm}
If $X \in \reals^n$ is a random point sampled from a logconcave distribution, then $ (\E \abs{X}^k)^{\sfrac{1}{k}} \leq 2k \E \abs{X}$.
\end{theorem} 
\begin{proof}[Proof of Lemma~\ref{lem-est-pre2-term1}]By linearity of the Gaussian distribution, given a $\zeta \sim \mathcal{N}(0, I_n)$ and for some $u \in \reals^n$, we have $\zeta^T u \sim \mathcal{N}(0, \norm{u}_2^2)$. Therefore $\jl(u, N)$ gives us a scaled chi-squared distribution, $X = \tfrac{\mu}{N} \chi^2_{N}$. For a point $x\sim X$, using the parameters of a standard chi-squared distribution gives us the following properties.
 \[\E x = \frac{\mu}{N} \cdot N = \mu, \mbox{ and } \Var x =\left(\frac{\mu}{N}\right)^2 N \lr{N + 2} - \mu^2 = 2 \frac{\mu^2}{N},\numberthis\label{unscaled-var}\] 
 which proves \textbf{($\mathbf{1}$)} and \textbf{($\mathbf{2}$)}. To prove \textbf{($\mathbf{3}$)}, we first scale the random variable $x$ by $N/\mu$ to make it of a standard chi-squared distribution; this makes our computations easier, since we later need to use the closed-form expression of the probability density function of $x$. After the scaling, we have \[\E\limits_{x \sim \chi^2_N} x = N \qquad \Var\limits_{x \sim \chi^2_N} = 2N. \numberthis\label{scaled-var-props}\]                                                       Therefore,
\begin{align*}
\E_{x \sim X} \lr{ \frac{\lr{\sigma^2 + \lr{\mu - x}^2}^{k}}{\min\lr{x, \mu}^{2k + 1}} } &\stackrel{\mbox{\numcircledtikz{1}}}{\leq} 2^k \E_{x\sim X} \lr{\frac{\sigma^{2k} + \lr{\mu - x}^{2k}}{\min\lr{x, \mu}^{2k + 1}}}\\
														&\stackrel{\mbox{\numcircledtikz{2}}}{=} 2^k \frac{N}{\mu} \underbrace{\E_{x \sim \chi^2_N}\lr{ \frac{\lr{2N}^k + \lr{N - x}^{2k}}{\min\lr{x, N}^{2k + 1}}}}_{\text{$\mbox{\numcircledtikz{A}}$}}. \numberthis\label{step1}
\end{align*} Here $\mbox{\numcircledtikz{1}}$ follows from Jensen's inequality applied to the function $g(x) = x^k$ for $k > 1$ and $x>0$; the equation $\mbox{\numcircledtikz{2}}$ follows from Equation~\ref{unscaled-var}. We now bound $\mbox{\numcircledtikz{A}}$  by considering the random variable in two disjoint intervals as follows.
\begin{align*}
\mbox{\numcircledtikz{A}} &= \E_{x \sim \chi^2_N} \left( \frac{ \lr{2N}^k + \lr{N - x}^{2k} }{\min\lr{x, N}^{2k+1}} \ones_{\left\{ x < \frac{N}{4} \right\}} \right) +  \E_{x \sim \chi^2_N} \left( \frac{\lr{2N}^k + \lr{N - x}^{2k}}{\min\lr{x, N}^{2k+1}} \ones_{\left\{ x \geq \frac{N}{4} \right\}} \right). \\
							&\leq  \underbrace{\E_{x \sim \chi^2_N} \lr{ \frac{\lr{2N}^k + \lr{N - x}^{2k}}{x^{2k + 1}} \ones_{\left\{ x< \frac{N}{4} \right\}}}}_{\text{$\mbox{\numcircledtikz{B}}$}} + \frac{1}{\lr{N/4}^{2k+1}} \underbrace{\E_{x \sim \chi^2_N} \left( \lr{2N}^k + \lr{N - x}^{2k} \right)}_{\text{$\mbox{\numcircledtikz{C}}$}}. \numberthis\label[ineq]{step2}
\end{align*} To bound $\mbox{\numcircledtikz{B}}$, we divide the region $\left\{ x < N/4\right\}$ into intervals of geometrically-varying lengths as follows. 
\begin{align*}
	\mbox{\numcircledtikz{B}} &= \sum\limits_{j = 2}^{\infty} \E_{x \sim \chi^2_N} \lr{ \frac{\lr{2N}^k + \lr{N - x}^{2k} }{x^{2k+1}} \ones_{\left\{\frac{N}{2^{j+1}} \leq x < \frac{N}{2^j}\right\}} }\\
					&\leq \sum\limits_{j = 2}^{\infty} \frac{N^{2k} 5^k}{\lr{N/2^{j+1}}^{2k+1}} \underbrace{\Prob\lr{x < N/2^j}}_{\text{$\mbox{\numcircledtikz{D}}$}}, \numberthis\label{D-for-B} 
\end{align*} where the inequality follows from the worst case upper bounds for the numerator and $1 + 2^k \leq 5^k$ for $k\geq 1$ and the worst case lower bounds for the denominator over each interval $\{N/2^{j+1} \leq x < N/{2^j} \}$. For $a>0$ and a random variable $x \sim \chi^2_N$, we have the following cumulative distribution function: 
\begin{align*}
	\Prob\lr{x \leq a} &= \int_0^a \frac{e^{-x/2} x^{N/2 - 1}}{2^{\lr{N/2}} \Gamma\lr{N/2}} dx\\
					&\leq \int_0^a \frac{e^{-\sfrac{x}{2}} x^{N/2 - 1}}{2^{N/2} (N/2e)^{(N-1)/2}} dx\\
					&\leq \frac{2 a^{N/2 - 1} e^{N/2}}{N^{(N-1)/2}},
\end{align*} where we used the Sterling approximation of Gamma function in the second inequality. Substituting $a = 2^{-j}N$ above and simplifying gives the following bound on the quantity from Inequality~\ref{D-for-B}. 
\begin{align*}
	\mbox{\numcircledtikz{D}} &\leq \frac{2^{j+1}}{\sqrt{N}} \lr{\frac{e}{2^j}}^{\frac{N}{2}}. \numberthis\label[ineq]{cdf-for-term2}
\end{align*} Substitute into Inequality~\ref{D-for-B} to get 
\begin{align*}
	\mbox{\numcircledtikz{B}} &\leq \sum_{j=2}^{\infty} N^{2k} 5^k \lr{\frac{2^{j+1}}{N}}^{2k+1} \frac{2^{j+1}}{\sqrt{N}} \lr{\frac{e}{2^j}}^{\sfrac{N}{2}}\\
			&= \frac{5^k 2^{2k+2} e^{\sfrac{N}{2}}}{N^{\sfrac{3}{2}}} \sum\limits_{j = 2}^{\infty} \frac{1}{2^{j\lr{\sfrac{N}{2} - 2k - 2 }}}\\
							&\leq \frac{2^{5k + 2} e^{\sfrac{N}{2}}}{N^{\sfrac{3}{2}}} \frac{2}{2^{N - 4k - 4}}\\
							&\leq \frac{e^{\sfrac{N}{2}}}{N^{\sfrac{3}{2}} 2^{N - 9k - 7}},\numberthis\label[ineq]{finalboundonb}
\end{align*} where we used the condition that $N \geq 4k + 6$ in the first two inequalities. Next, we bound $\mbox{\numcircledtikz{C}}$. 
\begin{align*} 
		\mbox{\numcircledtikz{C}} &=  \lr{2N}^{k} \left( \E \abs{ \frac{x - N}{\sqrt{2N}} }^{2k} + 1\right)\\
									%&=\sqrt{2N}^{2k} \left( \E\limits_{x \sim \chi^2_N} \abs{ \frac{x - N}{\sqrt{2N}} }^{2k} + 1 \right)\\
							 &\stackrel{\mbox{\numcircledtikz{1}}}{\leq} \lr{2N}^{k} \left( 2^{2k} \lr{2k}^{2k} \left( \E \frac{\abs{x - N}}{\sqrt{2N}} \right)^{2k} + 1 \right)\\
												&\stackrel{\mbox{\numcircledtikz{2}}}{\leq} \lr{2N}^{k} \left( 2^{2k} \lr{2k}^{2k} \lr{\frac{\sqrt{\E \abs{x - N}^2}}{\sqrt{2N}}}^{2k} + 1\right)\\
												&= \lr{2N}^k \lr{ 2^{2k} \lr{2k}^{2k} + 1 }\\
												&\leq \lr{2N}^k \lr{32k^2}^k, \numberthis\label[ineq]{term1-chisquared-split}
\end{align*} where $\mbox{\numcircledtikz{1}}$ is by invoking Theorem~\ref{lovasz-vempala-thm}, which is valid  by logconcavity of chi-squared distribution, and $\mbox{\numcircledtikz{2}}$ is by Jensen's inequality. Plugging Inequality~\ref{finalboundonb} and Inequality~\ref{term1-chisquared-split} into Equation~\ref{step1} gives: 
\begin{align*}
\E_{x\sim X} \lr{ \frac{\lr{\sigma^2 + \lr{x - \mu}^{2}}^k}{\min\lr{x, \mu}^{2k+1}} } &\leq 2^k \frac{N}{\mu}\lr{ \frac{e^{\sfrac{N}{2}}}{N^{\sfrac{3}{2}} 2^{N - 9k - 7}} +  \frac{4^{2k+1}}{N^{2k+1}}\lr{2N}^k \lr{32k^2}^k }\\
		&\leq \frac{1}{\mu} \lr{ \frac{e^{\sfrac{N}{2}}}{2^{N - 17k}}  + \frac{ 2^{13k} k^{2k}}{N^k} },
\end{align*} which is what is to be proved. 
\end{proof}

\subsubsection{Properties of $\widehat{\theta}_1$}\label{app-hat-theta-1}
We prove the bounds on first and second moments of $\widehat{\theta}_1$. Note that this is where we make our choice of $\testinvsqrt$ and $\testjl$ for the modules $\polyg$  and $\jl$ used in estimating $\theta_1$ in the subroutine $\estthetaone$.
\termone*
\begin{proof} Consider a random variable $x$ sampled from the distribution $(\widetilde{Z^2})_{ii}$. Because of Lemma~\ref{lem-est-pre2-term1}, we have $\E x =  (Z^2)_{ii}$. Then $x+1$ satisfies the required bias condition of Lemma~\ref{lem-est-pre-term1} for constructing a polynomial approximation for $1/\sqrt{1 + (Z^2)_{ii}}$. Then  $\widehat{\theta}_{1_i}$ satisfies
	\begin{align*}
		\abs{ \E \widehat{\theta}_{1_i} - \frac{1}{\sqrt{1 + (Z^2)_{ii}}}  } &\stackrel{\mbox{\numcircledtikz{1}}}{\leq} \E \lr{ \frac{\abs{x - (Z^2)_{ii}}^{\testinvsqrt}}{\min(x+ 1, (Z^2)_{ii}+1)^{\testinvsqrt + \sfrac{1}{2}}}}\\
																&\stackrel{\mbox{\numcircledtikz{2}}}{\leq} \sqrt{\E \frac{\lr{x - (Z^2)_{ii}}^{2\testinvsqrt}}{\min\lr{x + 1, (Z^2)_{ii}+1}^{2\testinvsqrt + 1}} }\\
																&\stackrel{\mbox{\numcircledtikz{3}}}{\leq} \sqrt{ \frac{1}{(Z^2)_{ii} + 1} \left(\frac{e^{\sfrac{\testjl}{2}}}{2^{\testjl - 17\testinvsqrt}}  + \frac{2^{13\testinvsqrt} {\testinvsqrt}^{2\testinvsqrt}}{{\testjl}^{\testinvsqrt}} \right) }. 
    \end{align*} where $\mbox{\numcircledtikz{1}}$ is by Lemma~\ref{lem-est-pre-term1}, $\mbox{\numcircledtikz{2}}$ is by Jensen's inequality, and $\mbox{\numcircledtikz{3}}$ is by a slight modification of the proof of $(3)$ in Lemma~\ref{lem-est-pre2-term1} (instead of scaling by $N/\mu$, we scale by $N\mu/(\mu+1)$ in the proof). Finally, set $\testinvsqrt = 1600 \log\lr{\tfrac{n}{\varepsilon}}$ and $\testjl = 2^{14} \testinvsqrt^2$ to get the claimed bias. Next, we can bound the variance as follows.
	\begin{align*}
		\E|\widehat{\theta}_{1_i}|^2 &\stackrel{\mbox{\numcircledtikz{1}}}{\leq} \testinvsqrt \sum_{k = 0}^{\testinvsqrt-1} \E \lr{ \frac{\left(\sigma^2 + \lr{x  - (Z^2)_{ii}}^2\right)^k}{(x+1)^{2k+1}} }\\
							&\leq \testinvsqrt \sum_{k = 0}^{\testinvsqrt-1} \E \lr{ \frac{\left(\sigma^2 + \lr{x  - (Z^2)_{ii}}^2\right)^k}{\min\lr{x + 1, (Z^2)_{ii} + 1}^{2k+1}} }\\
							&\stackrel{\mbox{\numcircledtikz{2}}}{\leq}  \frac{\testinvsqrt}{(Z^2)_{ii}}\sum_{k=0}^{\testinvsqrt-1} \lr{ \frac{e^{\testjl/2}}{2^{\testjl - 17 k}} + \frac{2^{13k} k^{2k} }{ \testjl^k } }\\
							&\stackrel{\numcircledtikz{3}}{=} \frac{\testinvsqrt}{(Z^2)_{ii}} \sum_{k = 0}^{\testinvsqrt-1} \lr{ 2^{17k}\lr{\frac{\sqrt{e}}{2}}^{2^{14}\testinvsqrt^2}  + \frac{k^{2k}}{2^k \testinvsqrt^{2k} }}
		\end{align*}  where $\mbox{\numcircledtikz{1}}$ is by $(2)$ in Lemma~\ref{lem-est-pre-term1}, $\mbox{\numcircledtikz{2}}$ is by $(3)$ in Lemma~\ref{lem-est-pre2-term1}, and $\mbox{\numcircledtikz{3}}$ is by writing $\testjl$ in terms of $\testinvsqrt$. We have the simplications, $\sum_{k = 0}^{\testinvsqrt-1}  2^{17k}\lr{\frac{\sqrt{e}}{2}}^{2^{14}\testinvsqrt^2} \leq \frac{2^{17 \testinvsqrt}}{1.2^{2^{14}\testinvsqrt} 2^{16}}$ and \[\sum_{k = 0}^{\testinvsqrt - 1} \left( \frac{k^2}{2\testinvsqrt^2} \right)^k \leq 1 + \frac{1}{2\testinvsqrt^2} + \frac{4}{\testinvsqrt^4} + \sum_{k = 3}^{\testinvsqrt/2} \left( \frac{k^2}{2\testinvsqrt^2} \right)^k + \sum_{k > \testinvsqrt/2} \left( \frac{k^2}{2\testinvsqrt^2} \right)^k. \] Finally, plug in the values of $\testinvsqrt$ to get the desired bound. 
\end{proof} In Algorithm~\ref{alg-new}, we construct the matrix $Z$ as an approximation to $\exp\lr{\tfrac{1}{2}\lr{Y^{(t)} + s\Delta}}$ by the subroutine $\cheby\lr{\tfrac{1}{2}\lr{Y^{(t)} + s\Delta}, T_{\mathrm{Cheby}}, \dcheby}$, with details as provided in Lemma~\ref{algcurr-cheby-cost-alg}. With this value of $Z$ and the same rest of the notation as in the above lemma, we therefore wish to compare $\E\widehat{\theta}_{1_i}$ with $\tfrac{1}{\sqrt{\exp\lr{Y^{(t-1)} + s\Delta}_{ii} + 1}}$. Note that the above lemma only tells us that we are close to $\frac{1}{\sqrt{(Z^2)_{ii} + 1}}$, but $Z$, as defined above in Lemma~\ref{algcurr-cheby-cost-alg}, is only an  approximation to $\exp\lr{\tfrac{1}{2}\lr{ Y^{(t-1)} + s\Delta} }$. We therefore obtain the following corollary which gives us a precise bound on the bias we care about.  
\begin{corollary}[Bias of $\widehat{\theta}_{1_i}$]\label{cor-bias-theta1}
The estimator $\widehat{\theta}_{1_i}$ described in Algorithm~\ref{alg-est-theta} satisfies\[\abs{\E\widehat{\theta}_{1_i} - \frac{1}{\sqrt{\exp\lr{Y^{(t-1)} + s\Delta}_{ii}+1}}}  \leq  \btonei \defeq \frac{(1 + 2\dexp) \sqrt{2}(\frac{\varepsilon}{n})^{400} + 2\dexp}{\sqrt{\exp\lr{Y^{(t-1)} + s\Delta}_{ii}+1}}  ,\] where $\dexp = 4800 \frac{\varepsilon^{401}}{n^{390}}$. 
\end{corollary}
\begin{proof} From Lemma~\ref{algcurr-cheby-cost-alg}, we know that $Z = \cheby\lr{\tfrac{1}{2}\lr{Y^{(t-1)}+ s\Delta}, \tcheby, \dcheby}$ satisfies 
\begin{align*}
	\abs{\lr{\exp\lr{Y^{(t-1) + s\Delta}}-Z^2}_{ii}} &\leq \frac{4800 \varepsilon^{401}}{n^{390}}.
\end{align*} For ease of notation, let $\dexp \defeq \frac{4800 \varepsilon^{401}}{n^{390}}$. Given $a - \delta \leq b \leq a + \delta$, we use the Taylor series approximation to compute the error $\frac{1}{\sqrt{a}} - \frac{1}{\sqrt{b}}$. We have: 
\begin{align*}
	\abs{\frac{1}{\sqrt{a}} - \frac{1}{\sqrt{b}}} &\leq \abs{\frac{1}{\sqrt{a}} - \frac{1}{\sqrt{-\delta + a}} }\\
	&= \frac{1}{\sqrt{a}}\abs{1 - \frac{1}{\sqrt{1 - \delta/a}}}\\
	&\leq  \frac{1}{\sqrt{a}} \frac{2\delta}{a} = \frac{2 \delta}{a^{3/2}},
\end{align*} where we used the Taylor approximation of $\frac{1}{\sqrt{1 - x}}$ for small $x$. Thus, we have, from the above and Lemma~\ref{lem-est-term1},
\begin{align*}
	\abs{\E \widehat{\theta}_{1_i} - \frac{1}{\sqrt{\exp\lr{Y^{(t-1)} + s\Delta}_{ii}+1}} } &\leq \frac{\sqrt{2}(\varepsilon/n)^{400}}{\sqrt{Z^2_{ii} + 1}} + \frac{2\delta}{\sqrt{\exp\lr{Y^{(t-1)} + s\Delta}_{ii}+1}} \\
	&\leq \frac{(1 + 2\delta)\sqrt{2}(\varepsilon/n)^{400} + 2\delta}{\sqrt{\exp\lr{Y^{(t-1)} + s\Delta}_{ii}+1}},
\end{align*} which proves the claim. 
\end{proof}

\subsubsection{Properties of $\widehat{\theta}_2$}\label{prop-hat-theta2}
\termtwo*
\begin{proof}
The bias is defined as
\begin{align*}
	\E \widehat{\theta}_{2_i} &= \ones_{i}^{T} Z_1 \Delta Z_2 \lr{\E\zeta\zeta^{T}} Z \ones_{i}\\
								&= \lr{Z_1 \Delta Z_2 Z }_{ii} = \theta_{2_i},
\end{align*} where the second step is from the fact that $\zeta\sim \mathcal{N}(0, I_n)$ and linearity of expectation, and the last is by definition of $\theta_{2}$. Next, from Lemma~\ref{lem_exp_bound-gaussian}, given $a, b \in \reals^n$  and $\zeta \sim \mathcal{N}(0, I_n)$, we conclude that $\E ( (\zeta^T a)^2 (\zeta^T b)^2 ) \leq 3 \norm{a}_2^2 \norm{b}_2^2$. Therefore,
\begin{align*}
	\E\abs{\widehat{\theta}_{2_i}}^{2} &= \E (\ones_{i}^{T} Z_1 \Delta Z_2 \zeta)^{2})(\zeta^{T}  Z \ones_{i})^{2}\\
							&\leq 3\norm{Z_2 \Delta Z_1 \ones_{i}}^{2}\norm{Z \ones_{i}}^{2}\\
							&= 3(Z_1 \Delta Z_2^2  \Delta Z_1)_{ii} (Z^2)_{ii}.
\end{align*} This proves the bound on the second moment. 
\end{proof} As before, we can obtain, as a corollary of this result, a comparison of the mean of our estimator with the quantity we actually are trying to compute. 
\begin{corollary}[Bias of $\widehat{\theta}_{2_i}$]\label{cor-bias-theta2}
The estimator $\widehat{\theta}_{2_i}$ described in Algorithm~\ref{alg-est-theta} satisfies \[ \abs{ \E\widehat{\theta}_{2_i} - \lr{ \exp(\bar{\tau} (Y^{(t-1)} + s\Delta))\Delta \exp((\tau - \tfrac{1}{2}) (Y^{(t-1)} + s\Delta))\exp(\tfrac{1}{2}(Y^{(t-1)} + s\Delta)) }_{ii}} \leq 15 \dexp \eta K\] where $\dexp = \frac{4800 \varepsilon^{401}}{n^{390}}$. 
\end{corollary}
\begin{proof} This proof simply involves writing out some matrix products and bounds on the diagonal entries of the products (using the operator norm of the individual matrices). We show this below. Let  $Z_1 = \exp\lr{\bar{\tau} \lr{Y^{(t-1)} + s\Delta}} + U_1$, $Z_2 = \exp\lr{(\tau - 1/2) \lr{Y^{(t-1)} + s\Delta}} + U_2$, and $Z = \exp\lr{\tfrac{1}{2}\lr{Y^{(t-1)} + s\Delta}} + U$. From Lemma~\ref{lem-est-term2}, we have that $\E\widehat{\theta}_{2_i} = \theta_{2_i}$. We now express $\theta_{2_i}$ in terms of the matrix exponentials we care about. For ease of notation, we use $Y_s = Y^{(t-1)} + s\Delta$. 
\begin{align*}
	\E \widehat{\theta}_{2_i} - \lr{\exp\lr{\bar{\tau}Y_s} \Delta \exp\lr{(\tau - 1/2) Y_s} \exp\lr{\tfrac{1}{2}Y_s}}_{ii} &= \lr{\exp\lr{\bar{\tau}Y_s} \Delta \exp\lr{(\tau - 1/2) Y_s}U}_{ii} \\
	&+  \lr{\exp\lr{\bar{\tau}Y_s} \Delta U_2 \exp\lr{\tfrac{1}{2}Y_s}}_{ii}  + \lr{\exp\lr{\bar{\tau}Y_s} \Delta U_2 U}_{ii}  \\
	&+ \lr{U_1 \Delta \exp\lr{(\tau - 1/2) Y_s}\exp\lr{\tfrac{1}{2}Y_s}}_{ii} \\
	&+ \lr{U_1 \Delta \exp\lr{(\tau - 1/2) Y_s} U}_{ii}\\
	&+ \lr{U_1 \Delta U_2 \exp\lr{\tfrac{1}{2}Y_s} }_{ii}
 + \lr{U_1 \Delta U_2 U}_{ii}. 
\end{align*}  We can bound this by bounding the operator norm of each of the terms. Matrix norm is sub-multiplicative, so this in turn is bounded by the operator norm of the individual terms in the matrices. From Inequality~\ref{specbnd2}, we know that $\opnorm{\exp\lr{\alpha Y_s}} \leq K^{\alpha}$, $\opnorm{\Delta} \leq \eta G$, $\opnorm{U_1} \leq \dexp$, $\opnorm{U_2} \leq \dexp$, and $	\opnorm{U} \leq \dexp$, where $\dexp = \tfrac{4800 \varepsilon^{401}}{n^{390}}$. Substituting these values here and bounding each term  by  the largest of all terms gives us the bound to be proved. 
\end{proof}

\subsubsection{Properties of the Overall Estimator, $\widehat{\theta}$}
\biasandvarbndest*
\begin{proof}We can get the bound on the bias by applying the results of Corollaries~\ref{cor-bias-theta1} and \ref{cor-bias-theta2} in $\E\widehat{\theta}_i = \E \widehat{\theta}_{1_i} \E\widehat{\theta}_{2_i}$. We need the following definition to concisely write out expressions in this proof. 
\begin{definition}Let $\theta_{1_i} = \frac{1}{\sqrt{\exp(Y_s)_{ii} + 1}}$, $\theta_{2_i} =  \frac{1}{2}\lr{ \exp\lr{\bar{\tau}Y_s} \Delta \exp\lr{(\tau - \tfrac{1}{2}) Y_s} \exp\lr{\tfrac{1}{2}Y_s}}_{ii}$, $b_{1_i} = \theta_{1_i} (2\dexp + (1+2\dexp) \sqrt{2}(\varepsilon/n)^{400})$, and $b_{2_i} = 15 \dexp \eta K$ for $Y_s = Y^{(t-1)} + s\Delta$. 
\end{definition}We have the following error bound. 
\begin{align*}
	\abs{\E \widehat{\theta}_i - \int_{s=0}^1 \theta_{1_i} \int_{\tau = 0}^1  \theta_{2_i}  d\tau ds} &= \abs{ \int_{s=0}^1 \E\widehat{\theta}_{1_i} \int_{\tau=0}^1 \E \widehat{\theta}_{2_i} d\tau ds - \int_{s=0}^1 \theta_{1_i} \int_{\tau=0}^1 \theta_{2_i} d\tau ds} \\
	&\leq \int_{s=0}^1 \int_{\tau=0}^1 \abs{\E \widehat{\theta}_{1_i} \E \widehat{\theta}_{2_i} - \theta_{1_i} \theta_{2_i} } d\tau ds\\
	&\leq \abs{\E\widehat{\theta}_{1_i} \E \widehat{\theta}_{2_i} - \theta_{1_i} \theta_{2_i}}.
\end{align*} From Corollary~\ref{cor-bias-theta1}, we have $ \E \widehat{\theta}_{1_i} \in [\theta_{1_i} \pm b_{1_i}]$. From Corollary~\ref{cor-bias-theta2}, we have $\E \widehat{\theta}_{2_i} \in [\theta_{2_i} \pm b_{2_i}]$. Therefore, the right hand side above is bounded by: 
\begin{align*}
	\abs{\E \widehat{\theta}_i - \int_{s=0}^1 \theta_{1_i} \int_{\tau = 0}^1 \theta_{2_i} ds d\tau} &\leq b_{1_i} \theta_{2_i}  + b_{2_i} \theta_{1_i} + b_{1_i} b_{2_i}. 
\end{align*} We now compute a quantity which will be useful later: 
\begin{align*}
\sum_{i=1}^n\lr{	\E\widehat{\theta}_i - \int_{s=0}^1 \theta_{1_i} \int_{\tau=0}^1 \theta_{2_i} ds d\tau}^2 &\leq b_{1_i}^2 \sum_{i=1}^n \theta_{2_i}^2 +  (2 b_{1_i} b_{2_i}) (1 + b_{1_i})  \sum_{i=1}^n \theta_{2_i} + n b_{2_i}^2 (1 + b_{1_i})^2. \numberthis\label{bias-squared-intermediate}
\end{align*} Here we used the fact that $\theta_{1_i} = \frac{1}{\sqrt{\exp(Y_s)_{ii} + 1}} \leq 1$. We compute each of these terms separately. 
\begin{align*}
	\sum_{i=1}^n \theta_{2_i}^2 &= \sum_{i=1}^n \lr{\lr{\exp\lr{\bar{\tau} Y_s} \Delta \exp\lr{(\tau - 1/2)Y_s} \exp\lr{\tfrac{1}{2} Y_s}}_{ii}}^2 \\
	&\stackrel{\text{$\mbox{\numcircledtikz{A}}$}}{\leq} \sum_{i=1}^n \lr{ \exp\lr{\bar{\tau}Y_s} \Delta \exp\lr{(\tau-1/2)Y_s} \exp\lr{\tfrac{1}{2}Y_s} \exp\lr{\bar{\tau}Y_s} \Delta \exp\lr{(\tau-1/2)Y_s} \exp\lr{\tfrac{1}{2}Y_s}}_{ii}\\
	&= \Tr \lr{\exp\lr{\bar{\tau}Y_s} \Delta \exp\lr{Y_s} \Delta \exp\lr{\tau Y_s}} \\
	&= \Tr \lr{\exp\lr{Y_s} \Delta \exp\lr{Y_s} \Delta} \\
	&\leq K^2 \eta^2 G^2. \numberthis\label{bias-squared-term1}
\end{align*} Here, $\text{$\mbox{\numcircledtikz{A}}$}$ was because $\sum_{i=1}^n (A_{ii})^2 \leq \sum_{i=1}^n (A^2)_{ii}$, which can be checked by a simple computation. Similarly, the sum in the cross-term can be computed as follows.
\begin{align*}
\sum_{i=1}^n  \theta_{2_i} &= \sum_{i=1}^n \lr{\exp\lr{\bar{\tau}Y_s} \Delta \exp\lr{(\tau-1/2) Y_s} \exp\lr{\tfrac{1}{2}Y_s}}_{ii}\\ &= \Tr \lr{\exp\lr{\bar{\tau}Y_s} \Delta \exp\lr{(\tau-1/2) Y_s} \exp\lr{\tfrac{1}{2}Y_s} } \\
&= \Tr \lr{ \exp\lr{\bar{\tau}Y_s} \Delta \exp\lr{\tau Y_s} }\\
&= \Tr \lr{\exp\lr{Y_s}\Delta}\\
&\leq K \eta G. \numberthis\label{bias-squared-term2}
\end{align*} Substituting Inequalities~\ref{bias-squared-term1} and \ref{bias-squared-term2} into Inequality~\ref{bias-squared-intermediate}, and using $\frac{1}{\sqrt{\exp(Y_s)_{ii} + 1}} \leq 1$ gives us: 
\begin{align*}
	\sum_{i=1}^n\lr{	\E\widehat{\theta}_i - \int_{s=0}^1 a_1 \int_{\tau  =0 }^1 a_2 dsd\tau }^2  &\leq (2\dexp + (1 + 2\dexp) \sqrt{2}(\varepsilon/n)^{400})^2 K^2 \eta^2 G^2\\
	&+ 900 n \delta^2 \eta^2 K^2 \\
	&+ 60 \eta \delta K (2\dexp + (1 +2\dexp) \sqrt{2}(\varepsilon/n)^{400})K \eta G\\
	&\leq 6 K^2 \eta^2 (\sqrt{2}(\varepsilon/n)^{400} + 2\dexp)\\
&\leq 400 n K^2 \eta^2 (\sqrt{2}(\varepsilon/n)^{400} + 2 \dexp). \numberthis\label{bias-squared-final}
\end{align*}
	We now prove the final variance bound.
	\begin{align*}
		\E_{s, \tau, \zeta_1, \zeta_2} \|\widehat{\theta}\|_2^2 &= \E_{s, \tau, \zeta_1, \zeta_2} \sum_{i=1}^n  |\widehat{\theta}_i|^2\\
		&= \int_{s=0}^1  \int_{\tau = 0}^1  \sum_{i=1}^n \E_{\zeta_1} |\widehat{\theta}_{1_i}|^2 \E_{\zeta_2}  |\widehat{\theta}_{2_i}|^2 ds d\tau.
\end{align*} Combining Lemmas~\ref{lem-est-term1} and \ref{lem-est-term2}, we get:
\begin{align*}
\E_{s, \tau, \zeta_1, \zeta_2}\|\widehat{\theta}\|_2^2 &= \int_{s=0}^1 \int_{\tau = 0}^1  \sum_{i=1}^n \underbrace{\tfrac{1630 \log(n/\varepsilon)}{ (Z^2)_{ii}}}_{\text{\mbox{\numcircledtikz{1}}}} \cdot 3 \underbrace{\lr{Z_2 \Delta Z_1^2 \Delta Z_2}_{ii} \lr{Z^2}_{ii}}_{\text{\mbox{\numcircledtikz{2}}}} ds d\tau, \\
		&\stackrel{\text{$\mbox{\numcircledtikz{A}}$}}{=} \sum_{i=1}^n \int_{s=0}^1 \int_{\tau=0}^1 4890 \log(n/\varepsilon) (Z_2 \Delta Z_1^2 \Delta Z_2)_{ii} ds d\tau \\
		&= 4890  \log(n/\varepsilon) \int_{s=0}^1 \int_{\tau=0}^1 \Tr\lr{ Z_2^2 \Delta Z_1^2 \Delta} dsd\tau, \numberthis\label{var-temp-1}
	\end{align*} where $Z_1 = \exp\lr{(\tau - 1/2)\lr{Y^{(t-1)} + s\Delta}} + U_1$ and $Z_2 = \exp\lr{\bar{\tau}\lr{Y^{(t-1)} + s\Delta}} + U_2$ as defined in Corollary~\ref{cor-bias-theta2}. The term $\text{$\mbox{\numcircledtikz{A}}$}$ shows the significance of carefully choosing the split in the estimator $\widehat{\theta}_2$, which enabled the cancellation of $\frac{1}{(Z^2)_{ii}}$ and $(Z^2)_{ii}$. We now bound $\Tr \lr{Z_2^2 \Delta Z_1^2 \Delta}$. In Lemma~\ref{algcurr-cheby-cost-alg} we showed how to construct $Z_1$ and $Z_2$ as $\dexp = 4800\varepsilon^{401}/n^{390}$ approximations to the respective matrix exponentials. Thus, writing $\opnorm{U_1} = \opnorm{U_2} = \dexp$ and expanding out the product $Z_2^2 \Delta Z_1^2 \Delta$ in terms of the true matrix exponentials and the error matrices, we get the following:	\[\Tr (Z_2^2 \Delta Z_1^2 \Delta) \leq \Tr(\exp(2\bar{\tau} (Y^{(t-1)} + s\Delta)) \Delta \exp((2\tau - 1)(Y^{(t-1)} + s\Delta)) \Delta ) + 30 \eta^2 \dexp K^2.\] Choosing $A = \exp\lr{Y^{(t-1)} + s\Delta}$ and $B = \Delta$ and combining with the fact that matrix exponential is positive semidefinite, and $\Delta$ is a symmetric matrix since the gradient of the objective is symmetric, invoking Fact~\ref{fact-extendedLiebThirring} gives:\[ \Tr \lr{Z_2^2 \Delta Z_1^2 \Delta} \leq \Tr (\exp(Y^{(t-1)} + s\Delta)\Delta^2) + 30 \eta^2 \dexp K^2 \leq 4 K \eta^2 + 30 \eta^2 \dexp K^2,\] where the last inequality follows from applying Holder's inequality with the nuclear norm and operator norm. Plugging this back into Equation~\ref{var-temp-1} and completing the integration gives
	\begin{align*}
	\E_{s, \tau, \zeta_1, \zeta_2} \|\widehat{\theta}\|_2^2 &\leq 4890 \log(n/\varepsilon) \lr{4 K \eta^2 + 30 K^2 \eta^2 \dexp}\leq 19600 \log(n/\varepsilon) K \eta^2 + 147000 K^2 \eta^2 \dexp.
	\end{align*} 
\end{proof}

\subsection{Number of Inner Iterations}\label{new-num-inner-iters}
We can use the general expression for overall running time to choose a value for number of `low-accuracy' iterations. The total computational cost of the algorithm is \[\Tout \times\frac{10^5 \lr{\log n}^{21} }{\varepsilon^2} T_{\mathrm{exp}} + \Tout \times \Tin \times 2^{30} \lr{ \log\lr{ \frac{1}{\varepsilon}} }^4 T_{\mathrm{exp}},\numberthis\label{totalcostnew}\] where the first term is the total cost of exact computations, and the second term is the total cost of approximate computations (done inside the inner loop); $T_{\mathrm{exp}}$ is the cost of approximating the products of matrix exponentials with a vector. This is optimal (ignoring polylogarithmic terms) when setting $\Tin = \Ord(1/\varepsilon^2)$. We set $\Tin = 1/\varepsilon^2$ due to technical reasons arising in Lemma~\ref{pf-algnew-distbndinalg}.

\subsection{Distance Bound Between Estimated and True Iterates}\label{sec-algnew-distbndinalg}
Since the estimators in the inner loop iterations are constructed to have a low variance, the estimated and true iterates aren't far apart, as we show now. This is also where we choose the step size $\eta$. 
\begin{restatable}{lemma}{algnewdistbndinalg}\label{pf-algnew-distbndinalg}
In Algorithm~\ref{alg-new}, after $t \leq \Tin$ iterations, we have $\E\|X^{\lr{t}} - \widetilde{X}^{\lr{t}}\|_{\mathrm{nuc}} \leq 1.132 n \varepsilon$. Recall, $\widetilde{X}^{(t)}$ is the approximate primal iterate, while $X^{(t)}$ is the exact iterate. 
\end{restatable}
\begin{proof} By the definition of $\dnorm{{}\cdot{}}$ and some algebra, we have
\begin{align*}
	\E \dnorm{X^{\lr{t}} - \widetilde{X}^{\lr{t}}} &= \E \sum\limits_{i = 1}^n \abs{ X^{\lr{t}}_{ii} - {\widetilde{X}^{\lr{t}}}_{ii} }\\
&= \E\sum\limits_{i = 1}^n \abs{ \lr{\sqrt{X^{\lr{t}}_{ii} + 1}}^2 - \lr{\sqrt{\widetilde{X}^{\lr{t}}_{ii} + 1}}^2}	\\							&= \E \sum\limits_{i = 1}^n 2\sqrt{X^{\lr{t}}_{ii} + 1} \abs{\sqrt{X^{\lr{t}}_{ii} + 1} - \sqrt{{\widetilde{X}^{\lr{t}}}_{ii} + 1}}+ \E\sum\limits_{i =1}^n \abs{\sqrt{X^{\lr{t}}_{ii} + 1} - \sqrt{{\widetilde{X}^{\lr{t}}}_{ii} + 1}}^2.
\end{align*} Next, apply Cauchy-Schwarz inequality and Lemma~\ref{lem-newalg-distbnd-opt} to get 
\begin{align*}
		\E \dnorm{X^{\lr{t}} - \widetilde{X}^{\lr{t}}} 	&\leq 2\E \sqrt{\Tr X^{\lr{t}} + n} \sqrt{\sum\limits_{i = 1}^n \lr{\sqrt{X^{\lr{t}}_{ii} + 1} - \sqrt{{\widetilde{X}^{\lr{t}}}_{ii} + 1}}^2} + \E \sum\limits_{i = 1}^n \lr{\sqrt{X^{\lr{t}}_{ii} + 1} - \sqrt{{\widetilde{X}^{\lr{t}}}_{ii} + 1}}^2\\
		 &\leq 2\sqrt{K + n} \underbrace{\E \sqrt{\sum\limits_{i = 1}^n \lr{\sqrt{X^{\lr{t}}_{ii} + 1} - \sqrt{{\widetilde{X}^{\lr{t}}}_{ii} + 1}}^2}}_{\text{$\mbox{\numcircledtikz{A}}$}} + \underbrace{\E \sum\limits_{i =1}^n \lr{\sqrt{X^{\lr{t}}_{ii} + 1} - \sqrt{{\widetilde{X}^{\lr{t}}}_{ii} + 1}}^2}_{\text{$\mbox{\numcircledtikz{B}}$}}. \numberthis\label[ineq]{new-alg-rate1}
\end{align*}  We first bound $\mbox{\numcircledtikz{B}}$. We can write a recursive formulation for as follows. 
\begin{align*}
	\sqrt{\widetilde{X}^{\lr{t}}_{ii} + 1} - \sqrt{X^{\lr{t}}_{ii} + 1} &= \underbrace{\left( \sqrt{\widetilde{X}^{\lr{0}}_{ii}  + 1} - \sqrt{X^{\lr{0}}_{ii} + 1} \right)}_{\text{$\mbox{\numcircledtikz{C}}$}} + \underbrace{\sum_{s = 1}^t \left( \widehat{\theta}^{\lr{s}}_i - \sqrt{X^{\lr{s}}_{ii} + 1} + \sqrt{X^{\lr{s-1}}_{ii} + 1} \right)}_{\text{$\mbox{\numcircledtikz{D}}$}}.
\end{align*} We invoke Johnson-Lindenstrauss lemma (restated in Lemma~\ref{thm-jl-statement} for completeness) and choose the accuracy parameter for it to be such that $\abs{X^{\lr{0}}_{ii} - \widetilde{X}^{\lr{0}}_{ii}} \leq \widetilde{\varepsilon} X^{\lr{0}}_{ii} = \frac{\varepsilon}{100 \lr{\log n}^{10}} X_{ii}^{(0)}$. Therefore, $\mbox{\numcircledtikz{C}} \leq  \tfrac{\widetilde{\varepsilon}}{2} \sqrt{X^{\lr{0}}_{ii}  + 1}  =  \frac{\varepsilon}{200 \lr{\log n}^{10}} \sqrt{X^{\lr{0}}_{ii} + 1}  $. Summing over all indices and taking expectations gives
\begin{align*}
\text{$\mbox{\numcircledtikz{B}}$} &\leq \E \sum\limits_{i=1}^n \left(\frac{\varepsilon}{200 \lr{\log n}^{10}} \sqrt{X^{\lr{0}}_{ii} + 1} + \sum\limits_{s = 1}^t \left( \widehat{\theta}^{\lr{s}}_i - \sqrt{X^{\lr{s}}_{ii} + 1} + \sqrt{X^{\lr{s-1}}_{ii} + 1} \right) \right)^2\\
														&\stackrel{\mbox{\numcircledtikz{1}}}{\leq} 2\frac{\varepsilon^2}{40000 \lr{\log n}^{20}} (\Tr X^{\lr{0}} + n)  + 2 \E\norm{\sum_{s = 1}^t \left( \widehat{\theta}^{\lr{s}} - \sqrt{\diag \lr{X^{\lr{s}}} + \ones} + \sqrt{\diag \lr{X^{\lr{s-1}}} + \ones} \right) }_2^2 \\
													    &\stackrel{\mbox{\numcircledtikz{2}}}{\leq} \frac{K \varepsilon^2}{10000 \lr{\log n}^{20}}  + 2 \underbrace{\E\norm{\sum_{s = 1}^t \left( \widehat{\theta}^{\lr{s}} - \sqrt{\diag \lr{X^{\lr{s}}} + \ones} + \sqrt{\diag \lr{X^{\lr{s-1}}} + \ones} \right) }_2^2}_{\text{$\mbox{\numcircledtikz{E}}$}},
\end{align*} where $\mbox{\numcircledtikz{1}}$ is by Cauchy-Schwarz inequality, and $\mbox{\numcircledtikz{2}}$ by  Lemma~\ref{lem-newalg-distbnd-opt}. A subtle point here is that even though the very first iterate in the algorithm satisfies a stronger inequality, namely, $\Tr X^{(0)} \leq n$, we \emph{cannot} use this stronger bound because we care about \emph{all} iterations, and this stronger bound doesn't hold later on. We now bound $\mbox{\numcircledtikz{E}}$ below. Note that since the random variable $\widehat{\theta}^{\lr{s}}$ is not entirely unbiased, the term $\mbox{\numcircledtikz{E}}$ is not the variance. Let $\theta^{\lr{s}} \defeq \E \widehat{\theta}^{\lr{s}}$ and $d^{(s)} = \sqrt{\diag \lr{X^{\lr{s}}} + \ones} - \sqrt{\diag \lr{X^{\lr{s-1}}} + \ones}$. Then,
\begin{align*}
	\mbox{\numcircledtikz{E}} &= \E \norm{\sum\limits_{s = 1}^t \lr{\widehat{\theta}^{\lr{s}} - \lr{ \sqrt{\diag \lr{X^{\lr{s}}} + \ones} - \sqrt{\diag \lr{X^{\lr{s-1}}} + \ones}}} }_2^2\\
	&= \E\norm{ \sum\limits_{s = 1}^t \lr{\widehat{\theta}^{\lr{s}} -\theta^{\lr{s}}  + \theta^{\lr{s}} - d^{\lr{s}} }}_2^2\\
	&= \E\sum\limits_{i = 1}^n \lr{\sum\limits_{s = 1}^t \lr{ \widehat{\theta}^{\lr{s}}_i - \theta^{\lr{s}}_{i} }^2 + \sum\limits_{s = 1}^t \lr{ \theta^{\lr{s}}_{i} - d^{\lr{s}}_{i} }^2 + 2\sum_{ s\neq \ell} \lr{ \widehat{\theta}^{\lr{s}}_{i} - \theta^{\lr{s}}_{i} }\lr{\theta^{\lr{\ell}}_{i} - d^{\lr{\ell}}_{i}} }\\
	&= \sum\limits_{s = 1}^t \E\norm{\widehat{\theta}^{\lr{s}} - \theta^{\lr{s}}}_2^2 + \sum\limits_{s = 1}^t \underbrace{ \sum\limits_{i = 1}^n  \lr{ \theta^{\lr{s}}_{i} - d^{\lr{s}}_{i} }^2}_{\text{$\mbox{\numcircledtikz{F}}$}} + 0\\
	&\leq \sum_{ s= 1}^t \lr{ \E \norm{\widehat{\theta}^{\lr{s}}}^2 + \mbox{\numcircledtikz{F}}},
\end{align*}  
where the last step is by the bound on variance by its second moment. Recall that we already have from Inequality~\ref{bias-squared-final}, $\mbox{\numcircledtikz{F}} \leq 400 n K^2 \eta^2 (\sqrt{2} (\varepsilon/n)^{400} + 2\dexp)$. Substitute this into the bound for $\mbox{\numcircledtikz{E}}$ and $\mbox{\numcircledtikz{B}}$, and apply the result of Lemma~\ref{lem-var-bound-on-xupdate} to bound $\E\|\widehat{\theta}^{\lr{s}}\|_2^2$; we choose $t = \Tin = \tfrac{1}{\varepsilon^2}$ and get 
\begin{align*}
	\mbox{\numcircledtikz{B}} &\leq \underbrace{\frac{K \varepsilon^2}{10000 \lr{\log n}^{20}} +  \frac{1}{\varepsilon^2}\lr{\underbrace{19600 \log(n/\varepsilon) K \eta^2 + 147000 K^2 \eta^2 \dexp}_{\text{second-moment bound from Lemma}~\ref{lem-var-bound-on-xupdate}} +  \underbrace{400 n K^2 \eta^2 \lr{\sqrt{2}(\varepsilon/n)^{400} + 2 \delta}}_{\text{squared error in bias}} }}_{\text{\mbox{\numcircledtikz{G}}}}. \numberthis\label[ineq]{stepszvar}
\end{align*} Next, we bound $\mbox{\numcircledtikz{A}}$ using Jensen's inequality, and use Inequality~\ref{stepszvar} in Inequality~\ref{new-alg-rate1} to get 
\begin{align*}
	\E \dnorm{X^{\lr{t}} - \widetilde{X}^{\lr{t}}} &\leq 2 \sqrt{K + n} \sqrt{\mbox{\numcircledtikz{G}}} + \mbox{\numcircledtikz{G}}. \numberthis\label[ineq]{dnorm-bnd}
\end{align*} Note that to bound $\mbox{\numcircledtikz{G}}$, we only need to take care of the second term in Inequality~\ref{stepszvar}, because the first term is already fixed, and the remaining can be fixed by appropriate choices of $\dexp$. We choose the step size to be \[ \eta = \varepsilon^2 \frac{1}{8 \times 10^4 (\log (n/\varepsilon))^{11}}. \numberthis\label{newalg-optstepsz}\] Substituting this in Inequality~\ref{stepszvar} gives \[\mbox{\numcircledtikz{G}} \leq  \frac{K \varepsilon^2 }{10^4 \lr{\log n}^{20}}+  \frac{ K\varepsilon^2}{6 \times 10^5 \lr{\log (n/\varepsilon)}^{21}} +   \frac{ K \varepsilon^2 n \dexp}{2500 \lr{\log (n/\varepsilon)}^{12}}  +  \frac{K \varepsilon^2 n^2 \lr{\sqrt{2} (\varepsilon/n)^{400} + 2\dexp}}{4 \times 10^5 \times \lr{\log (n/\varepsilon)}^{12}}. \] Plugging this back into Lemma~\ref{dnorm-bnd} with the value of $\dexp$ from Definition~\ref{defs-estimator} gives:
\begin{align*}
	\mbox{\numcircledtikz{G}} &\leq  \frac{K \varepsilon^2 }{10^4 \lr{\log n}^{20}}+  \frac{K\varepsilon^2}{ 6 \times 10^5 \lr{\log n}^{21}}  +   \frac{ 2 K \varepsilon^{403}}{\lr{\log (n/\varepsilon)}^{12} n^{389}} + \frac{3 K \varepsilon^{402} }{41 \lr{\log(n/\varepsilon)}^{12} n^{388}} \\
		&\leq K\varepsilon^2 \lr{ \frac{1}{10^4 \lr{\log n}^{20}} + \frac{1}{6 \times 10^5 \lr{\log (n/\varepsilon)}^{21}} + \frac{2\varepsilon^{401}}{\lr{\log(n/\varepsilon)}^{12} n^{389}} + \frac{3 \varepsilon^{402}}{41 n^{388} \lr{\log(n/\varepsilon)}^{12}} }\\
		&\leq K \varepsilon^2 \lr{ \frac{1}{5\times 10^3 \lr{\log n}^{20}} + \frac{6 \varepsilon^{401}}{\lr{\log n}^{20} n^{380}}}\\
		&\leq \frac{K\varepsilon^2}{4999 \lr{\log n}^{20}}
\end{align*} Plugging this back into Inequality~\ref{dnorm-bnd} and using $K = 40 n \lr{\log n}^{10}$ gives $\E \dnorm{X^{\lr{t}} - \widetilde{X}^{\lr{t}}} \leq 1.132 n\varepsilon$. Since Algorithm~\ref{alg-new} only uses the diagonal entries of $\widetilde{X}^{\lr{t}}$ at any iteration $t$, we can assume the off-diagonal entries exactly equal those in $X^{\lr{t}}$. Therefore $\widetilde{X}^{\lr{t}} - X^{\lr{t}}$ is a diagonal matrix. For a diagonal matrix $A$, we can see that $\dnorm{A} = \nnorm{A}$. Therefore, we have $\E \|X^{\lr{t}} - \widetilde{X}^{\lr{t}}\|_{\mathrm{nuc}} \leq  1.132 n \varepsilon$. 
\end{proof}

\subsection{The Expanded Domain Trick for Projection}\label{pf-proj-in-exp-int} 
The goal of this section is two-fold: first, we show that if the trace constraint is inactive, the projection step is simple and requires no trace normalization; second, we prove that the trace constraint remains inactive throughout the run of our algorithm. We remark that this is also the lemma where we choose the optimal number of iterations in the outer loop of Algorithm~\ref{alg-new}.

\begin{restatable}{lemma}{projectionstepnew}\label{proj-new}
Consider the mirror map $\Phi(X) = X \bullet \log X - \Tr X$   over the domain $\{X: X\succeq 0, \Tr X \leq K\}$. Assuming that the trace inequality is never active, we have that $\exp Y  = \argmin_{X \succeq 0, \Tr X \leq  K} \Phi(X) - Y \bullet X$.
\end{restatable}
\begin{proof}
	We wish to solve 
	\begin{equation}
	\begin{aligned}	
	\min X \bullet \log X - \Tr X - X \bullet Y, \mbox{ subject to } X \succeq 0, \Tr X \leq K.
	\end{aligned}
	\end{equation} By diagonalizing $X$ as $X = U\Lambda U^\top$ and $Y$ as $Y = V \Sigma V^\top$, we can rewrite this problem as
	\begin{equation}
	\begin{aligned}
	\min \sum_{i=1}^n \lambda_i \log \lambda_i - \sum_{i=1}^n \lambda_i - \sum_{i=1}^n \lambda_i \widetilde{y}_i, \mbox{ subject to } \lambda_i \geq 0, \sum_{i=1}^n \lambda_i \leq K,
	\end{aligned}
	\end{equation} where $\widetilde{y}_i$ is the $i$'th diagonal entry of the matrix $U^\top Y U$. The Lagrangian is given by $\mathcal{L}(\lambda_i, \nu) = \sum_{i=1}^n \lambda_i \log \lambda_i - \sum_{i=1}^n \lambda_i - \sum_{i=1}^n \lambda_i \widetilde{y}_i + \nu\lr{\sum_{i = 1}^n \lambda_i - K}$. Setting the gradient to zero gives $\nabla_{\Lambda} \mathcal{L} = \ones + \log \lambda^* - \ones - \widetilde{y} + \nu \ones = 0$,  which gives  $\lambda_i^* = \exp(\widetilde{y}_i - \nu)$ for all $i$. Since we assumed that the trace constraint is \emph{not} active, it means, by complementary slackness, $\nu = 0$ (note that this assumption is justified because  we prove it in Lemma~\ref{lem-newalg-distbnd-opt}). This gives $\lambda_i^* = \exp(\widetilde{y}_i)$ which translates to $X^* = \exp(Y)$, as claimed. 
\end{proof} 
 Before we start the second proof, we need the following result.
\begin{lemma}\label{lem-md-errbnd-approx} Fix a norm $\norm{{}\cdot{}}$. Given an $\alpha$-strongly convex mirror map $\Phi: \D \rightarrow \reals$, a convex, $G$-Lipschitz objective $f: \X \rightarrow \reals$, the diameter of $\X \cap \D$ denoted by $D \defeq \sup\limits_{x \in \mathcal{X} \cap \mathcal{D}} \Phi(X) - \inf\limits_{x \in \X \cap \D}\Phi\lr{ x }$, step size $\eta$, and parameter $\delta^{\prime}$ where $\E\norm{x^{\lr{t}} - \widetilde{x}^{\lr{t}}} \leq \delta^{\prime}$, running mirror descent for $T$ iterations gives iterates $\{\widetilde{x}^{\lr{t}}\}_{t = 1}^T$ that satisfy the inequality \[f\left( \frac{1}{T-1}\sum_{t = 1}^{T-1} \widetilde{x}^{\lr{t}} \right) - f\lr{x^*} \leq \frac{\eta G^2 }{2\alpha} + \frac{1}{\eta \lr{T-1}} ( D_{\Phi} (x^*, \widetilde{x}^{\lr{1}}) - D_{\Phi}(x^*, \widetilde{x}^{\lr{T}}) )  + \delta^{\prime} G.\]
\end{lemma}
This can be derived the same way as Theorem $4.2$ in \cite{Bubeck}, by incorporating the error in iterate, just as we did in the proof of Theorem~\ref{thm-almd}.
\newdistbndfromopt*
\begin{proof}\label{proof-newdistbndfromopt} We prove this by induction on the iteration count.

\textbf{Induction Hypothesis.} We assume that for any iteration $t$, the primal iterate is not too far from the optimal point, satisfying $\dnorm{\widetilde{X}^{\lr{t}} - X^*} \leq 38 n \lr{\log n}^{10}$.

\textbf{Base Case.} Since $Y^{(1)} = 0$, the primal iterate $\widetilde{X}^{(1)} = I$. We also know that the optimal point satisfies $\Tr X^* = n$. Therefore, $\dnorm{\widetilde{X}^{\lr{1}} - X^*} \leq 2n \leq 38 n \lr{\log n}^{10}$. The hypothesis is thus true for the base case, $t=1$.

\textbf{Induction.} Suppose that the hypothesis is true for some $ t = t^{\prime}$. We prove that this would make it true for $t = t^{\prime}+1$ as well. Our technique is to first prove a weak bound for $\dnorm{\widetilde{X}^{(t)} - X^*}$ using triangle inequality of norms; then we boost our bound (and obtain the stronger guarantee of the induction hypothesis) by invoking strong convexity of Bregman Divergence. We now show the details. 
\begin{align*} 
	\dnorm{\widetilde{X}^{\lr{t^{\prime}+1}} - X^*} &\leq \dnorm{\widetilde{X}^{\lr{t^{\prime}+1}}  - \widetilde{X}^{\lr{t^{\prime}}}} + \dnorm{ \widetilde{X}^{\lr{t^{\prime}}} - X^*}\\
	&\leq \underbrace{\nnorm{\widetilde{X}^{\lr{t^{\prime}+1}}  - \widetilde{X}^{\lr{t^{\prime}}}}}_{\text{Inequality~\ref{almd-chain-conc}}} + \underbrace{\dnorm{ \widetilde{X}^{\lr{t^{\prime}}} - X^*}}_{\text{induction hypothesis}}.\\
									   &\leq \underbrace{\frac{2 \eta G}{\alpha}}_{\text{\mbox{\numcircledtikz{A}}}}  + 38 n \lr{\log n}^{10}. \numberthis\label[ineq]{dist-iter-opt-step1}
\end{align*} The first step here used the fact that $\dnorm{M} \leq \nnorm{M}$(We can show this by H\"older's Inequality, $\inner{X}{Y} \leq \opnorm{Y}\nnorm{X}$. Select $Y = \diag \lr{\mbox{sgn}\lr{\diag X}}$, that is, $Y$ is a  diagonal matrix with $Y_{ii} = \mbox{sgn}\lr{X_{ii}}$). We can plug in parameters of the mirror map and the step size, as displayed in Table~\ref{TableNewParams}, to obtain: 
\begin{align*}
\mbox{\numcircledtikz{A}}			&= 2 \cdot \frac{\varepsilon^2}{80000 (\log (n/\varepsilon))^{11}} \cdot 2 \cdot 4 (40 n (\log n)^{10}) \leq  \frac{n \varepsilon^2}{125}.
\end{align*} Plugging this back into Equation~\ref{dist-iter-opt-step1} while using $\varepsilon < 1/2$ and $K = 40 n (\log n)^{10}$ gives $\dnorm{\widetilde{X}^{\lr{t^{\prime}+1}} - X^*} \leq  \frac{n \varepsilon^2}{125} + 38 n \lr{\log n}^{10}$, which implies that $\Tr (\widetilde{X}^{(t^{\prime})}) < ( n( \varepsilon^2/125  + 38 (\log n)^{10} ) +n) < 40 n (\log n)^{10} = K$, which says that the trace constraint on the iterates is not active on the first $t^{\prime}$ iterations. 

Since the trace constraint is not active on the first $t^{\prime}$ iterations, the projection step does not require a normalization. This implies that Algorithm~\ref{alg-almd} now is identical to Approximate Mirror Descent with this mirror map and objective. We now recall Lemma~\ref{lem-md-errbnd-approx} for $T = t^{\prime}+1$:\[ f\left( \frac{1}{t^{\prime}}\sum_{t = 1}^{t^{\prime}} \widetilde{X}^{\lr{t}} \right) - f\lr{X^*} \leq \frac{\eta G^2 }{2\alpha} + \frac{1}{\eta t^{\prime}} ( D_{\Phi} (X^*, \widetilde{X}^{\lr{1}}) - D_{\Phi}(X^*, \widetilde{X}^{\lr{t^{\prime}+1}}) )  + \delta^{\prime} G. \]  Multiplying throughout by $\eta t^{\prime}$ and rearranging the terms gives
\begin{align*}
	D_{\Phi}( X^*, \widetilde{X}^{ \left( t^{\prime}+1 \right)} ) &\leq \frac{\eta^2 G^2 t^{\prime}}{2\alpha} + D_{\Phi}(X^*, \widetilde{X}^{\left( 1 \right)}) - \eta t^{\prime} \underbrace{\left( f \left(\frac{1}{t^{\prime}} \sum_{t = 1}^{t^{\prime}} \widetilde{X}^{\left( t \right)} \right) - f\left(X^* \right) \right)}_{\text{positive}} + \eta t^{\prime} \delta^{\prime} G \numberthis\label[ineq]{newalg-distbnd-phiineq} 
\end{align*} Since $\Phi$ is $\alpha$-strongly convex in the nuclear norm, we have $D_{\Phi}(X^*, \widetilde{X}) \geq \frac{\alpha}{2}\|X^* - \widetilde{X}\|_{\mathrm{nuc}}^2$. Since this is at least $\frac{\alpha}{2} \dnorm{X^* - \widetilde{X}}^2$. Chaining this with Inequality~\ref{newalg-distbnd-phiineq} gives
\begin{align*}
	\dnorm{\widetilde{X}^{\lr{t^{\prime} + 1}} - X^*}^2 &\leq \underbrace{\frac{\eta^2 G^2 t^{\prime}}{\alpha^2}}_{\text{$\mbox{\numcircledtikz{B}}$}} + \underbrace{\frac{2 D_{\Phi}(X^*, \widetilde{X}^{\left( 1 \right)}) }{\alpha}}_{\text{$\mbox{\numcircledtikz{C}}$}} + \underbrace{\frac{2}{\alpha}\eta t^{\prime} \delta^{\prime} G}_{\text{$\mbox{\numcircledtikz{D}}$}}, \numberthis\label{norm-sqrd-def-bnd}
\end{align*} We now bound each of the terms on the right-hand side. We remark that this is actually where we choose the appropriate value of $\Tout$.
\begin{align*}
	\mbox{\numcircledtikz{B}} &= \frac{\eta^2 G^2 T_{\mathrm{inner}} T_{\mathrm{outer}}}{\alpha^2} \\
			&= \frac{\varepsilon^4}{64\times 10^8 \lr{\log (n/\varepsilon)}^{22} } \cdot 4  \cdot \frac{1}{\varepsilon^2} \cdot \frac{1}{\varepsilon} 24\times 10^5 \lr{\log (n/\varepsilon)}^{11}\log n \cdot 16 \lr{40 n \lr{\log n}^{10}}^2\\
			&\leq 40\varepsilon n^2 \lr{\log n}^{10} 
\end{align*} To bound the second term $\mbox{\numcircledtikz{C}} = \frac{2D_{\Phi}(\widetilde{X}^{(1)}, X^*)}{\alpha}$, we need to compute $D_{\Phi}(\widetilde{X}^{(1)}, X^*)$. Recall that $\widetilde{X}^{(1)} = I$ by our algorithm. Therefore, $\Phi(\widetilde{X}^{(1)}) = -n$ and $\nabla \Phi(\widetilde{X}^{(1)}) = 0$. Applying H\"older's inequality gives $\Phi(X^*) \leq \Tr X^* \log \opnorm{X^*} \leq n \log n$. Therefore  $D_{\Phi}(X^*, \widetilde{X}^{(1)}) \leq n \log n$. Now we go  back to the quantity we were trying to bound:
\begin{align*}
	\mbox{\numcircledtikz{C}} &\leq 2 \cdot n \log n \cdot 4 (40 n (\log n)^{10}) \leq 320 n^2 \lr{\log n}^{11}.
\end{align*} Finally, the last term is: 
\begin{align*}
	\mbox{\numcircledtikz{D}} &= \frac{2}{\alpha} \eta \Tin \Tout \delta^{\prime} G \leq 2 \cdot 4K \cdot \frac{30 \log n}{\varepsilon} \cdot 1.132 n \varepsilon \cdot 2 = 21735 n^2 \lr{\log n}^{11}
\end{align*} Summing these terms and plugging back into Inequality~\ref{norm-sqrd-def-bnd} gives 
\begin{align*}
	\dnorm{\widetilde{X}^{(t^{\prime} + 1)} - X^*}^2 &\leq n^2 ( 40 \varepsilon (\log n)^{10} + 320 (\log n)^{11} + 21735 (\log n)^{11}). \\
	&< n^2 (0.77  (\log n)^{20} + 17 (\log n)^{20} + 1150 (\log n)^{20} )\\
	&\leq 1168 n^2 \lr{\log n}^{20} \leq 35 n \lr{\log n}^{10},
\end{align*}which completes the induction. Therefore we have $\dnorm{\widetilde{X}^{\lr{t}} - X^*} \leq 38 n \lr{\log n}^{10}$ for all $t$. Since $\Tr X^* = n$, this proves $\Tr \widetilde{X}^{\lr{t}} < 40 n \lr{\log n}^{10} = K$. 
\end{proof} 

\subsection{Error bound}\label{sec-app-errorbound}
Finally, we put together all the parameters derived above to obtain our claimed error bound. 
\begin{restatable}{lemma}{errbnd}\label{lem-err-bnd} Running Algorithm~\ref{alg-new} gives an output for (\ref{final-prob}) that has an error bound of $K\varepsilon$.
\end{restatable} Our algorithm is in the framework of approximate lazy mirror descent, with error bound given by Theorem~\ref{thm-almd}, restated below. \thmalmd* 
\begin{proof} Our proof involves plugging in the values of the parameters (from Table~\ref{TableNewParams}) in the above bound. Since we assume $n\geq 4$, we use $\log n\leq \sqrt{n}$ in one of the calculations below. 
	\begin{align*}
		\frac{D}{T \eta} &= K\varepsilon \frac{\log K}{30 \log n}\leq K \varepsilon \frac{\log 40 + 6 \log n}{30 \log n} \leq 0.29 K \varepsilon. \\
	\frac{2\eta G^2}{\alpha} &= \frac{32 \varepsilon^2 K}{8 \times 10^4 \lr{\log n}^{11} }= \frac{K\varepsilon	}{2500 \lr{\log n}^{11}}\leq 2 \times 10^{-5} K\varepsilon \\
\delta G &= 1.132 n \varepsilon\leq \frac{K \varepsilon}{35 \lr{\log n}^{10}}\leq 11 \times 10^{-4} K\varepsilon 
	\end{align*} Summing these quantities gives the upper bound on the error to be $\varepsilon K$, as claimed.  
\end{proof}

\section{General Technical Results}\label{lin-al-results}
\begin{lemma}\label{lem_exp_bound-gaussian}
Given $a, b \in \reals^n$ , we have that $\E_{\zeta \sim \N(0, I)} \left( (\zeta^T a)^2 (\zeta^T b)^2 \right) \leq 3 \norm{a}_2^2 \norm{b}_2^2$.
\end{lemma}
\begin{proof}
By Cauchy-Schwarz inequality, the functions $f_1$ and $f_2$ satisfy $\E_{\zeta \sim \N(0, I)} ( f_1\lr{\zeta} f_2\lr{\zeta} ) \leq \sqrt{\E_{\zeta} (f_1\lr{\zeta})^2 \E_{\zeta} (f_2\lr{\zeta})^2}$. Choose $f_1\lr{\zeta} = (\zeta^T a)^2$ and $f_2\lr{\zeta} = (\zeta^T b)^2$. Since $\zeta \sim \N(0, I)$ and all the coordinates of $\zeta$ are independent, $\Var(\zeta^T a) = \sum_{i = 1}^n \Var(\zeta_i a_i) = \sum_{i = 1}^n a_i^2 = \norm{a}_2^2$. Therefore $\zeta^T a \sim \N(0, \norm{a}_2^2)$. For $X \sim \N(0, \sigma^2)$, we have $\E X^4 = 3 \sigma^4$. Applying this to $\zeta^T a$ and $\zeta^T b$ proves the desired inequality.
\end{proof}
%no need because we already cite azlo for this
%\begin{lemma}\label{lem-trace-ab2}
%Given $A \succeq 0$, $B \in \symm^n$ and $\alpha \in (0, 1)$, we have $\Tr A^{\alpha} B A^{1-\alpha} B \leq \Tr AB^2$. 
%\end{lemma}
%\begin{proof}Let $A = U \Lambda U^T$ and $Y = U^T B U$. Then
%		\begin{align*}
%			\Tr \left( A^{\alpha} B A^{1-\alpha} B\right) &= \Tr \left( \Lambda^{\alpha} Y \Lambda^{1-\alpha} Y \right)\\
%														&= \sum_{i, j} \Lambda_{ii}^{\alpha} \Lambda^{1-\alpha}_{jj} Y^2_{ij}\\
%														&\leq \sum_{i, j} \left( \alpha \Lambda_{ii} + (1-\alpha) \Lambda_{jj} \right) Y^2_{ij}\\
%														&= \alpha \sum_{i, j} \Lambda_{ii}Y_{ij}^2 + ( 1- \alpha) \sum_j \Lambda_{jj} Y_{ij}^2 = \Tr AB^2 
%		\end{align*}
%\end{proof}

\end{appendices}

\end{document}